\documentclass[a4paper,USenglish,cleveref, autoref, thm-restate,numberwithinsect]{lipics-v2021}

\pdfoutput=1 \hideLIPIcs  

\title{A Dividing Line for Structural Kernelization of Component Order Connectivity via Distance to Bounded Pathwidth}
\titlerunning{A Dividing Line for Kernelization of COC via Distance to Bounded Pathwidth}

\author{Jakob Greilhuber}
{CISPA Helmholtz Center for Information Security, Saarbrücken, Germany}
{jakob.greilhuber@cispa.de}
{https://orcid.org/0009-0001-8796-6400}
{Member of the Saarbrücken Graduate School of Computer Science, Germany.}

\author{Roohani Sharma} {Discrete Mathematics Group, Institute for Basic Science (IBS), Daejeon, South Korea}{roohani@ibs.re.kr}{https://orcid.org/0000-0003-2212-1359}{Supported by the Young Scientist Fellowship (IBS-R029-C1) of the Institute for Basic Science.}

\authorrunning{J. Greilhuber and R. Sharma}

\Copyright{Jakob Greilhuber and Roohani Sharma} 



\nolinenumbers 

\EventEditors{John Q. Open and Joan R. Access}
\EventNoEds{2}
\EventLongTitle{42nd Conference on Very Important Topics (CVIT 2016)}
\EventShortTitle{CVIT 2016}
\EventAcronym{CVIT}
\EventYear{2016}
\EventDate{December 24--27, 2016}
\EventLocation{Little Whinging, United Kingdom}
\EventLogo{}
\SeriesVolume{42}
\ArticleNo{23}

\usepackage{cite}
\usepackage[most]{tcolorbox}

\newcommand{\function}[2]{\textnormal{\textsf{#1}}(#2)}
\newcommand{\Oh}{O}

\newcommand{\defparproblem}[4]{\begin{tcolorbox}[
            enhanced,
            colback=white,
            colframe=black,
            boxrule = 0.5pt,
            coltitle=black,
            title=#1,
            rounded corners,
            attach boxed title to top left={yshift=-10pt, xshift=6pt},
            bottom=-0.5mm,
            boxed title style={
                    interior style={fill=white},
                    frame hidden
                }
        ]
        \begin{tabularx}{12.5cm}{ r X p {0.5cm}}
            Input:     & #2 \\
            Parameter: & #3 \\
            Question:  & #4
        \end{tabularx}
    \end{tcolorbox}
}

\newcommand{\defproblem}[3]{\begin{tcolorbox}[
            enhanced,
            colback=white,
            colframe=black,
            boxrule = 0.5pt,
            coltitle=black,
            title=#1,
            rounded corners,
            attach boxed title to top left={yshift=-10pt, xshift=6pt},
            bottom=-0.5mm,
            boxed title style={
                    interior style={fill=white},
                    frame hidden
                }
        ]
        \begin{tabularx}{12.5cm}{ r X p {0.5cm}}
            Input:    & #2 \\
            Question: & #3
        \end{tabularx}
    \end{tcolorbox}
}

\newcommand{\range}[2][1]{[#1,#2]}
\newcommand{\abs}[1]{{\left| #1 \right|}}

\newcommand{\complexityClass}[1]{\textnormal{\textsf{{#1}}}}
\newcommand{\PClass}{\complexityClass{P}}
\newcommand{\NP}{\complexityClass{NP}}
\newcommand{\FPT}{\complexityClass{FPT}}
\newcommand{\WOne}{\textup{\textnormal{\textsf{W}}[1]}}
\newcommand{\coNPPoly}{\complexityClass{coNP/poly}}

\newcommand{\paramproblem}[2]{\textnormal{#1/#2}}

\newcommand{\problem}[1]{\textnormal{\textsc{#1}}}
\newcommand{\parameter}[1]{\texttt{#1}}

\newcommand{\vc}{\hyperref[problem:vertex_cover]{\problem{Vertex Cover}}}

\newcommand{\vcAbbreviated}{\hyperref[problem:vertex_cover]{\problem{VC}}}

\newcommand{\vcFullAbbreviated}{\hyperref[problem:vertex_cover]{\problem{Vertex Cover} (\problem{VC})}}

\newcommand{\dcoc}[1][d]{\hyperref[problem:dcoc]{$#1$\text{-}\problem{COC}}}
\newcommand{\dcocModToCaterpillars}[1][d]{\hyperref[problem:dcoc_mod_to_caterpillars]{
        \paramproblem{$#1$\text{-}\problem{COC}}{\parameter{pw}-$1$}}}
\newcommand{\dcocModToG}[1][d]{\hyperref[problem:dcoc_mod_to_g]{
        \paramproblem{$#1$\text{-}\problem{COC}}{\parameter{dist}-\parameter{to}-$\mathcal{G}$}}}

\newcommand{\coc}{\hyperref[problem:coc]{\problem{COC}}}

\newcommand{\cocFullAbbreviated}{\hyperref[problem:coc]{\problem{Component Order Connectivity} (\problem{COC})}}

\newcommand{\cocModToCaterpillarsPlusD}{\hyperref[problem:coc_mod_to_caterpillars_plus_d]{\paramproblem{\problem{COC}}{$d$+\parameter{pw}-$1$}}}

\newcommand{\cocModToGPlusD}{\hyperref[problem:coc_mod_to_g_plus_d]{\paramproblem{\problem{COC}}{$d$+\parameter{dist}-\parameter{to}-$\mathcal{G}$}}}

\newcommand{\cocModToCaterpillarsWithoutD}{\hyperref[problem:coc_mod_to_caterpillars_without_d]{\paramproblem{\problem{COC}}{\parameter{pw}-$1$}}}

\newcommand{\cocModToVCAnnotatedPlusK}{\hyperref[problem:coc_mod_to_vc_annotated_plus_k]{\paramproblem{\problem{A-COC}}{$k$+\parameter{vc}}}}

\newcommand{\cocModToVCPlusK}{\hyperref[problem:coc_mod_to_vc_plus_k]{\paramproblem{\problem{COC}}{$k$+\parameter{vc}}}}

\newcommand{\cocModToVCWithoutK}{\hyperref[problem:coc_mod_to_vc_without_k]{\paramproblem{COC}{\parameter{vc}}}}

\newcommand{\cocSolSizePlusD}{\hyperref[problem:coc_sol_size_plus_d]{\paramproblem{\problem{COC}}{$d$+$k$}}}

\newcommand{\MRSS}{\hyperref[problem:mrss]{\problem{UMRSS}}}
\newcommand{\MRSSFullAbbrv}{\hyperref[problem:mrss]{\problem{Unary Multidimensional Relaxed Subset Sum (UMRSS)}}}

\newcommand{\exactSetCover}{\hyperref[problem:exact_set_cover]{\problem{Exact Set Cover}}}

\newcommand{\exactSetCoverUniverse}{\hyperref[problem:exact_set_cover_universe]{\paramproblem{\problem{Exact Set Cover}}{\parameter{U}}}}

\newcommand{\fvs}{\hyperref[problem:fvs]{\problem{Feedback Vertex Set}}}

\newcommand{\boundedDegreeDeletion}{\hyperref[problem:bounded_degree_deletion]{\problem{Bounded Degree Deletion}}}

\newcommand{\fdeletion}{\hyperref[problem:fdeletion]{$\mathcal{F}$-\problem{MinorDeletion}}}

\newcommand{\perfectTSetMatching}{\problem{Perfect $t$-Set Matching}}

\newcommand{\dcocset}{\textnormal{$d$\text{-}\texttt{coc}} set}
\newcommand{\idFunction}{\textnormal{\textsf{id}}^d}
\newcommand{\decFunction}{\textnormal{\textsf{dec}}^d}
\newcommand{\incFunction}{\textnormal{\textsf{inc}}^d}

\newcommand{\spine}[1]{\function{spine}{#1}}
\newcommand{\pendants}[1]{\function{pendants}{#1}}
\newcommand{\opt}[1]{{\textnormal{\textsf{OPT}}}_d(#1)}
\newcommand{\conflicts}[2]{\textnormal{\textsf{conf}}_{d}(#1,#2)}

\newcommand{\partitionSeeingLotsBlocksLeftAndRight}{\tilde{\mathcal{P}}}

\newcommand{\smallPackingConstant}{{(d^3 + 1)}}
\newcommand{\packingConstant}{{c_\textsf{p}}}

\DeclareMathOperator{\fixedPoints}{\textup{\textsf{fp}}} 
\newtheorem{reductionrule}{Reduction Rule}
\crefname{reductionrule}{Reduction Rule}{Reduction Rules}

\crefname{mtheorem}{Main Theorem}{Main Theorems}

\crefname{enumi}{Step}{Steps}
\crefname{claim}{Claim}{Claims}

\begin{document}

\maketitle

\begin{abstract}
    In this work we study a classic generalization of the ubiquitous {\sc Vertex Cover (VC)} problem, called the {\sc Component Order Connectivity (COC)} problem.
In {\sc COC}, given an undirected graph $G$, integers $d \geq 1$ and $k$,
the goal is to determine if there is a set of at most $k$ vertices whose deletion results in a graph where each connected component has at most $d$ vertices.
When $d=1$, this is exactly \textsc{VC}.

This work is inspired by polynomial kernelization results with respect to structural parameters for {\sc VC}. On one hand, Jansen \& Bodlaender [\textit{TOCS} 2013] show that {\sc VC} admits a polynomial kernel when the parameter is the distance to treewidth-$1$ graphs, on the other hand Cygan, Lokshtanov, Pilipczuk, Pilipczuk \& Saurabh [\textit{TOCS} 2014] showed that {\sc VC} does not admit a polynomial kernel when the parameter is distance to treewidth-$2$ graphs.

Greilhuber \& Sharma [\textit{IPEC} 2024] showed that, for any $d \geq 2$, $d$-\textsc{COC} cannot admit a polynomial kernel when the parameter is distance to a forest of pathwidth $2$.
Here, $d$-\textsc{COC} is the same as \textsc{COC} only that $d$ is a fixed constant not part of the input.
We complement this result and show that like for the {\sc VC} problem where distance to treewidth-$1$ graphs versus distance to treewidth-$2$ graphs is the dividing line between structural parameterizations that allow and respectively disallow polynomial kernelization,
for {\sc COC} this dividing line happens between distance to pathwidth-$1$ graphs and distance to pathwidth-$2$ graphs.
The main technical result of this work is that
    {\sc COC} admits a polynomial kernel parameterized by distance to pathwidth-$1$ graphs plus $d$.

The problem  $d$-\textsc{COC} can also be expressed as an \fdeletion{} problem for an appropriate graph family $\mathcal{F}$.
One of the central questions around \fdeletion{} is for which families $\mathcal{F}$ and minor-closed graph classes $\mathcal{G}$ the problem admits a polynomial kernel when parameterized by the distance to $\mathcal{G}$.
Full dichotomies are known when (1) $\mathcal{F} = \{P_2\}$ [Bougeret et al., \textit{SIDMA} 2022], and (2) when $\mathcal{F}$ is a finite collection of biconnected graphs on at least three vertices that contains at least one planar graph [Bougeret et al., \textit{STACS 2026}].
But, as explicitly pointed out by Bougeret et al. [\textit{STACS 2026}], these results do not even capture the $2$-\textsc{COC} problem.
Our result shows that, when $d \geq 2$, the line of tractability for polynomial kernelization of \dcoc{} parameterized by the distance to $\mathcal{G}$ is different from the tractability line for any \fdeletion{} problem that falls into categories (1) or (2).
Thus, with our result, $d$-\textsc{COC} serves as an outlier in the class of \fdeletion{} problems when it comes to understanding the dichotomies for polynomial kernelization when parameterizing by the distance to some minor-closed graph class. \end{abstract}

\section{Introduction}
\label{sec:introduction}
\subparagraph*{Vertex cover.}
The \vcFullAbbreviated{} problem is one of the 21 classical \NP{}-complete problems identified in the seminal work of Karp~\cite{karpReducibilityCombinatorialProblems1972}
and serves as a canonical problem in parameterized complexity, that can be considered the paradigmatic example for developing and testing algorithmic techniques in the field \cite{cyganParameterizedAlgorithms2015,fominKernelizationTheoryParameterized2019}.
In this problem, given an undirected graph $G$ and an integer $k$,
the question is to decide whether there is a set $S \subseteq V(G)$ of size at most $k$ such that each connected component of $G - S$ has size at most one.

A particularly rich line of inquiry for \vcAbbreviated{} has focused on its kernelization complexity
when parameterized by the solution size $k$.
Kernelization (see \cite{fominKernelizationTheoryParameterized2019} or \cite[Chapter 2]{cyganParameterizedAlgorithms2015} for textbook introductions to the topic) is a rigorous paradigm for studying preprocessing for \NP{}-hard problems.
The fundamental question here is to design polynomial-time algorithms that reduce an instance of a parameterized problem to an equivalent instance whose size and parameter are bounded by a function of the parameter of the input.
The gold standard is when the output size is a polynomial function of the parameter, in which case the problem is said to admit a polynomial kernel.

Kernelization for \vcAbbreviated{} with respect to the solution size parameter $k$ has inspired a rich array of techniques
including the Buss rule~\cite{DBLP:journals/siamcomp/BussG93},
crown decompositions~\cite{DBLP:conf/wg/ChorFJ04,DBLP:conf/wg/Fellows03},
the (weighted) expansion lemma~\cite{DBLP:conf/iwpec/KumarL16},
and the LP-based kernel using the celebrated Nemhauser-Trotter theorem~\cite{DBLP:journals/jal/ChenKJ01,DBLP:journals/sigact/Khuller02,DBLP:conf/iwpec/KumarL16,soleimanfallahKernelOrder2kc2011},
resulting in the current best kernel on at most $2k - c \log k$ vertices for the problem, for some constant $c$~\cite{lampisKernelOrder22011}.
These results
have inspired a wide array of generalizations and extensions to broader problem families.

\subparagraph*{Component order connectivity.}
A natural generalization of \vcAbbreviated{} is the \cocFullAbbreviated{} problem: given a graph $G$ and integers $d \geq 1$ and $k$, decide whether there is a set $S \subseteq V(G)$ of size at most $k$ such that deleting $S$ from $G$ results in a graph where every connected component has size at most $d$.
Such a set $S$ is called a \dcocset{} of $G$.
When $d=1$ this is precisely the \vcAbbreviated{} problem.
Apart from being a classic generalization of \vcAbbreviated{},
\coc{} has independently appeared as a natural network vulnerability measure~\cite{grossSurveyComponentOrder2013,kazmierczak2003relationship}.
In this work, we are also interested in the problem family obtained by fixing $d$ to be some constant not part of the input.
Concretely, for any integer $d \geq 1$ the input of the problem \dcoc{} is a graph $G$ and an integer $k$, and one needs to decide whether $G$ has a \dcocset{} of size at most $k$.
The problem \dcoc{} is NP-complete for each $d$ \cite{lewisNodeDeletionProblemHereditary1980}.
Notice that \vc{} is exactly the same problem as \dcoc[1]{}.

Interestingly, many kernelization techniques developed for \vcAbbreviated{}
such as the Buss rule, crown decomposition, (weighted) expansion lemma and even the Nemhauser-Trotter theorem,
also extend to \dcoc{} when $d \geq 2$~\cite{drangeComputationalComplexityVertex2016,DBLP:conf/iwpec/KumarL16,xiaoLinearKernelsSeparating2017a,caselCombiningCrownStructures2024,DBLP:conf/esa/Casel0INZ21, DBLP:journals/algorithmica/BaguleyFNNPZ25,fominKernelizationTheoryParameterized2019}.
This suggests that positive kernelization results for the solution-size parameterization of \vcAbbreviated{} frequently generalize to \dcoc{} parameterized by the solution size.

However, the kernelization landscape dramatically shifts when we consider structural parameters.
One of the most important works in this direction for the \vcAbbreviated{} problem is due to Jansen \& Bodlaender~\cite{jansenVertexCoverKernelization2013b}
who demonstrated the existence of a polynomial kernel for \vcAbbreviated{} when parameterizing by the feedback vertex set number (fvs) of the graph, which is a parameter provably smaller or equal to the size of a minimum solution.
The parameter fvs can equivalently be thought of as the (vertex deletion) distance to treewidth-$1$ graphs.
This positive result is complemented by the result of Cygan, Lokshtanov, Pilipczuk, Pilipczuk, Saurabh~\cite{DBLP:journals/mst/CyganLPPS14} who showed that \vcAbbreviated{} does not admit a polynomial kernel parameterized by the distance to treewidth-$2$ graphs.\footnote{All the lower bounds mentioned here are conditional under reasonable complexity assumptions.}

In stark contrast to \vcAbbreviated{},
the behavior of \dcoc{} under structural parameterizations is more pessimistic.
Greilhuber \& Sharma~\cite{greilhuberComponentOrderConnectivity2024}
show that, when $d \geq 2$, \dcoc{} does not admit a polynomial kernel when parameterizing by fvs,
and even more, no polynomial kernel exists even when parameterizing by the distance to forests of pathwidth $2$.\footnote{The result given in \cite{greilhuberComponentOrderConnectivity2024} is stated for the parameter distance to subdivided comb graphs, which can easily seen to be forests with pathwidth at most two.}
The kernelization lower bound with respect to fvs is also implied by the result of Jansen \& Donkers~\cite{donkersTuringKernelizationDichotomy2021}.

\subparagraph*{Our main result.}
In this work, we identify a clean structural dichotomy for \dcoc{}:
analogous to the distance to treewidth-$1$ graphs versus distance to treewidth-$2$ graphs divide for polynomial kernelization for \vcAbbreviated{},
we establish that {\em distance to pathwidth-$1$ graphs versus distance to  pathwidth-$2$ graphs} form a boundary for polynomial kernelization of \dcoc{}.
Our main result shows that \dcoc{} parameterized by the distance to graphs of pathwidth at most one admits a polynomial kernel.

We denote this problem as \dcocModToCaterpillars{}.
The input is a graph $G$, an integer $k$, and a set $M \subseteq V(G)$ such that $G - M$ has pathwidth at most one.
The task is deciding if $G$ has a \dcocset{} of size at most $k$, and the parameter is $|M|$.
The assumption that the modulator $M$ is given as part of the input is without loss of generality in the sense that one can compute a modulator that is only a constant factor larger than a minimum modulator in polynomial time \cite{guptaLosingTreewidthSeparating2019}.

In our main result we show that even the problem \cocModToCaterpillarsPlusD{} admits a polynomial kernel.
Here, the problem \cocModToCaterpillarsPlusD{} is the same as \dcocModToCaterpillars{}, with the change that $d$ is now part of the input and that the parameter is $|M| + d$.
See \cref{sec:problem_definitions} for fully formal definitions of the problems we consider in this paper.
Giving a kernel for \cocModToCaterpillarsPlusD{} is significantly stronger than just having the result for \dcocModToCaterpillars{}.
For example, a running time of the form $\Oh(|M|^d)$ or kernel size of the form $O(d^d \cdot |M|^4)$ is allowed for \dcocModToCaterpillars{} kernels, but completely unacceptable for \cocModToCaterpillarsPlusD{}.

\begin{restatable}{mtheorem}{mainThmPolyKernelCaterpillarMod}
    \label{main_thm:poly_kernel_caterpillar_mod}
    The problem \cocModToCaterpillarsPlusD{} admits a kernel with $\Oh(d^7 \abs{M}^3 + d^6 \abs{M}^4)$ vertices.
    For each integer $d \geq 1$ the problem \dcocModToCaterpillars{} admits a kernel with $\Oh(\abs{M}^4)$ vertices.
\end{restatable}

\cref{main_thm:poly_kernel_caterpillar_mod} affirmatively answers an open question of Greilhuber \& Sharma~\cite{greilhuberComponentOrderConnectivity2024}, who asked whether \dcoc\ admits a polynomial kernel when the parameter is deletion distance to a linear forest (a forest in which each connected component is a path).

\subparagraph*{Minor deletion problems.}
One of the most important problem families in kernelization (and beyond) is the family of \fdeletion{} problems.
Here, $\mathcal{F}$ is a fixed finite family of graphs.
Given a graph $G$ and integer $k$, the goal of \fdeletion{} is to decide whether there is a set $S \subseteq V(G)$ of size at most $k$ such that $G - S$ is $\mathcal{F}$-minor free, that is, $G-S$ has no graph from $\mathcal{F}$ as a minor.
Numerous fundamental problems including \vc{}, \fvs{}, and \textsc{Planar Vertex Deletion} are captured by this expressive problem family (see \cite{bougeretKernelizationDichotomiesHitting2025} for a recent and detailed introduction to \fdeletion{}).
The problem \dcoc{} is no exception, and also covered by the framework.
Concretely, \vcAbbreviated{} is the same as \fdeletion{} when $\mathcal{F} = \{P_2\}$,
the problem \dcoc[2]{} is obtained by setting $\mathcal{F} = \{P_3\}$, and for $d \geq 3$ setting $\mathcal{F}$ to be the family of all connected graphs on $d+1$ vertices suffices.
With respect to the parameter solution-size polynomial kernels are known to exist when $\mathcal{F}$ contains a planar graph \cite{fominPlanarFdeletionApproximation2012}, and there is an influential line of work dedicated to understanding the kernelization complexity of \fdeletion{} problems with respect to structural parameterizations
\cite{bougeretBridgedepthCharacterizesWhich2022,bougeretKernelizationDichotomiesHitting2025,donkersTuringKernelizationDichotomy2021,dekkerKernelizationFeedbackVertex2024,DBLP:journals/tcs/JansenP20,bougeretHowMuchDoes2019}.
The perhaps most important question that has emerged is the following: For which families $\mathcal{F}$ and minor-closed graph classes $\mathcal{G}$ does \fdeletion{} parameterized by the distance to $\mathcal{G}$ have a polynomial kernel?

This question is fully settled for some problem families $\mathcal{F}$.
By a result of Bougeret, Jansen \& Sau \cite{bougeretBridgedepthCharacterizesWhich2022}, when $\mathcal{F} = \{P_2\}$, that is, for \vcAbbreviated{}, a polynomial kernel exists if and only if $\mathcal{G}$ has bounded bridgedepth.
Here, bridgedepth is a structural parameter that was introduced for this result.
Without going into the formal details, we remark that any forest has bridgedepth at most one.
When $\mathcal{F} = \{K_3\}$, that is for \fvs, a polynomial kernel exists if and only if $\mathcal{G}$ has bounded elimination distance to the class of forests \cite{dekkerKernelizationFeedbackVertex2024}.
Very recently, the result of \cite{dekkerKernelizationFeedbackVertex2024} has been vastly generalized to show that, if $\mathcal{F}$ is a finite set of biconnected graphs on at least three vertices containing at least one planar graph, then a polynomial kernel exists if and only if $\mathcal{G}$ has bounded elimination distance to the class of $\mathcal{F}$-minor free graphs \cite{bougeretKernelizationDichotomiesHitting2025}.
Hence, all of the dichotomies identified so far are centered around $\mathcal{G}$ either having bounded bridgedepth, or bounded elimination distance to the class of $\mathcal{F}$-minor free graphs.

The problem family \dcoc{} is not captured by any of these results.
In fact Bougeret et al. \cite{bougeretKernelizationDichotomiesHitting2025} explicitly mention that the kernelization complexity of \dcoc[2]{} is unclear, and ask whether bridgedepth plays a role for the problem.
As mentioned before, when $d \geq 2$, \dcoc{} does not have a polynomial kernel parameterized by the distance to forests \cite{greilhuberComponentOrderConnectivity2024,donkersTuringKernelizationDichotomy2021},
so $\mathcal{G}$ having bounded bridgedepth does certainly not suffice.
Our \cref{main_thm:poly_kernel_caterpillar_mod} exhibits a polynomial kernel when $\mathcal{G}$ is the class of graphs with pathwidth at most one.
This class $\mathcal{G}$ does not have bounded elimination distance to being $\mathcal{F}$-minor free.
Therefore, when $d \geq 2$, the dichotomy for \dcoc{} must rely on yet-another, so far unidentified parameter.
To the best of our knowledge, \dcoc{} is even the first (non-trivial) \fdeletion{} problem where it is known that the dichotomy relies neither on $\mathcal{G}$ having bounded bridgedepth, nor on $\mathcal{G}$ having bounded elimination distance to the class of $\mathcal{F}$-minor free graphs.

\subparagraph*{Technical challenges.}
On the technical front, let us compare our main result with other polynomial kernels for structural parameterizations of \dcoc{}.
Jansen and Pieterse~\cite{DBLP:journals/tcs/JansenP20} show that, for each $d$ and each constant $\eta$ the problem \dcoc{} has a polynomial kernel parameterized by the size of a treedepth-$\eta$ modulator.
There are other publications that yield polynomial kernels for \dcoc{} \cite{bougeretKernelizationDichotomiesHitting2025,bhyravarapuDifferenceDeterminesDegree2023}, but in fact the existence of polynomial kernels for the parameters they consider is already implied by the aforementioned result of Jansen and Pieterse.
To the best of our knowledge, no further positive results are known.
A key feature of the kernel by Jansen and Pieterse is that
it suffices to be able to bound the number of connected components of the graph obtained after deleting the modulator to obtain the result.
The size of these connected components can then be bounded iteratively in at most $\eta$ further rounds
by exploiting the structure of graphs with treedepth at most $\eta$.
Each round involves moving a bounded number of vertices (the roots of the elimination forest) into the modulator, and then rebounding the number of connected components.

For our main result, the key technical challenge is bounding the size of these connected components.
In particular, long paths are graphs with pathwidth one that do not have bounded treedepth, so the iterative approach used in the previous positive result completely fails.
Our result is the first in \coc{} kernelization that designs reduction rules that bound the size of connected components of the graph after removing the modulator corresponding to the parameter.
We explain this phase and its challenges in more detail in Section~\ref{sec:technical_overview}.

\subsection{Further Results.}
We complement our contribution with a variety of further results about \coc{} kernelization.
While these results are interesting in their own right, they are less technically involved than \cref{main_thm:poly_kernel_caterpillar_mod} and most of them follow by either lifting known results for \vc{} to \dcoc{}, or by relatively standard techniques.

\subparagraph*{On $d$ in the parameter.}
First, we investigate the role the value $d$ plays in more detail.
Recall that in \cref{main_thm:poly_kernel_caterpillar_mod} $d$ is either a constant, or at least part of the parameter.
We show that this is unavoidable.
In \cref{main_thm:vc_kernelization_lower_bound} we exhibit that \coc{} parameterized by the \emph{vertex cover number} of the
graph plus the solution size $k$ cannot have a polynomial kernel, unless \NP{} $\subseteq$ \coNPPoly{} (but the problem is \FPT{} by an easy branching algorithm, see \cref{thm:coc_mod_to_vc_fpt}).
We additionally show in \cref{main_thm:w_one_hardness} that for the parameter
distance to pathwidth one graphs alone \coc{} is even  \WOne{}-hard.
Therefore, there is no hope of obtaining a result similar to \cref{main_thm:poly_kernel_caterpillar_mod} when $d$ is neither a constant, nor part of the parameter.

\subparagraph*{Parameterizing by the distance to bounded degree graphs.}
Next, we utilize the result of \cref{main_thm:poly_kernel_caterpillar_mod} to obtain a polynomial kernel when parameterizing by the distance to the class of graphs with maximum degree at most two (plus $d$ for the problem \coc{}).
We combine this with a \NP{}-hardness result for \dcoc{} on planar graphs with maximum degree three to obtain a kernelization dichotomy.
This result represents a generalization of a result by Majumdar, Raman \& Saurabh~\cite{majumdarPolynomialKernelsVertex2018} for \vc{} to the more general problem \coc{}.
For a graph class $\mathcal{G}$ and integer $d \geq 1$ the problem \dcocModToG{} is \dcoc{} where the input is provided together with a modulator to $\mathcal{G}$, that is a set $M \subseteq V(G)$ such that $G - M \in \mathcal{G}$, and the parameter is $|M|$.
The problem \cocModToGPlusD{} is defined the same way, only that $d$ is part of the input and the parameter is $|M| + d$.

\begin{restatable}{theorem}{mainThmDegreeModDichotomy}
    \label{main_thm:kernel_degree_dichotomy}
    When $\mathcal{G}$ is the class of graphs with maximum degree at most $2$ the problem \cocModToGPlusD{} admits a polynomial-kernel, and, for each $d \geq 1$, \dcocModToG{} admits a polynomial kernel.
    When $\mathcal{G}$ is the class of planar graphs with maximum degree $3$ and $d$ any positive integer, the problem \dcocModToG{} admits no kernel of any size, unless $\PClass = \NP$.
\end{restatable}

\subparagraph*{Bounding the number of connected components.}
Finally, as part of our kernelization routine for \cref{main_thm:poly_kernel_caterpillar_mod} we need to be able to reduce the number of connected components of $G - M$, recall that $M$ is the modulator given in the input.
To do this we perform a careful, but relatively standard lifting of a result for \vc{} due to Hols, Kratsch \& Pieterse~\cite[Theorems 1.3 and 3.12]{holsEliminationDistancesBlocking2022}.
We obtain a preprocessing routine that is usable for a much more general class of parameters in \cref{main_thm:algorithm_to_reduce_components}.
In fact, for any $d$, and any graph class $\mathcal{G}$ that is (1) closed under taking the disjoint-union and (2) such that the problem \dcocModToG{} has a polynomial kernel, our algorithm can be used to reduce the number of connected components of $G - M$ to $|M|^{\Oh(1)}$.
Therefore, our algorithm is the ultimate preprocessing algorithm for reducing the number of connected components.

We remark that the result in \cref{main_thm:algorithm_to_reduce_components} relies on $\mathcal{G}$ having bounded minimal $d$-blocking sets.
A $d$-blocking set of a graph $G$ is a set $X \subseteq V(G)$
such that no minimum \dcocset{} of $G$ contains $X$,
and a $d$-blocking set is minimal if no proper subset of it is a $d$-blocking set \cite{greilhuberComponentOrderConnectivity2024}.
A graph class $\mathcal{G}$ has bounded minimal $d$-blocking set if there is some constant $b$ such that for each graph $G \in \mathcal{G}$ any minimal $d$-blocking set of $G$ has size at most $b$.
The notion of $1$-blocking sets (sometimes implicitly) plays a role in many kernelization results for structural parameterizations of \vc{} \cite{jansenVertexCoverKernelization2013b,fominVertexCoverStructural2016,holsSmallerParametersVertex2017,holsEliminationDistancesBlocking2022,bougeretBridgedepthCharacterizesWhich2022,bougeretHowMuchDoes2019}, and similar blocking set notions are important for \fdeletion{} problems in general, see e.g. \cite{bougeretKernelizationDichotomiesHitting2025}.

\subsection{Outlook.}
In this paper we provide several new positive kernelization results for \coc{}.
In particular, we develop a novel approach of reducing the size of connected components of $G - M$ when $M$ is a modulator to graphs of pathwidth one.
The task of reducing the size of connected components is new in the sense that all previously known positive kernelization results for structural parameterizations of \dcoc{} follow directly from bounding the number of connected components.

Obtaining a full dichotomy for \dcocModToG{} similar to the results which are known for other \fdeletion{} problems is a natural ultimate goal.
However, it seems we do not yet have a sufficiently good understanding of the problem to prove such a result.
We believe that, just like it was the case for \vc{}, it would be fruitful to first collect various positive results before trying to prove a full dichotomy.
Nevertheless, we want to present the following conjecture 

\begin{conjecture}
    For minor-closed graph class $\mathcal{G}$ and all constants $d \geq 1$ the problem \dcocModToG{} admits a polynomial kernel if and only if $\mathcal{G}$ has bounded minimal $d$-blocking sets.
\end{conjecture}

This conjecture seems plausible since,
by a result of Greilhuber \& Sharma \cite{greilhuberComponentOrderConnectivity2024}, the problem \dcoc{} has no polynomial kernel parameterized by the distance to $\mathcal{G}$ when $\mathcal{G}$ is closed-under the disjoint union and does not have bounded minimal $d$-blocking sets.
Moreover, the conjecture is true for the case $d = 1$ since $\mathcal{G}$ has bounded minimal $1$-blocking sets if and only if $\mathcal{G}$ has bounded bridgedepth \cite{bougeretBridgedepthCharacterizesWhich2022}.
And, so far, a bound on minimal $d$-blocking sets of $\mathcal{G}$ was always eventually turned into a polynomial kernel for \dcoc{}.

\section{Technical Overview of Main Theorem 1}
\label{sec:technical_overview}
In this overview we present the main ideas behind the proof of \cref{main_thm:poly_kernel_caterpillar_mod}.
The focus will especially be on how we can reduce the {\em size} of connected components of $G - M$,
as this is the most innovative part of our work.

A graph has pathwidth at most one if and only if it is a caterpillar forest, that is, a forest in which every connected component is a caterpillar \cite[Lemma 2.4]{arnborgMonadicSecondOrder1990}.
Caterpillars are connected graphs such that, after deleting each vertex with degree one, a (possibly empty) path is obtained.
We start by setting some handy terminology.
Then, we define a specific packing of connected subgraphs of $G - M$ that exploits a tight Erd\"os-P\'osa property for connected subgraphs of size at least $d+1$ in trees.
Based on this packing, we identify a subcaterpillar with a large spine (the large central path at which the degree one vertices are attached)
that can be replaced with an equivalent smaller one.
This replacement is based on a novel new insight which we call the {\em essence} of a caterpillar, and which allows one to partition caterpillars into equivalence classes such that the caterpillars in the same equivalence class can be replaced with each other.
An interesting feature about our proof is showing an explicit bound on the size of small representatives of this equivalence class via relations to abstract algebra.
All the above allows one to reduce the length of the spine of a caterpillar of $G - M$.
To bound the number of caterpillars of $G-M$, we exploit that the size of minimal $d$-blocking sets of caterpillars is at most two for all $d \geq 1$.
Finally, we apply a known kernel for the parameter solution-size (plus $d$) to guarantee that the remaining reduced instance is small.

\subparagraph*{Caterpillars and spines.}
For every caterpillar $C$,
we assume for convenience that we are given $\spine{C}$,
which is a fixed nonempty subpath of $C$, such that each vertex of $C$ that is not in $\spine{C}$ has degree one.
Vertices which are not part of the spine are called pendants, and we use $\pendants{C}$ to denote the set of all pendants of $C$.
For each pendant $v$ of $C$, we call the unique neighbor of $v$ the parent of $v$.
Moreover, we assume that the spine has a fixed order, and we refer to one end of the spine as the left end, and to the other end as the right end.
This order remains intact even when taking connected subgraphs $C'$ of the caterpillar $C$, where the spine of the resulting caterpillar $C'$ is the spine of $C$ restricted to the vertices of $C'$.
A spine vertex $v$ is left of spine vertex $w$ if $v$ is closer to the left end than $w$, and left of a pendant $w$ if $v$ is left of the parent of $w$.
A pendant $v$ is left of vertex $w$ if the parent of $v$ is left of $w$.
We similarly define the notions of a vertex/vertex set/subgraph being left/right of another vertex/vertex set/subgraph.

Now, we proceed to how we reduce the size of connected components of $G - M$.
On a very high level, the strategy is to find a large subcaterpillar of $G - M$ which we can safely replace with a smaller caterpillar.
Unfortunately, this is quite a difficult challenge: not only can the concrete structure of a subcaterpillar affect solutions in different ways, but the subcaterpillar can also have many connections to the modulator, which need to be taken into account.

\subparagraph*{Solution-tight packings.}
For figuring out which subcaterpillars we can replace with others, we utilize specific maximum-sized packings of vertex disjoint connected graphs on at least $d+1$ vertices.
These packings are useful because
they exhibit
a tight Erdős-Pósa property
as stated below.
For a graph $H$, we denote the size of a minimum \dcocset{} of $H$ as $\opt{H}$.
The following lemma directly follows from the work of Cockayne et al. \cite[Theorem 1]{cockayneMatchingsTransversalsHypergraphs1979}.

\begin{restatable}[Follows from {Cockayne et al.  \cite[Theorem 1]{cockayneMatchingsTransversalsHypergraphs1979}}]{lemma}{thmPackingCOCDuality}
    \label{thm:packing_coc_duality}
    Let $T$ be a tree, and let $\mathcal{P}$ be a maximum packing of pairwise vertex-disjoint connected subgraphs of $T$ with $d+1$ vertices each.
    Then, $\opt{T} = |\mathcal{P}|$.
    In particular, any optimal \dcocset{}
    of $T$ must select exactly one vertex of each graph of $\mathcal{P}$.
\end{restatable}

Now, we introduce the notion of a solution-tight packing, which builds upon \cref{thm:packing_coc_duality}.
A solution-tight packing is a specific packing of vertex-disjoint connected graphs of size at least $d+1$ each.

\begin{restatable}[Solution-tight Packing]{definition}{defSolutionTightPacking}
    \label{def:solution_tight_packing}
    Let $C$ be a caterpillar with $\spine{C} = v_1,\dots,v_n$, where $v_1$ is the left end and $v_n$ the right end of the spine.
    A set $\mathcal{P}$ of pairwise vertex-disjoint connected subgraphs of $C$ on at least $d+1$ vertices each is a \emph{solution-tight packing} if
    \begin{enumerate}
        \item $\abs{\mathcal{P}} = \opt{C}$, and
        \item for each $p \in \pendants{C}$ if the parent $p'$ of $p$ is part of a graph $H$ of $\mathcal{P}$, then $p$ is also part of $H$, and
        \item there exists an integer $i \in \range[0]{n}$ such that each spine vertex $v_j$ with $j \leq i$ is part of some graph of $\mathcal{P}$, and no spine vertex $v_j$ with $j > i$ is part of a graph of $\mathcal{P}$, and
        \item for each $H \in \mathcal{P}$, the set of size one containing the rightmost spine vertex of $H$ is a \dcocset{} of $H$.
    \end{enumerate}
    Let $F$ be a caterpillar forest, and $\mathcal{C}$ the set of connected components of $F$.
    Then, a packing $\mathcal{P}$ is a solution-tight packing of $F$ if $F = \bigcup_{C' \in \mathcal{C}} \mathcal{P}_{C'}$, where for each $C' \in \mathcal{C}$ we have that $\mathcal{P}_{C'}$ is a solution-tight packing of $C'$.
\end{restatable}

The properties a solution-tight packing must fulfill are sufficient to ensure any caterpillar forest has exactly one unique solution-tight packing.
For any caterpillar forest, it is easy to compute its solution-tight packing using a greedy algorithm (see \cref{thm:compute_sol_tight_packing_poly_time}).
A solution-tight packing of $G - M$ is powerful for kernelization, because we know that any optimum solution must select at least one vertex of each graph of the packing, and that there exists a solution for $G - M$ that selects exactly one vertex of each packed graph.
The further properties of solution-tight packings give us other useful guarantees, for example, if some graph $H \subseteq V(G - M)$ has a solution-tight packing that contains all of its vertices, then there exists an optimal solution for $H$ that selects the rightmost spine vertex of $H$.
We thus have particularly much knowledge about such graphs $H$.

Given this intuition, the subcaterpillars we replace are created by merging several graphs of the solution-tight packing.
Formally we introduce the notion of $\alpha$-merged packings, which are created by merging each $\alpha$ consecutive graphs of the solution-tight packing together.

\begin{restatable}[$\alpha$-merged Packing]{definition}{defMergedPacking}
    \label{def:merged_packing}
    Let $C$ be a caterpillar, $\mathcal{P} = \{H_1,\dots,H_{|\mathcal{P}|}\}$ be the solution-tight packing of $C$, where the graphs of $\mathcal{P}$ are ordered by their appearance along the spine such that graph $H_i$ is left of graph $H_{i+1}$ for all $i \in \range{|\mathcal{P}|-1}$.
    Let $\alpha \geq 1$ be an integer.
    Then, the $\alpha$-merged packing of $\mathcal{P}$ is the packing consisting of $t = \left \lfloor \frac{|\mathcal{P}|}{\alpha} \right \rfloor$ graphs $H_0',\dots,H_{t-1}'$ where $H_i'$ is the subgraph of $C$ induced by the vertex set $\bigcup_{j \in \range[1]{\alpha}} V(H_{i \cdot \alpha + j})$ for each $i \in \range[0]{t-1}$.

    When $F$ is a caterpillar forest, then the $\alpha$-merged packing of $F$ is obtained by taking the union of the $\alpha$-merged packings of each connected component of $F$.
\end{restatable}

We also define the related notion of merged graphs, observe that each graph of the $\alpha$-merged packing is a merged graph.

\begin{restatable}[Merged Graphs]{definition}{defMergedGraphs}
    \label{def:merged_graphs}
    Let $F$ be a caterpillar forest with solution-tight packing $\mathcal{P}$.
    An induced subgraph $C$ of $F$ is a \emph{merged graph} if $V(C) = \bigcup_{C' \in \mathcal{P}'} V(C')$ for some $\mathcal{P}' \subseteq \mathcal{P}$.
\end{restatable}

\subparagraph*{Subcaterpillars with no connections to the modulator.}

Now, we want to elaborate on how we can use these packings to replace subcaterpillars of $G - M$.
To get an initial intuition, let us first consider a graph $C$ of the $\alpha$-merged packing of $G - M$, for some suitable large value of $\alpha$ that we will elaborate on later, such that $C$ has no edges to the modulator.\footnote{We remark that the actual kernelization does not have a specific reduction rule to handle this case. However, elaborating on this case is useful to gain an intuition for how the actual reduction rule we use proceeds.}
Let $\hat C$ be the connected component of $G - M$ of which $C$ is a subgraph.
Since $C$ has no edges to the modulator we do not need to worry about the modulator having a significant effect on the solution within $C$.
However, depending on the exact structure of $C$, the caterpillar $C$ can significantly affect the solution in the rest of $\hat C$.

Since $C$ lies in $\hat C$, we can naturally partition $\hat C$ into three parts, one part is the connected component $L$ of $\hat C - C$ that is left of $C$, the second part is $C$, and the third part $R$ is the connected component of $\hat C -C$ that is right of $C$.
Of course, $L, R$ might not exist, but let us assume that they do for this overview.
Consider some \dcocset{} $S$ of $G$ of size at most $k$.
Then, $S_L = S \cap V(L)$ is a \dcocset{} of $L$, $S_C = S \cap V(C)$ is a \dcocset{} of $C$, and $S_R = S \cap V(R)$ is a \dcocset{} of $R$.
When we replace $C$ with a different graph $C'$, we want to ensure that there is a sufficiently small solution $S_{C'}$ for $C'$ that combines nicely with $S_L$ and $S_R$.
Here, by replacing $C$ with $C'$ we mean deleting $C$, adding $C'$, adding an edge from the rightmost spine vertex of $L$ to the leftmost spine vertex of $C'$, and adding an edge from the leftmost spine vertex of $R$ to the rightmost spine vertex of $C'$.

In $L - S_L$ the rightmost spine vertex of $L$ is in a connected component of size $x$ (set $x = 0$ if the rightmost spine vertex of $L$ is in $S_L$).
Then, these $x$ vertices are not ``taken care of'' by $S_L$, it is the job of $S_C$ to make sure they do not end up in connected components that are too large.
Therefore, $S_C$ is not only a \dcocset{} of $C$, but even a \dcocset{} of the graph $C_x$ we obtain when adding $x$ additional pendants to the leftmost spine vertex of $C$.
If we consider the graph $R$, we find a similar phenomenon: the rightmost spine vertex of $C$ is in some connected component of $C - S_C$ that has size $y$ (set $y = 0$ if the rightmost spine vertex of $C$ is selected).
Then, it is the job of $S_R$ to make sure that these $y$ vertices do not cause any problems.
We can observe that, among all solutions for $C_x$ of size $|S_C|$ we would prefer using the solution that minimizes the respective value of $y$.
If $|S_C| > \opt{C}$, then we are also allowed to choose $S_{C'}$ to be larger than $\opt{C'}$, which is an easy task since $C'$ will have a \dcocset{} of size at most  $\opt{C'} + 1$ that selects the leftmost and rightmost spine vertex of $C'$.
Therefore, we are mostly interested on how small the value $y$ can get when only considering solutions of size $\opt{C}$ for the graph $C_x$.

For exactly this purpose, we introduce the notion of an essence of a caterpillar.
An essence is a function we can compute in polynomial-time which precisely captures how solutions in $C$ can interact with solutions the rest of the graph $\hat C$.
\begin{restatable}[Essence of a Caterpillar]{definition}{defEssence}
    \label{def:caterpillar_essence}
    Let $C$ be a caterpillar such that the solution-tight packing of $C$ contains all vertices of $C$.
    Let the spine of $C$ be $v_1,\dots,v_n$ where $v_1$ is the left end and $v_n$ the right end of the spine.

    For each integer $x \geq 0$, define $C_x$ to be the caterpillar obtained from $C$ by adding $x$ additional pendants to $v_1$ while keeping the spine the same.
    For each integer $x \geq 0$
    and for any \dcocset{} $Y$ of $C_x$ we define $\beta_x(Y)$ to be $0$ if $v_n \in Y$, and otherwise define $\beta_x(Y)$ to be the number of vertices of the connected component of $C_x - Y$ that contains $v_n$.
    For all integers $x \geq 0$, let $\mathcal{D}_x$ be the set of all minimum \dcocset{}s of $C_x$.

    The essence of $C$ (relative to $d$) is the function $\gamma: \range[0]{d+1} \rightarrow \range[0]{d+1}$ where
    \begin{align*}
        \gamma(x) = \begin{cases}
                        \min\{\beta_x(Y) \mid Y \in \mathcal{D}_x\} & \text{if $x \leq d$ and $\opt{C_x} = \opt{C}$}, \\
                        d+1                                         & \text{otherwise}.
                    \end{cases}
    \end{align*}
\end{restatable}

\begin{figure}
    \centering
    \begin{subfigure}{\linewidth}
        \centering
        \includegraphics[scale = 0.6]{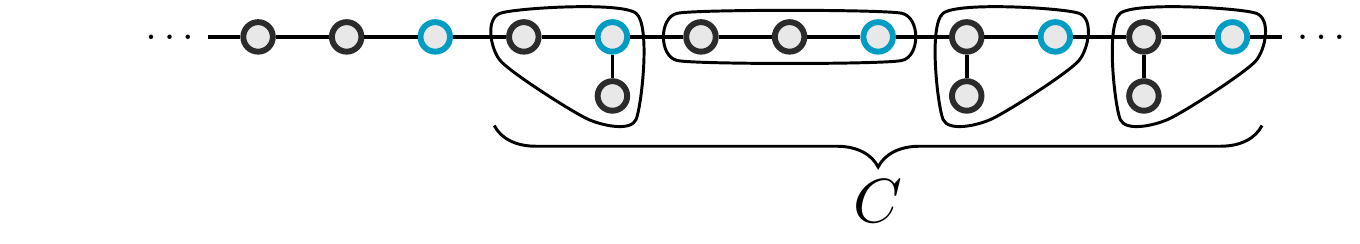}
        \subcaption{The selection before the replacement when $x = 0$.}
    \end{subfigure}
    \begin{subfigure}{\linewidth}
        \centering
        \includegraphics[scale = 0.6]{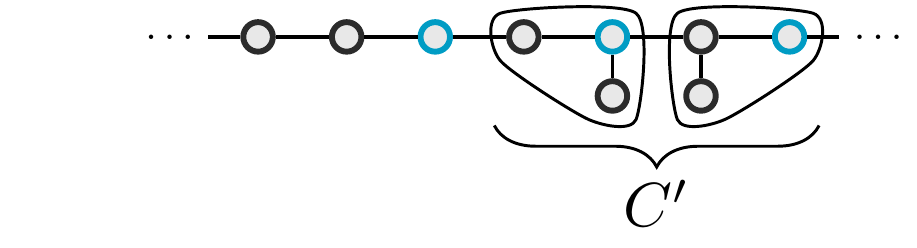}
        \subcaption{The selection after the replacement when $x = 0$.}
    \end{subfigure}

    \begin{subfigure}{\linewidth}
        \centering
        \includegraphics[scale = 0.6]{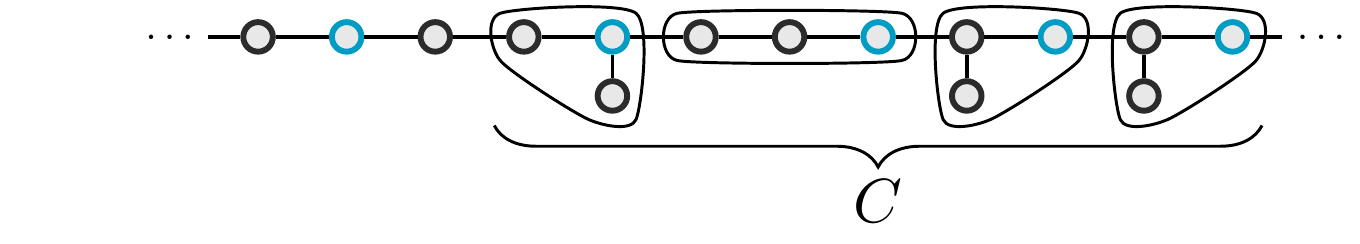}
        \subcaption{The selection before the replacement when $x = 1$.}
    \end{subfigure}
    \begin{subfigure}{\linewidth}
        \centering
        \includegraphics[scale = 0.6]{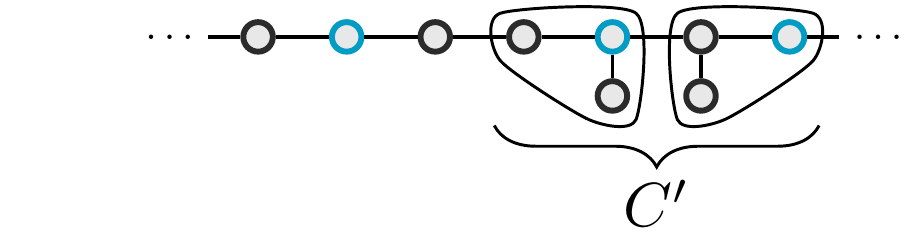}
        \subcaption{The selection after the replacement when $x = 1$.}
    \end{subfigure}

    \begin{subfigure}{\linewidth}
        \centering
        \includegraphics[scale = 0.6]{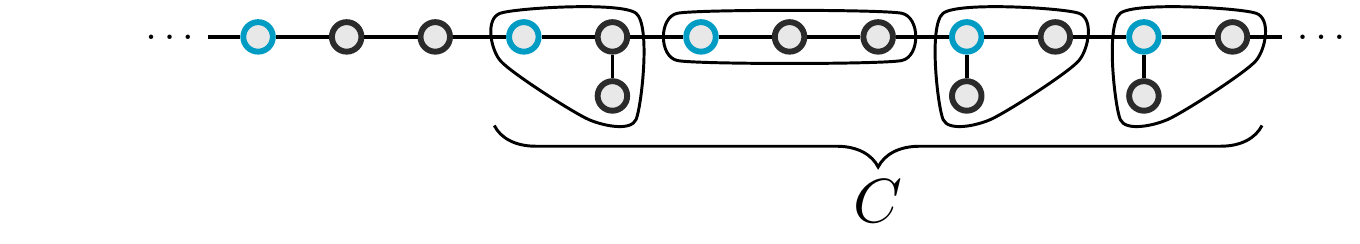}
        \subcaption{The selection before the replacement when $x = 2$.}
    \end{subfigure}
    \begin{subfigure}{\linewidth}
        \centering
        \includegraphics[scale = 0.6]{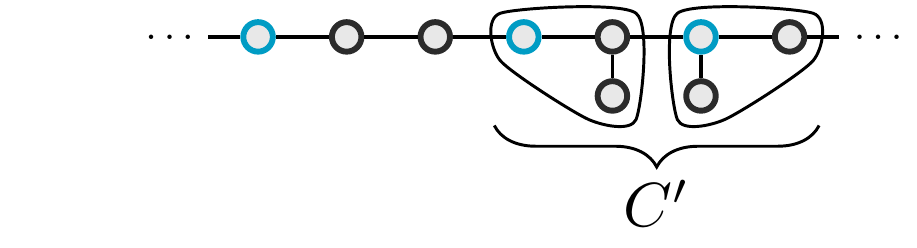}
        \subcaption{The selection after the replacement when $x = 2$.}
    \end{subfigure}
    \caption{The replacement of a subcaterpillar $C$ with a caterpillar $C'$ that has the same essence $\gamma$, in the example $d = 2$. The solution-tight packings of $C$ and $C'$ are highlighted. The essence $\gamma$ fulfills $\gamma(0) = 0$, $\gamma(1) = 0$, $\gamma(2) = 1$, and $\gamma(3) = 3$.
        For any possible number $x$ of vertices of the caterpillar left of $C$ (or $C'$) which are in a connected component with the leftmost spine vertex of $C$ (or $C'$) we draw the corresponding solution within $C$ and $C'$ that also witnesses the values of the essence.
    }
    \label{fig:replacement_same_essence}
\end{figure}

Later, the functional composition of essences will play an important role and is also the reason why we define $d+1$ to be a fixed-point.
The idea behind essences is to replace $C$ with a smaller caterpillar $C'$ that has the same essence as $C$.
\cref{fig:replacement_same_essence} illustrates such a replacement.
The Figure also showcases specific selections depending on how many unselected vertices of the caterpillar to the left of $C$ are in the same connected component with the leftmost spine vertex of $C$ (and $C'$).
The key property of essences is that replacement is safe.

\begin{restatable}{lemma}{thmReplaceWithEssenceSafeNoModConnections}
    \label{thm:replace_with_essence_safe_if_no_mod_connections}
    Let $(G,d,k,M)$ be an instance of \cocModToCaterpillarsPlusD{}, and $P$ be a connected merged subgraph of $G - M$.
    Let $\gamma$ be the essence of $P$, and $P'$ be an arbitrary caterpillar with essence $\gamma$.
    Create $G'$ from $G$ by replacing $P$ with $P'$.
    If there is a \dcocset{} $S$ of $G$ of size at most $k$, then there is an \dcocset{} $S'$ of $G'$ of size at most $k' = k - \opt{P} + \opt{P'}$.
\end{restatable}
The intuition behind the safety of this replacement is that, recalling our split of $\hat C$ into $L$, $C$, $R$, no matter the value $x$ based on the set $S_L$ and the corresponding value $y$ of $S_C$, we have a solution $S_{C'}$ of appropriate size for $C'_x$ such that the rightmost spine vertex of $C' - S_{C'}$ is in a connected component of size at most $y$ (or selected if $y = 0$).
Hence, we find that $(S \setminus S_C) \cup S_{C'}$ is a solution for $G'$ of size at most $|S| - \opt{P} + \opt{P'}$.

We have established that replacing a subcaterpillar $C$ that has no connections to the modulator with a different caterpillar $C'$ with the same essence is a safe operation.
However, it is unclear how we can quickly find such a caterpillar $C'$ which is smaller than $C$, such that a reduction rule using such an operation would actually decrease the size of the instance.

\subparagraph*{Computing caterpillars with specific essences.}
Note that for any fixed $d$ the number of possible essences relative to $d$ is bounded by a function of $d$.
Therefore, for the problem \dcoc{} the kernel could hardcode, for each possible essence $\gamma$, the smallest caterpillar that has essence $\gamma$.
However, our result is for the problem \coc{} where $d$ is not a constant.
So, we must be able to compute small graphs with specific essences on the fly.
Let us now explain how we can do that.
This part of our work has interesting connections to abstract algebra.
Observe that a caterpillar essence $\gamma$ is non-decreasing in the sense that if $x \leq y$ then $\gamma(x) \leq \gamma(y)$, and fulfills $\gamma(0) = 0$ and $\gamma(d+1) = d+1$.
Keeping these properties in mind, we define the essence monoid as follows.\footnote{A monoid is an algebraic structure $(M,\oplus,e)$, where $M$ is a set, $\oplus$ an associative binary operator, $M$ is closed under $\oplus$, and $e \in M$ is an element such that $m \oplus e = m$  and $e \oplus m = m$ for all $m \in M$.}

\begin{restatable}[The Essence Monoid]{definition}{defEssenceMonoid}
    \label{def:essence_monoid}
    Let $d$ be a positive integer, and $\mathcal{M}_d$ be the set of all non-decreasing functions $\gamma: \range[0]{d+1} \rightarrow \range[0]{d+1}$ where $\gamma(0) = 0$, and $\gamma(d+1) = d+1$.
    Let $\idFunction: \range[0]{d+1} \rightarrow \range[0]{d+1}$ be the identity function.
    Then, $(\mathcal{M}_d,\circ, \idFunction)$, where $\circ$ is the function composition operation, is the essence monoid (of $d$).
\end{restatable}

It is not difficult to confirm that this structure is indeed a monoid with neutral element $\idFunction$.
We show that for any $d$, and any function $f$ of the essence monoid of $d$, we can efficiently compute at most $d^3$ simple functions whose function composition is $f$.
We refer to these simple functions as basic functions and define them next.

\begin{restatable}[Basic Functions]{definition}{defBasicFunction}
    Let $d$ be a positive integer.
    For any $i \in \range[1]{d}$ let $\decFunction_i: \range[0]{d+1} \rightarrow \range[0]{d+1}$ be the function defined by
    \begin{align*}
        \decFunction_i(x) = \begin{cases}
                                x-1 & \text{if $x = i$}, \\
                                x   & \text{otherwise}.
                            \end{cases}
    \end{align*}
    Let $\incFunction: \range[0]{d+1} \rightarrow \range[0]{d+1}$ be defined by
    \begin{align*}
        \incFunction(x) = \begin{cases}
                              x + 1 & \text{if $1 \leq x \leq d$}, \\
                              x     & \text{otherwise}.
                          \end{cases}
    \end{align*}
    The set of basic functions (relative to $d$) is the set $\{\idFunction\} \cup \{\decFunction_i \mid i \in \range[1]{d}\} \cup \{\incFunction\}$.
\end{restatable}

We now state the algorithmic result for computing functions as a composition of basic functions next.

\begin{restatable}{lemma}{thmFunctionCompositionAlg}
    \label{thm:function_composition_alg}
    There is an algorithm that takes a non-decreasing function $f: \range[0]{d+1} \rightarrow \range[0]{d+1}$ with $f(0) = 0$ and $f(d+1) = d+1$
    as input and outputs $g_1,\dots,g_\ell$, where $\ell \leq d^3$, such that $f = g_\ell \circ g_{\ell -1} \circ \dots \circ g_1$, and the function $g_i$ is a basic function relative to $d$ for all $i \in \range[1]{\ell}$.
    The running time of the algorithm is polynomial in $d$.
\end{restatable}

Although it seems that we cannot directly employ a result given in the literature in a black-box fashion, the approach we utilize to prove \cref{thm:function_composition_alg} closely follows insights of Higgins \cite{higginsCombinatorialResultsSemigroups1993} and Gomes \& Howie \cite{gomesRanksCertainSemigroups1992} given in the abstract algebra literature.
This means that our paper marks one of these exciting cases where ``old'' results in a seemingly unrelated field (semigroup theory) suddenly lead to state-of-the-art algorithmic results.
From an algebraic point of view it is an interesting question to characterize which functions we can use for a generator of $\mathcal{M}_d$, and what the least number of functions is that need to be composed to generate any function $f \in \mathcal{M}_d$.
Settling this question could improve our kernel size.

We now elaborate further on the links between the essence monoid and caterpillars.
It turns out that not only is each caterpillar essence $\gamma$ in the essence monoid, but also for any function $f$ in the essence monoid there is a caterpillar with essence $f$.
We find the fact that all conceivable essences have a caterpillar with that essence quite interesting, as it shows that the interaction of a subcaterpillar with the caterpillars to the left and right of it can be ``as complicated as imaginable''.
To show this, we prove that the operation of composing functions $f_1$ and $f_2$ of the monoid is virtually the same as connecting the spines of the caterpillars with essence $f_1$ and $f_2$.
Formally, we prove \cref{thm:caterpillar_signatures_function_composition}.

\begin{restatable}{lemma}{thmCaterpillarSignatureFunctionComposition}
    \label{thm:caterpillar_signatures_function_composition}
    Let $C_1$ be a caterpillar with a solution-tight packing that contains all vertices of $C_1$.
    Let the essence of $C_1$ be $\gamma_1$, and let the spine of $C_1$ be $v_1,\dots,v_n$.
    Similarly, let $C_2$ be a caterpillar with a solution-tight packing that contains all vertices of $C_2$, the essence of $C_2$ be $\gamma_2$, and its spine be $w_1,\dots,w_m$.

    Create caterpillar $C$ by taking the disjoint union of $C_1$ and $C_2$ and adding an edge from $v_n$ to $w_1$.
    The spine of $C$ is $v_1,\dots,v_n,w_1\dots,w_m$.
    Then, the essence $\gamma$ of $C$ fulfills $\gamma = \gamma_2 \circ \gamma_1$.
\end{restatable}

Using this fact, we can relatively easily compute a caterpillar with a specific essence by using \cref{thm:function_composition_alg}, assuming that we have caterpillars that have the basic functions as their essence available.
So, we show that we can easily compute such caterpillars next.

\begin{restatable}{lemma}{thmBasicFunctionCaterpillarAlg}
    \label{thm:basic_function_caterpillar_alg}
    There is an algorithm that takes any basic function $f$ (relative to some $d \geq 1$) as input and outputs a caterpillar $C$ with essence $f$ that has a solution-tight packing of size one (that therefore contains only $C$) in time polynomial in $d$.
\end{restatable}

By combining \cref{thm:function_composition_alg,thm:basic_function_caterpillar_alg} and using \cref{thm:caterpillar_signatures_function_composition} we can easily compute a small caterpillar with a specific essence.

\begin{restatable}{lemma}{thmComputeCaterpillarEssence}
    \label{thm:compute_caterpillars_with_essence}
    There is an algorithm that takes a non-decreasing function
    $\gamma: \range[0]{d+1} \rightarrow \range[0]{d+1}$ such that $\gamma(0) = 0, \gamma(d+1) = d+1$, as input, and outputs a caterpillar $C$ that has essence $\gamma$, such that the solution-tight packing $\mathcal{C}$ of $C$ contains all vertices of $C$, and $|\mathcal{C}| \leq d^3$ in time polynomial in $d$.
\end{restatable}

The caterpillars we want to replace should have solution-tight packings of size at least $d^3 + 1$, as we can then use \cref{thm:compute_caterpillars_with_essence} to compute a caterpillar with the same essence that has a strictly smaller solution-tight packing.
Hence, it makes sense to focus on replacing caterpillars of the $(d^3 + 1)$-merged packing of $G - M$.

\subparagraph*{Handling connections to the modulator.}
We have already seen that, if $C$ has no connections to the modulator, then we can quickly replace $C$ with a smaller graph that has the same essence.
But what if the connected components of $G - M$ have edges to the modulator?

Without going into full detail, in our actual reduction rule (\cref{rule:replace_by_essence_graph}) we identify some large subcaterpillar $C$ of $G - M$ that we replace with a smaller caterpillar $C'$ that has the same essence as $C$, such that $C$ has some further useful guarantees that help us deal with the potential edges to $M$.
We also set $k' = k - \opt{C} + \opt{C'}$.
The forward direction of safety of the rule follows directly from \cref{thm:replace_with_essence_safe_if_no_mod_connections}.
For the backward direction, let $S$ be a \dcocset{} of the reduced graph of size at most $k'$.
If $N(C) \cap (M \setminus S) = \emptyset$, the backward direction of safety again follows from \cref{thm:replace_with_essence_safe_if_no_mod_connections}.
Otherwise, the additional guarantees that $C$ provides come into play.
They ensure that $C$ is actually a subgraph of some significantly larger caterpillar $\hat C$.
Furthermore, we have the guarantee that there exists an optimal solution $S_{\hat C}$ for $\hat C$ selecting all vertices of $N(M \setminus S) \cap V(\hat C)$.
Now, let $\hat C'$ be the caterpillar obtained when replacing $C$ of $\hat C$ with $C'$, clearly $\hat C'$ is a subgraph of the reduced instance.
We ensure that each vertex $v \in N(C) \cap (M \setminus S)$ has $d+1$ neighbors in $\hat C$ which are left of $C$, and $d+1$ neighbors in $\hat C$ which are right of $C$.
Therefore, there are also $d+1$ neighbors of $v$ in $\hat C'$ that are left and right of $C'$.
Hence, $S$ must select at least one of these neighbors on each side of $C'$.
These selected vertices effectively cut $\hat C'$ into three parts, one of them, call it $P^{*\prime}$, contains $C'$.
Let $P^*$ be the graph $P^{*\prime}$ where $C'$ is replaced with $C$.
It can then be shown that $(S \setminus V(P^{*\prime})) \cup (V(P^*) \cap S_{\hat C})$ is a \dcocset{} of the input graph of size at most $k$, showing that the rule is safe.

\subparagraph*{Finishing the kernel.}
We can apply the sketched reduction rule to bound the length of the spine of $G - M$ as we can always find suitable graphs $C$ and $\hat C$ if the $\alpha$-merged packing (for a specific value of $\alpha$ that is polynomial in the parameter) of $G - M$ is sufficiently large.
However, $G - M$ could still contain an unbounded number of connected components.
By exploiting that caterpillars have minimal $d$-blocking sets of size at most two for each integer $d$, one can then bound the number of connected component of $G - M$ using relatively standard techniques going back to Jansen \& Bodlaender \cite{jansenVertexCoverKernelization2013b}.
In \cref{thm:remove_components_quickly_d_non_constant} we use a more involved approach based on the expansion lemma that lifts a result of Hols et al. \cite{holsEliminationDistancesBlocking2022} to \coc{}.

After exhaustively applying the previous rules, the size of $\spine{G - M}$, the graph consisting of the spines of each component of $G - M$, is bounded by a polynomial in the parameter.
In particular, the set $S = M \cup V(\spine{G - M})$ has size polynomial in $d + |M|$, and as $G - S$ is an independent set, the set $S$ is a \dcocset{} of $G$.
If $|S| \leq k$ we have a yes-instance, otherwise $k$ is also bounded by a polynomial in the parameter, and we can apply the kernel of Xiao \cite{xiaoLinearKernelsSeparating2017a} for the parameter $k + d$ to finish our kernelization routine.
 
\section{Preliminiaries}
\label{sec:preliminaries}
\subparagraph*{Graphs.}
We use standard notation, but explicitly explain some of it here to avoid ambiguity.
All considered graphs are undirected and simple.
For a graph $G$ and vertex set $X \subseteq V(G)$, $G[X]$ denotes the graph induced by $X$.
For a graph $G$ and a set $X \subseteq V(G)$ we also write $G - X$ for the graph $G[V(G) \setminus X]$.
For $H$ a subgraph of $G$ we also write $G - H$ for the graph $G - V(H)$.
For a (sub)graph $H$ we denote $\abs{H} = \abs{V(H)}$.
The maximum degree of graph $G$ is denoted as $\Delta(G)$.

A \emph{bipartite graph} is a graph $(V,E)$ such that there exists a partition $(A,B)$ of $V$ such that all edges of the graph are between vertices of $A$ and vertices of $B$.
We generally denote such a graph as the triple $(A,B,E)$.

For a vertex $v$ and a graph $G$ we denote the \emph{open neighborhood} of $v$ in $G$ as $N_G(v)$ and the \emph{closed neighborhood} as $N_G[v]$.
For $X \subseteq V(G)$  we define $N_G(X) = \bigcup_{v \in X} N_G(v) \setminus X$, and $N_G[X] = N_G(X) \cup X$.
If $H$ is a subgraph of $G$ we define $N_G(H) = N_G(V(H))$ and $N_G[H] = N_H[V(H)]$.
We may drop the subscript $G$ when it is clear from the context which graph is meant.

Given a graph $G$ and edge $uv \in E(G)$, the operation of contracting edge $uv$ creates the graph $G'$ obtained by deleting $u$ and $v$, and by then adding a fresh vertex $w$ that is made adjacent to each vertex of $N_G(\{u,v\})$.
A graph $H$ is a minor of $G$ if $H$ can be obtained from $G$ by deleting vertices, deleting edges and contracting edges.

\subparagraph*{Graph classes and modulators.}
A graph class $\mathcal{G}$ is a (usually infinite) set of graphs.
Given a graph $G$, a \emph{modulator to $\mathcal{G}$} of $G$ is a set $M \subseteq V(G)$ such that $G - M$ is a graph in $\mathcal{G}$.

\subparagraph*{Caterpillars.}
A \emph{caterpillar} graph $C$ is a tree, such that deleting all vertices of degree one $C$ results in a path.
For a caterpillar $C$, a \emph{spine} of $C$ is a path $P$ that is a subgraph of $C$ such that all vertices of $V(C) \setminus V(P)$ have degree one in $C$.
Observe that any nonempty caterpillar has a (not necessarily unique) nonempty spine.
Furthermore, given a caterpillar graph, it is easy to compute a nonempty spine for it in polynomial-time. 
So, we can assume that any caterpillar we work with comes with such a fixed spine.
That is, for any (nonempty) caterpillar $C$ we have that $\spine{C}$ is some fixed nonempty spine of $C$, and we call $\spine{C}$ \emph{the} spine of $C$.

A forest $F$ is a \emph{caterpillar forest} if each connected component of $F$ is a caterpillar.
Let $F$ be a caterpillar forest and $\mathcal{C}$ be the set of connected components of $F$.
The \emph{spine} $\spine{F}$ of a caterpillar forest $F$ is the graph obtained by taking the disjoint union of the graphs $\{\spine{C} \mid C \in \mathcal{C}\}$.
A vertex of $\spine{F}$ is called a \emph{spine vertex}.
The \emph{pendants} of $F$ is the set $\pendants{F}$ of all vertices of $F$ which are not part of $\spine{F}$.
Each pendant $v \in \pendants{F}$ has degree exactly one, we refer to the sole neighbor of $v$ as the \emph{parent} of $v$.
Note that the parent $p$ of a pendant $v$ is always a spine vertex.\footnote{Otherwise, the caterpillar $C$ containing $v$ would have to be a path on two vertices, and $\spine{C}$ would have to be empty.}
For a spine vertex $v$ of $F$, we define $\pendants{v} = N(v) \cap \pendants{F}$.

For convenience, given a caterpillar $C$, we will also assume that the spine $\spine{C}$ has a fixed order, that is, we treat the spine as the path $v_1,\dots,v_t$, and refer to $v_1$ as the left end of $C$, and to $v_t$ as the right end of $C$.
A spine vertex $v$ is left of (right of) spine vertex $w$ if $v$ is closer to the left end (right end) than $w$, and left of (right of) a pendant $w$ if $v$ is left of (right of) the parent of $w$.
A pendant $v$ is left of (right of) vertex $w$ if the parent of $v$ is left of (right of) $w$.
Similarly, when $W,W'$ are subgraphs of $C$ (or subsets of $V(C)$), we say that $W$ is left of $W'$ if $W$ contains a vertex that is left of all vertices of $W'$, and $W$ is right of $W'$ if $W$ contains a vertex that is right of all vertices of $W'$.
Furthermore, if we consider some connected subgraph $H$ of a caterpillar $C$, then the spine of $H$ is exactly the spine of $C$ restricted to the vertices of $H$, and the order stays intact.

\subparagraph*{Kernelization}

We now introduce the central notions regarding parameterized complexity and kernelization.

\begin{definition}[Parameterized Problems {\cite[Definition 1.1]{fominKernelizationTheoryParameterized2019}}]
    A parameterized problem is a language $L \subseteq \Sigma^* \times \mathbb{N}$, where $\Sigma$ is a fixed, finite alphabet. For an instance $(x,k) \in \Sigma^* \times \mathbb{N}$, $k$ is called the parameter.
\end{definition}

A kernel is a specific preprocessing algorithm for a parameterized problem.

\begin{definition}[Kernelization {\cite[Definition 1.3]{fominKernelizationTheoryParameterized2019}}]
    Let $L \subseteq \Sigma^* \times \mathbb{N}$ be a parameterized problem.
    A kernelization algorithm (kernel) for $L$ is an algorithm that takes an instance $(x,k) \in \Sigma^* \times \mathbb{N}$ as input and
    \begin{itemize}
        \item runs in time polynomial in $\abs{x} + k$, and
        \item outputs an equivalent instance $(x',k')$ with $\abs{x'}, k' \leq h(x)$ for some computable function $h$.
    \end{itemize}
    We also call $(x',k')$ the kernel of instance $(x,k)$, and $h$ the size of the kernel.
    If $h$ is a polynomial function, then the kernel is polynomial.
\end{definition}

In the area of kernelization, one often works with smaller algorithms that can reduce part of an instance, so-called \emph{reduction rules}.
These rules can have specific purposes, for example reducing the number of connected components or reducing the size of individual connected components.
Such a rule is \emph{safe} if it outputs an equivalent instance of the problem.
Of course, if a rule is applied in the kernelization algorithm, it has to run in polynomial-time.

We now introduce the complexity class \FPT{}, and remark that it is well-known that a parameterized problem is in \FPT{} if and only if it has a (not necessarily polynomial) kernel \cite[Theorem 1.4]{fominKernelizationTheoryParameterized2019}.

\begin{definition}[Fixed-Parameter Tractable Problems {\cite[Definition 1.2]{fominKernelizationTheoryParameterized2019}}]
    A parameterized problem $L \subseteq \Sigma^* \times \mathbb{N}$ is in the class \FPT{} (fixed-parameter tractable) if there is an algorithm that takes instance $(x,k) \in \Sigma^* \times \mathbb{N}$ as input, and decides whether $(x,k) \in L$ in time $f(k) \cdot \abs{x}^{\Oh(1)}$ for some computable function $f: \mathbb{N} \rightarrow \mathbb{N}$.
\end{definition}

Problems that are \WOne{}-hard are conjectured to not be \FPT{}.
Similarly, there are results ruling out polynomial kernels for problems under the assumption that $\NP \not \subseteq \coNPPoly$.
In fact, the lower bounds generally rule out \emph{polynomial compressions}.
A polynomial compression is similar to a polynomial-kernel, except that the target problem can be a different (unparameterized) problem.

\begin{definition}[Polynomial Compressions {\cite[Definition 1.5]{fominKernelizationTheoryParameterized2019}}]
    Let $A \subseteq \Sigma^* \times \mathbb{N}$ be a parameterized problem and $B \subseteq \Sigma^*$ an unparameterized problem.
    A polynomial compression from $A$ to $B$ is an algorithm that, given instance $(I,k)$ of problem $A$ outputs an equivalent instance $I'$ of problem $B$ such that $|I'| \leq g(k)$ for some polynomial function $g$.
    Moreover, the algorithm has to run in time polynomial in $|I| + k$.
\end{definition}
Observe that ruling out polynomial compressions (into any target problem) also rules out polynomial kernels.
More background on these notions can be found in the standard textbook on parameterized algorithms by Cygan et al. \cite{cyganParameterizedAlgorithms2015} or the book on kernelization by Fomin et al. \cite{fominKernelizationTheoryParameterized2019}.
For our hardness proofs, we utilize the corresponding reduction machinery.

\begin{definition}[Parameterized Reductions {\cite[Definition 13.1]{cyganParameterizedAlgorithms2015}}]
    A \emph{parameterized reduction} from parameterized problem $A$ to parameterized problem $B$ is an algorithm that takes instances of $A$ as input.
    Given instance $(I,k)$ of problem $\mathcal{A}$, the reduction algorithm produces an equivalent instance $(I',k')$ of $\mathcal{B}$.
    Moreover, $k'$ has to be bounded by a computable function of the input parameter $k$, and the running time of the algorithm has to be bounded by $f(k) \cdot |I|^{\Oh(1)}$ for some computable function $f$.
\end{definition}

If problem $\mathcal{A}$ is \WOne{}-hard, and there is a parameterized reduction from $\mathcal{A}$ to $\mathcal{B}$, then also $\mathcal{B}$ is \WOne{}-hard.
A similar type of reduction can be used to rule out polynomial compressions.

\begin{definition}[Polynomial Parameter Transformation {\cite[Definition 15.14]{cyganParameterizedAlgorithms2015}}]
    A \emph{polynomial parameter transformation} from parameterized problem $\mathcal{A}$ to parameterized problem $\mathcal{B}$ is an algorithm that takes instances of problem $\mathcal{A}$ as input.
    Given an instance $(I,k)$ of $A$ it runs in time polynomial in $|I| + k$, and outputs and equivalent instance $(I',k')$ of problem $\mathcal{B}$.
    Furthermore, $k' \leq g(k)$ for some polynomial function $g$.
\end{definition}
If problem $\mathcal{A}$ does not have a polynomial compression unless $\NP \subseteq \coNPPoly$, and there is a polynomial parameter transformation from $\mathcal{A}$ to $\mathcal{B}$, then also problem $\mathcal{B}$ has no polynomial compression unless $\NP \subseteq \coNPPoly$. 
\section{Reducing the Number of Connected Components}
\label{sec:reducing_connected_components}
In this section, we cover a general procedure to reduce the number of connected components of $G - M$, where $G$ is a graph, and $M$ is a modulator of $G$ to some graph class $\mathcal{G}$ that fulfills certain properties we will cover later.
This procedure is a generalization of the one introduced in
Hols et al. \cite{holsEliminationDistancesBlocking2022} for the \vc{} problem.
More concretely, our \cref{main_thm:algorithm_to_reduce_components} lifts \cite[Theorem 1.3]{holsEliminationDistancesBlocking2022} and \cite[Theorem 3.12]{holsEliminationDistancesBlocking2022} to the \dcoc{} problem, and \cref{thm:remove_components_quickly_d_non_constant} lifts \cite[Theorem 1.3]{holsEliminationDistancesBlocking2022} to \cocModToGPlusD{}.
We take their approach and proof strategy and make them work for \dcoc{} by making some adaptions.
For example, we utilize the expansion lemma instead of matchings, and make other similar changes to handle the difference between \vc{} and \dcoc{}.
It should be mentioned that this results in an approach that is also similar to the strategy used in Bhyravarapu et al. \cite{bhyravarapuDifferenceDeterminesDegree2023} for their kernel of \dcoc{}  parameterized by the size of a \dcoc[c]{} set, for $c \geq d$.

Formally, we consider the parameterized problem \dcocModToG{}, which is \dcoc{}, where the input additionally contains a modulator $M$ to $\mathcal{G}$ of $G$.
The parameter of the problem is $\abs{M}$.
Note that we treat the problem as a promise problem in the sense that
inputs are guaranteed to be of this form.
The issue is that the graph class $\mathcal{G}$ could be a class such that detecting whether a graph belongs to $\mathcal{G}$ is not decidable, or at least not something that can be achieved in polynomial-time.
Hence, it may not be possible to check whether $M$ is actually a modulator to $\mathcal{G}$ quickly enough.
However, if one can detect whether a graph is in the class $\mathcal{G}$ in polynomial-time, which is possible in all our results in which we fix $\mathcal{G}$ to be some explicit graph class, one can also drop this assumption and treat the problem as a regular parameterized problem.
The results we state in this section will also partially apply to the problem \cocModToGPlusD{} where $d$ is not a fixed constant, but we first focus on the problem \dcocModToG{}, and later explain in which case they also extend to the more general problem.

Now, we formally introduce $d$-blocking sets.
The notion of blocking sets for \vc{} was formalized in \cite{holsSmallerParametersVertex2017,holsEliminationDistancesBlocking2022} and \cite{bougeretBridgedepthCharacterizesWhich2022}, and had, albeit not necessarily explicitly, played an important role in many previous works (see \cite{holsEliminationDistancesBlocking2022} for more details).
Recently, a notion of blocking sets also appeared in Bougeret et al. \cite{bougeretKernelizationDichotomiesHitting2024a}.
The notion of blocking sets for \vc{} was generalized to \dcoc{} in the natural way \cite{greilhuberComponentOrderConnectivity2024}.

\begin{definition}[Blocking Sets {\cite[Definition 3]{greilhuberComponentOrderConnectivity2024}}]
    Let $G$ be a graph, and $d \geq 1$ an integer.
    A subset $X$ of $V(G)$ is a $d$-blocking set if there is no minimum \dcocset{} of $G$ that contains $X$ as a subset.
    Moreover, a $d$-blocking set $X$ is minimal if no strict subset of $X$ is a $d$-blocking set.

    Let $\mathcal{G}$ be a graph class.
    We say that $\mathcal{G}$ has minimal $d$-blocking sets of size at most $b$, for some integer $b$, if for any graph $G \in \mathcal{G}$ it holds that any minimal $d$-blocking set of $G$ has size at most $b$.
\end{definition}

For \dcocModToG{} where $\mathcal{G}$ is closed under taking disjoint unions, it is vital that the size of minimal $d$-blocking sets of graphs in $\mathcal{G}$ is bounded by some constant.
Otherwise, no polynomial kernel can exist \cite[Theorem 4]{greilhuberComponentOrderConnectivity2024} under a standard complexity-theoretic assumption.
In this section we show that when $\mathcal{G}$ has minimal $d$-blocking sets of size at most $b$ then, under some mild additional conditions, we can reduce the number of connected components of $G - M$ to a polynomial in $\abs{M}$ whose exponent depends on $b$.

We now define the notion of \emph{chunks}, which is a concept introduced by Jansen and Bodlaender \cite{jansenVertexCoverKernelization2013b} for their kernelization of \vc{} parameterized by the size of a feedback vertex set.
Hols et al. \cite{holsEliminationDistancesBlocking2022} generalized the notion to the problem \dcocModToG[1]{} for any $\mathcal{G}$ with bounded minimal $1$-blocking sets.

\begin{restatable}[Chunks]{definition}{defChunks}
    \label{def:chunks}
    Let $(G,k,M)$ be an instance of \dcocModToG{}, where $\mathcal{G}$ has minimal $d$-blocking sets of size at most $b$.
    The set of chunks of this instance is $\mathcal{X} = \{X \in 2^M \mid 1 \leq |X| \leq b\}$.
\end{restatable}

We are interested in a chunk $X \in \mathcal{X}$ because we can quantify how bad it is for a solution if it contains no vertex of $X$.
In particular, if taking no vertex of $X$ is very bad for the solution size, then we can deduce that any optimal solution has to contain at least one vertex of $X$.
To quantify how bad it is to not take a vertex of a chunk into the solution, we introduce the notion of a conflict, which is again a concept going back to the seminal work of Jansen and Bodlaender \cite{jansenVertexCoverKernelization2013b}.
Recall that for a graph $H$ we denote the size of a minimum \dcocset{} of $H$ as $\opt{H}$.

\begin{restatable}[Conflict]{definition}{defConflicts}
    \label{def:conflict}
    Let $(G,k,M)$ be an instance of \dcocModToG{}, $G'$ a subgraph of $G - M$, and $X \in \mathcal{X}$ a chunk of the instance.
    The conflict of $X$ imposed on $G'$ is $\conflicts{G'}{X} = |N_{G'}(X)| + \opt{G' - N_{G'}(X)} - \opt{G'}$.
\end{restatable}

For a chunk $X$ and graph $G'$, the conflict of $X$ imposed on $G'$ is the difference between the best solution for $G'$ that contains all vertices of $N_{G'}(X)$, and the best solution for $G'$.
For the \vc{} problem whenever no vertex of a chunk $X$ is taken, any solution must take all neighbors of $X$.
Hence, the conflict value directly describes how much more one has to pay in $G'$ if no vertex of $X$ is taken.
For the problem \coc{} it is generally not true that not selecting a vertex of a chunk $X$ means that we have to select all neighbors of $X$.
Nonetheless, our definition is a natural generalization of the definition for \vc{}, and it turns out that since we have to take \emph{almost all} neighbors of a chunk $X$ if we select no vertex of $X$ the stated notion of conflicts is perfectly suited for our problem.

Next, we introduce the expansion lemma, which is the main tool we use to get rid of connected components.
Before stating the lemma itself, we need to define the notion of a $q$-expansion.

\begin{definition}[Expansion {\cite[Section 2.4]{cyganParameterizedAlgorithms2015}}]
    Let $(A,B,E)$ be a bipartite graph.
    For any integer $q \geq 1$, a set $Q \subseteq E(G)$ is a $q$-expansion of $A$ into $B$ if every vertex of $A$ is incident to exactly $q$ edges of $Q$, and exactly $q \cdot |A|$ vertices of $B$ are incident to an edge of $Q$.
    Endpoints of the edges in $Q$ are called saturated (by $Q$).
\end{definition}

Now, we can state the expansion lemma itself.

\begin{lemma}[Expansion Lemma; {Follows from {\cite[Lemma 2.18]{cyganParameterizedAlgorithms2015}}}]
    \label{thm:expansion_lemma}
    Let $(A,B,E)$ be a bipartite graph, and $q \geq 1$ be an integer, such that $\abs{B} \geq q \cdot \abs{A}$ and there are no isolated vertices in $B$.
    Then, there exists a polynomial-time algorithm that finds nonempty sets $X \subseteq A$ and $Y \subseteq B$ such that there is a $q$-expansion of $X$ into $Y$, no vertex in $Y$ has a neighbor outside $X$, and $\abs{Y} = q \cdot \abs{X}$.
\end{lemma}

The stated variant is not exactly the one given in \cite[Lemma 2.18]{cyganParameterizedAlgorithms2015}, as the one stated there does not guarantee that $\abs{Y} = q \cdot \abs{X}$.
Still, by inspecting the proofs of \cite{cyganParameterizedAlgorithms2015}, it is not difficult to see that we can actually find a $q$-expansion from $X$ to $Y$ in polynomial-time if it exists (using algorithms for maximum matching), and then we can obtain the variant of the lemma we state by just removing all vertices of $Y$ which are not saturated by the found $q$-expansion.
The variant we state is more convenient for us since we would otherwise need to differentiate between those vertices of $Y$ which are saturated and those that are not.
Next, we define the bipartite graph, called the conflict graph, on which we want to employ the expansion lemma.

\begin{definition}[Conflict Graph]
    Let $(G,k,M)$ be an instance of \dcocModToG{}, where $\mathcal{G}$ has minimal blocking sets of bounded size, and $\mathcal{X}$ is the set of chunks of the instance.
    Let $\mathcal{C}$ be the set of connected components of $G - M$.
    The conflict graph $H$ of the instance is the bipartite graph $(\mathcal{X},\mathcal{C},E)$ where $E$ contains an edge between a vertex $X \in \mathcal{X}$ and $C \in \mathcal{C}$ if and only if $\conflicts{C}{X} > 0$.
\end{definition}

Now, we provide our reduction rule to remove connected components.

\begin{reductionrule}
    \label{rule:remove_connected_components_expansion_lemma}
    Let $(G,k,M)$ be an instance of \dcocModToG{} where $\mathcal{G}$ has minimal $d$-blocking sets of size at most $b$.
    Let $\mathcal{X}$ be the set of chunks of the instance,
    and let $\mathcal{C}$ be the connected components of $G - M$.
    Set $q = (d-1) \cdot b + 1$, and let $(\mathcal{X},\mathcal{C},E)$ be the conflict graph of the instance.

    Perform the following procedure:
    \begin{enumerate}
        \item \label{step:remove_ccs_init}     Create a copy $(A',B',E')$ of the conflict graph, set $k' = k$, and let $G'$ be a copy of $G$.

        \item  \label{step:remove_ccs_delete_isolated} Let $\mathcal{I}$ be the set of isolated vertices of $B'$ of graph $(A',B',E')$.
              Update $(A',B',E')$ by deleting $\mathcal{I}$ from the graph and set $k' = k' - \sum_{C \in \mathcal{I}} \opt{C}$.
              Also remove all connected components of $\mathcal{I}$ from the graph $G'$.
        \item \label{step:remove_ccs_terminate} If $|B'| < q \cdot |A'|$, then terminate and output the instance $(G',k',M)$.
        \item \label{step:remove_ccs_apply_expansion_lemma} Otherwise, if $|B'| \geq q \cdot |A'|$, apply the expansion lemma on $(A',B',E')$ to find nonempty sets $X \subseteq A'$ and $Y \subseteq B'$ such that there is a $q$-expansion from $X$ to $Y$.
              Update $(A',B',E')$ by deleting all vertices of $Y$ and $X$, and go to \cref{step:remove_ccs_delete_isolated}.
    \end{enumerate}
\end{reductionrule}

Next, we prove that this rule is safe.

\begin{lemma}
    \cref{rule:remove_connected_components_expansion_lemma} is safe.
\end{lemma}
\begin{proof}
    Let $(G,k,M)$ be the input instance, and $(G',k',M)$ be the output instance after applying the reduction rule.
    Moreover, let $(A,B,E)$ be the conflict graph of the input instance.
    Let $\mathcal{I} \subseteq B$ be the set of all connected components of $G - M$ which were deleted by the rule application.
    Observe that only \cref{step:remove_ccs_delete_isolated} of the reduction rule deletes vertices of $G$, so each vertex of $\mathcal{I} $ was an isolated vertex of the copy $(A',B',E')$ of the conflict graph at some point of the rule execution.
    We show both directions of safety separately.

    \proofsubparagraph*{$\Rightarrow$:}
    Assume that the input instance is a yes-instance, and let
    $S$ be a \dcocset{} of $G$ of size at most $k$.
    Then, $S$ must select at least $\opt{C}$ vertices from each $C \in \mathcal{I}$.
    Hence, $S \cap V(G')$ is a solution for the output instance of size at most $k'$.

    \proofsubparagraph*{$\Leftarrow$:}
    Assume the output instance is a yes instance.
    Let $S'$ be a \dcocset{} of $G'$ of size at most $k'$, and set $M' = M \setminus S'$.
    Let $(A',B',E')$ be the modified copy of the conflict graph at the point of the rule application when it terminates, that is, we have that $\abs{A'} < q \cdot \abs{B'}$ and the if-condition of \cref{step:remove_ccs_terminate} evaluates to true.
    Let $\mathcal{X}' = \mathcal{X} \setminus A'$.
    We proceed by showing that we can assume that $S'$ picks a vertex of $X$ for any chunk $X \in \mathcal{X}'$.

    \begin{claim}
        There is a \dcocset{} of $G'$ of size at most $k'$ that intersects each chunk of $\mathcal{X}'$.
    \end{claim}
    \begin{claimproof}
        Toward this goal, we change $S'$ as follows to create a solution $\hat S'$.
        For each chunk $X \in \mathcal{X}'$, where no vertex of $X$ is selected yet by $S'$, select an arbitrary vertex of $X$.
        Now, define $\hat M = M \setminus \hat S'$.

        Let $\mathcal{D}$ be those vertices of $B$ which were deleted by the reduction rule in \cref{step:remove_ccs_apply_expansion_lemma}.
        Observe that each $C \in \mathcal{D}$ is a subgraph of the output graph $G'$.

        We proceed to argue that for each $C \in \mathcal{D}$ there exists some minimum \dcocset{} of $C$ that selects all vertices of $N(\hat M) \cap V(C)$.
        Assume otherwise for the sake of contradiction.
        Then, for some $C \in \mathcal{D}$, there is no minimum \dcocset{} selecting all of $N(\hat M) \cap V(C)$.
        So, $N(\hat M) \cap V(C)$ is a $d$-blocking set for $C$.
        Because minimal $d$-blocking sets of $\mathcal{G}$ have size at most $b$, also minimal $d$-blocking sets of each connected component of a graph of $\mathcal{G}$ must have size at most $b$.
        So, we can find a minimal $d$-blocking set $B$ of $C$ which is a subset of $N(\hat M) \cap V(C)$, and we have $\abs{B} \leq b$.
        Then, by selecting an arbitrary neighbor $w \in \hat M$ of $v$ for each vertex $v \in B$, we obtain a chunk $X \subseteq \hat M$, which is
        a neighbor of $C$ in the conflict graph.
        But, because no vertex of $X$ is selected by $\hat S'$, it must be the case that $X$ is part of $A'$, and hence $X$ has an edge to $C$ in the conflict graph copy at the point when $C$ is removed in \cref{step:remove_ccs_apply_expansion_lemma}.
        However, when $C$ is deleted from the copy of the conflict graph, by the properties of the expansion lemma, all neighbors of $C$ in the conflict graph copy were also deleted, this contradicts that $X$ is part of $A'$.

        So, for each graph $C \in \mathcal{D}$, we remove $V(C)$ from $\hat S'$, and instead select a minimum \dcocset{} of $C$ that selects all vertices of $N(\hat M) \cap V(C)$.
        We call the resulting selection $\hat S''$.

        It is easy to see that $\hat S''$ is a solution of the output instance, because we start from solution $S'$, add more vertices to it, and then only change selections within connected components of $G' - M$ if the selection we change to is both a solution for a connected component and selects all neighbors of $\hat M$ within the component.

        So, we now need to argue that $\hat S''$ has size at most $k'$.
        Toward this, let $\mathcal{X}''$ be the set of chunks $X$ of $\mathcal{X}'$ such that $S'$ selected no vertex of $X$.
        We have that $\abs{\hat S'} \leq k' + \abs{\mathcal{X}''}$, because we add at most one vertex of each chunk of $\mathcal{X}''$ to $S'$ to create $\hat S'$.

        Now, we need to ensure that $\hat S''$ is sufficiently much smaller than $\hat S'$.
        Observe that, by \cref{step:remove_ccs_apply_expansion_lemma}, there is a $q$-expansion from $\mathcal{X}''$ to $\mathcal{D}$.
        Indeed, each chunk $X$ of $\mathcal{X}''$ is, at some point, saturated by a $q$-expansion in \cref{step:remove_ccs_apply_expansion_lemma}, and all endpoints of any $q$-expansion found by the reduction rule are immediately deleted from the copy of the conflict graph afterwards.
        Fix some chunk $X$ of $\mathcal{X}''$, and observe, that $S'$ must select all but at most $b(d-1)$ neighbors of $X$.
        However, since $q = b(d-1) + 1$, this means that there exists a connected component $C$ of $\mathcal{D}$ with $\conflicts{X}{C} > 0$ in which $S'$ selects all vertices of $N(X) \cap V(C)$.
        In particular, this means that $\abs{S' \cap V(C)} > \opt{C}$.
        Using the $q$-expansion, we can injectively map the chunks of $\mathcal{X}''$ to connected components of $\mathcal{D}$ such that $S'$, and therefore also $\hat S'$, had to select a non-optimal solution from each of these components.
        This now yields that $\abs{\hat S''} \leq k'$, because $\hat S''$ selects an optimal solution in each component of $\mathcal{D}$.
    \end{claimproof}

    So, we can assume that $S'$ selects at least one vertex of each $X \in \mathcal{X}'$.
    Now, it is easy to show that there is a minimum \dcocset{} $S_C$ of each $C \in \mathcal{I}$ such that $N(M') \cap V(C) \subseteq S_C$.
    Assume that for some $C \in \mathcal{I}$ there is no such set $S_C$.
    Then, the set $N(M') \cap V(C)$ is a $d$-blocking set of $C$, and there exists a minimal $d$-blocking set $B \subseteq N(M') \cap V(C)$ of size at most $b$.
    Now, for each $v \in B$ let $v'$ be an arbitrary neighbor of $B$ in $M'$, and $X = \{v' \mid v \in B\}$.
    We know that $\abs{X} \leq b$, and hence $X$ is a chunk.

    Moreover, we have $\conflicts{C}{X} > 0$ because $B$ was a $d$-blocking set, and there is an edge between $C$ and $X$ in the conflict graph.
    As $S'$ selects intersects all chunks in $\mathcal{X}'$, but $S'$ selects no vertex of $X$, it must be the case that $X \in A'$.
    This contradicts that $C$ was ever an isolated vertex of a copy of the conflict graph during the reduction rule execution, so $C$ cannot be in $\mathcal{I}$.
    Therefore, a minimum \dcocset{} $S_C$ that selects all neighbors of $M'$ exists for each $C \in \mathcal{C}$.

    Finally, we set $S = S' \cup \bigcup_{C \in \mathcal{I}} S_C$, which is a \dcocset{} of $G$ of size at most $k$.
\end{proof}

Let us also remark that we deliberately allow $k$ to be a negative integer in the problem definition.
Otherwise, the rule above would not necessarily output a valid instance of \dcoc{}.
Note that a negative value of $k$ clearly implies that the instance is a no-instance, and hence one could change the rule to output a constant-sized no-instance instead in the actual kernelization procedure.
Now, we are ready to state our main lemma about reducing the number of connected components.
\begin{lemma}
    \label{thm:algorithm_to_reduce_components_using_rule_1}
    Let $\mathcal{G}$ be a graph class with minimal $d$-blocking sets of size at most $b$.
    \cref{rule:remove_connected_components_expansion_lemma} takes an instance $(G,k,M)$ of \dcocModToG{} as input and outputs an equivalent instance $(G',k')$ of \dcoc{}, where $G'$ is a subgraph of $G$ obtained by deleting connected components of $G - M$, $k' \leq k$, and the number of connected components of $G' - M$ is at most $((d-1) \cdot b + 1) \cdot |M|^b$.
\end{lemma}
\begin{proof}
    Let $(G,k,M)$ be an instance of \dcocModToG{}.
    Let $(G',k',M)$ be the graph output by \cref{rule:remove_connected_components_expansion_lemma}, and observe that $G',k'$ indeed fulfill that $G'$ is obtained from $G$ by deleting connected components of $G - M$, and $k' \leq k$.
    Recall that the integer $q$ is set to $(d-1) \cdot b + 1$.
    Hence, all that is left to argue is that $G' - M$ contains at most $q \cdot |M|^b$ connected components.

    Let $(A',B',E')$ be the copy of the conflict graph when the application of the reduction rule terminates in \cref{step:remove_ccs_terminate}.
    Observe that any connected component of $G' - M'$ that remains in the graph $G'$ is either part of $B'$ (and was thus never deleted from the conflict graph by the rule application), or deleted in \cref{step:remove_ccs_apply_expansion_lemma} of the rule application.
    When the rule terminates, we have that $\abs{B'} < q \cdot \abs{A'}$. 
    Moreover, $A' \subseteq \mathcal{X}$, where $\mathcal{X}$ is the set of chunks of the input instance.
    On the other hand, we also have that the number of components deleted from the conflict graph (which are not deleted from $G$) in \cref{step:remove_ccs_apply_expansion_lemma}  is exactly $\abs{\mathcal{X} \setminus A'} \cdot q$.
    So, the number of remaining connected components is at most $q \cdot \abs{\mathcal{X}}$.
    Observe that $\abs{\mathcal{X}} \leq \abs{M}^b$ to conclude the proof.
\end{proof}

Note that we cannot state that the output instance $(G',k',M)$ is necessarily also an instance of \dcocModToG{}, because it is not clear whether $G' - M$ is in $\mathcal{G}$.
However, when $\mathcal{G}$ is closed under deleting connected components, clearly $G' - M$ is part of $\mathcal{G}$, and hence in this case $(G',k',M)$ is another instance of \dcocModToG{}.

Next, we want to show that, assuming certain conditions, we can apply \cref{rule:remove_connected_components_expansion_lemma} in polynomial-time, which in turn means that \cref{thm:algorithm_to_reduce_components_using_rule_1} can be used in a kernelization routine.
Recall that, similar to Hols et al.~\cite{holsEliminationDistancesBlocking2022}, we defined the graph class $\mathcal{G} + c$ for an integer $c$ as the graph class containing exactly those graphs which have a modulator to $\mathcal{G}$ of size exactly $c$.
Observe that from a graph in $\mathcal{G}$ we can create a graph in $\mathcal{G} + c$ by adding $c$ isolated vertices, and this does not change the size of a minimum \dcocset{}.
Therefore, we have the following.

\begin{observation}
    \label{thm:solve_G_plus_c_poly_solve_G_poly}
    For any integer $d \geq 1$, if \dcoc{} can be solved in polynomial-time on instances of $\mathcal{G} + c$ for some integer $c \geq 0$, then \dcoc{}  can be solved in polynomial-time on instances of $\mathcal{G}$.
\end{observation}

Being able to solve instances of $\mathcal{G} + d$ quickly also helps us to detect blocking sets of graphs of $\mathcal{G}$ quickly.

\begin{lemma}
    \label{thm:compute_blocking_sets_poly_time}
    If we can solve \dcoc{}  in polynomial-time on instances of $\mathcal{G} + d$, then, for any graph $G \in \mathcal{G}$ and any set $X \subseteq V(G)$ we can compute whether $X$ is a $d$-blocking set of $G$ in polynomial-time.
\end{lemma}
\begin{proof}
    Begin by computing $\opt{G}$ in polynomial-time by using \cref{thm:solve_G_plus_c_poly_solve_G_poly} (for example by using binary search to find the size of a minimum \dcocset{} of $G$).
    The set $X$ is a $d$-blocking set of $G$ if and only if $\opt{G-X} + |X| > \opt{G}$.
    Add a clique on $d$ fresh vertices to $G$, call one of the $d$ vertices of the added clique $v$.
    Connect $v$ to each vertex of $X$ with an edge.
    Call the resulting graph $G'$.
    We have that $G' \in \mathcal{G} + d$, hence, we can compute $\opt{G'}$ in polynomial-time.

    Now, observe that if $X$ is not a blocking set, then, there exists some minimum \dcocset{} $S$ of $G$ that selects all vertices of $X$.
    This set $S$ is also a minimum \dcocset{} of $G'$, so we have that $\opt{G'} = \opt{G}$.
    Any \dcocset{} of $G'$ must select a vertex of the fresh clique of size $d$ we added to $G$, or it must select all vertices of $X$.
    Hence, if $X$ is a blocking set, we must have that $\opt{G'} > \opt{G}$.
    This means that we can decide whether $X$ is a blocking set of $G$ in polynomial-time by computing $\opt{G}$ and $\opt{G'}$ and comparing the values.
\end{proof}

We now show a condition under which \cref{rule:remove_connected_components_expansion_lemma} can  be applied in polynomial-time.

\begin{lemma}
    \label{thm:apply_removing_cc_rule_poly_time}
    If $\mathcal{G}$ is a graph class such that
    \begin{itemize}
        \item minimal $d$-blocking sets of graphs of $\mathcal{G}$ have size at most $b$, and
        \item for any $H \in \mathcal{G}$ and set $X \subseteq V(H)$ we can decide in polynomial-time whether $X$ is a $d$-blocking set of $H$,
    \end{itemize}
    then we can apply \cref{rule:remove_connected_components_expansion_lemma} in polynomial-time.
\end{lemma}
\begin{proof}
    Let $(G,k,M)$ be an instance of \dcocModToG{}, and let $\mathcal{C}$ be the set of connected components of $G - M$.
    As the expansion lemma can be applied in polynomial-time, we must only be able to compute the conflict graph of the input instance in polynomial-time, and moreover, we must be able to compute $\opt{C}$ in polynomial-time for connected components $C$ of $G - M$.

    By assumption, we can compute for any set $X \subseteq V(G - M)$ whether $X$ is a $d$-blocking set of $G - M$ in polynomial-time.
    This allows us to compute $\opt{C}$ of any connected component $C$ of $G - M$ in polynomial-time.
    To do that, we will build a minimum \dcocset{} $S$ of $C$ iteratively.
    Initially, iterate over all vertices $v$ of $V(C)$, and check whether $\{v\}$ is a $d$-blocking set of $G - M$.
    If each vertex of $V(C)$ forms a $d$-blocking set, then $\opt{C} = 0$, and we are done.
    Otherwise, let $v \in V(C)$ be an arbitrary vertex such that $\{v\}$ is not a $d$-blocking set of $G - M$, and set $S_1 = \{v\}$.
    Because $S_1$ is not a $d$-blocking set, there is a minimum \dcocset{} of $G - M$ that contains $v$, and thus there is also a minimum \dcocset{} of $C$ that contains $v$.
    We will keep the invariant that set $S_i$ is a subset of a minimum \dcocset{} of $C$ as we proceed.
    To compute $S_{i+1}$ from $S_i$, iterate over all vertices $v$ of $V(C) \setminus S_i$, and check whether $\{v\} \cup S_i$ is a $d$-blocking set of $G - M$.
    If $S_i \cup \{v\}$ is a $d$-blocking set for all $v \in V(C) \setminus S_i$, then, there exists no minimum \dcocset{} of $C$ that is a superset of $S_i$.
    But, since $S_i$ is a subset of a minimum \dcocset{} of $C$, this implies that $\opt{C} = i$.
    Otherwise, we find a vertex $v$ such that $S_i \cup \{v\}$ is not a $d$-blocking set of $G - M$, so we set $S_{i+1} = S_i \cup \{v\}$.
    We directly obtain that $S_{i+1}$ is a subset of a minimum \dcocset{} of $C$ from the fact that $S_{i+1}$ is not a $d$-blocking set of $G - M$, and thus also not a $d$-blocking set of $C$.

    The vertex set of the conflict graph is just $\mathcal{X} \cup \mathcal{C}$, where $\mathcal{X}$ is the set of chunks of the input instance, and both $\mathcal{X}$ and $\mathcal{C}$ have size polynomial in the input size.
    Hence, to compute the conflict graph quickly enough, we must only be able to decide whether there is an edge between some $X \in \mathcal{X}$ and some $C \in \mathcal{C}$ in the conflict graph in polynomial-time.
    This requires us to compute whether $\conflicts{X}{C} > 0$, which can be done in polynomial-time by just checking whether $N_G(X) \cap V(C)$ is a blocking set of $G - M$.
\end{proof}

Using the elaborations about the graph class $\mathcal{G} + d$, and some even simpler observations about hereditary graph classes, we now obtain our main result about a polynomial preprocessing algorithm to reduce the number of connected components of $G - M$.

\begin{restatable}{theorem}{mainTheoremReducingComponents}
    \label{main_thm:algorithm_to_reduce_components}
    Let $d \geq 1$ be an integer, and $\mathcal{G}$ be a graph class such that minimal $d$-blocking sets of $\mathcal{G}$ have size at most $b$ for some constant $b$.
    Furthermore, assume that
    \begin{enumerate}
        \item[\textup{\textsf{(1)}}]\dcoc{} can be solved in polynomial-time on graph class $\mathcal{G} + d$, \emph{or}
        \item[\textup{\textsf{(2)}}] \dcoc{} can be solved in polynomial-time on graph class $\mathcal{G}$ and $\mathcal{G}$ is hereditary.
    \end{enumerate}

    Then, there exists a polynomial-time algorithm that, given an instance $(G,k,M)$ of \dcocModToG{} outputs an equivalent instance $(G',k')$ of \dcoc{}, where
    \begin{itemize}
        \item $G'$ is a subgraph of $G$ obtained by deleting connected components of $G - M$, and
        \item the number of connected components of $G' - M$ is at most $((d-1) \cdot b + 1) \cdot |M|^b$, and
        \item $k' \leq k$.
    \end{itemize}
\end{restatable}
\begin{proof}
    First, consider the case that \dcoc{} can  be solved in polynomial-time on graph class $\mathcal{G} + d$.
    Then, the theorem
    follows directly from \cref{thm:compute_blocking_sets_poly_time,thm:apply_removing_cc_rule_poly_time,thm:algorithm_to_reduce_components_using_rule_1}.

    Otherwise, we have that $\mathcal{G}$ is hereditary, and that we can solve \dcoc{} in polynomial-time on instances of $\mathcal{G}$.
    This also means that we can compute $\opt{H}$ in polynomial-time for any $H \in \mathcal{G}$.
    Then, for any $H \in \mathcal{G}$ and $X \subseteq V(H)$, we can compute whether $X$ is a blocking set of $H$ by just comparing $\opt{H}$ and $|X| + \opt{H - X}$, where we use that $H - X$ is in $\mathcal{G}$ to see that also $\opt{H - X}$ can be computed in polynomial-time.
    The theorem then follows from \cref{thm:apply_removing_cc_rule_poly_time,thm:algorithm_to_reduce_components_using_rule_1}.
\end{proof}
Note that indeed only one of the conditions (1) or (2) must be fulfilled for the algorithm to work, in particular $\mathcal{G}$ must not be hereditary if \dcoc{} is polynomial-time solvable on $\mathcal{G} + d$.

We briefly argue that the preconditions in \cref{main_thm:algorithm_to_reduce_components} are mild using observations already given by Hols et al. \cite{holsEliminationDistancesBlocking2022}.
A parameterized problem cannot have a kernel of any size if it is not \textsf{FPT}.
On the other hand, if the problem \dcocModToG{} is \textsf{FPT}, then there is a polynomial-time algorithm,
that solves \dcoc{} on the graph class $\mathcal{G} + c$ for any fixed $c$.
In this sense, the preconditions of \cref{main_thm:algorithm_to_reduce_components} are extremely mild: for any kernel to exist, we must have polynomial-time algorithms for \dcoc{}  on the graph class $\mathcal{G} + d$, furthermore, there cannot exist a polynomial-kernel if $\mathcal{G}$ does not have bounded minimal $d$-blocking sets when $\mathcal{G}$ is closed under disjoint union \cite{greilhuberComponentOrderConnectivity2024}.
This means that for any possible graph class $\mathcal{G}$ closed under disjoint union such that there is a polynomial-kernel for the problem \dcocModToG{} \cref{main_thm:algorithm_to_reduce_components} can be applied.

The results stated so far were only for the problem \dcoc{}, in which $d$ is assumed to be a constant.
Since \cref{main_thm:poly_kernel_caterpillar_mod} is about the problem \cocModToCaterpillarsPlusD{}, where $d$ is not constant (but part of the parameter), we need to ensure that we can also apply the rule quickly enough for that problem.
Luckily condition (2) of \cref{main_thm:algorithm_to_reduce_components} about hereditary graph classes $\mathcal{G}$ extends nicely to this setting.
Note that we extend definitions for the problem \dcocModToG{} to the problem \cocModToGPlusD{} in the natural way by treating the instance $(G,d,k,M)$ of \cocModToGPlusD{} as an instance $(G,k,M)$ of \dcocModToG{}.

\begin{restatable}{corollary}{corollaryRemoveComponentsQuicklyDNonConstant}
    \label{thm:remove_components_quickly_d_non_constant}
    Let $\mathcal{G}$ be a hereditary graph class and $b$ a constant such that \coc{} can be solved in polynomial-time on instances of $\mathcal{G}$,
    and minimal $d$-blocking sets of $\mathcal{G}$ have size at most $b$ for all $d \geq 1$.
    Then, there exists a polynomial-time algorithm that takes an instance $(G,d,k,M)$ of \cocModToGPlusD{} as input and outputs an equivalent instance $(G',d,k',M)$ of \cocModToGPlusD{} where
    \begin{itemize}
        \item $G'$ is a subgraph of $G$ obtained by deleting connected components of $G - M$, and
        \item the number of connected components of $G' - M$ is at most $((d-1) \cdot b + 1) \cdot |M|^b$, and
        \item $k' \leq k$.
    \end{itemize}
\end{restatable}
\begin{proof}
    The conflict graph of the input instance has polynomial-size.
    Then, even though $d$ is not a constant, the conflict graph of the instance can be still computed in polynomial-time when we can compute $\opt{H}$ in polynomial-time for any $H \in \mathcal{G}$, and when we can compute whether a set $X \subseteq V(H)$ is a blocking set of a graph $H$ of $\mathcal{G}$.

    For any $H \in \mathcal{G}$, we can easily compute $\opt{H}$ in polynomial-time by using our algorithm for \coc{}.
    Furthermore, since $\mathcal{G}$ is hereditary, given a graph $H$ and $X \subseteq V(H)$, we can compute whether $X$ is a blocking set of $H$ by comparing $\opt{H}$ and $|X| + \opt{H - X}$.
    So, we can compute the conflict graph quickly enough.
    The expansion lemma also runs in polynomial-time, and hence, \cref{rule:remove_connected_components_expansion_lemma} runs in polynomial-time.
    Then, the corollary follows directly from \cref{thm:algorithm_to_reduce_components_using_rule_1}, when additionally observing that $G' - M$ is in $\mathcal{G}$, because $\mathcal{G}$ is hereditary and thus closed under the deletion of connected components.
\end{proof} 
\section{Minimal Blocking-Sets of Caterpillars have Size at Most Two}
\label{sec:caterpillars_bounded_blocking_sets}
It is essential for us to show that caterpillar forests, which are exactly the graphs with pathwidth at most one \cite[Lemma 2.4]{arnborgMonadicSecondOrder1990}, have no large minimal $d$-blocking sets.
Concretely, we show that just like in the case of paths (see \cite{greilhuberComponentOrderConnectivity2024}), the size of any minimal $d$-blocking set of a caterpillar forest is at most two.
For this purpose, we begin by restating a simple lemma of \cite{greilhuberComponentOrderConnectivity2024}.

\begin{restatable}[{\cite[Lemma 13]{greilhuberComponentOrderConnectivity2024}}]{lemma}{thmBlockingSetConnectivityProperty}
    \label{thm:blocking_set_connectivity_property}
    Let $G$ be a graph, and $X \subseteq V(G)$ a $d$-blocking set of $G$ for some integer $d \geq 1$.
    If there exists a subset $X' \subseteq X$ such that $G_1$ and $G_2$ are distinct connected components of $G - X'$, $V(G_1) \cap X \not = \emptyset$, and $V(G_2) \cap X \not = \emptyset$, then $X$ is not a minimal $d$-blocking set.
\end{restatable}

Now, we can proceed to the proof of caterpillar forests having bounded minimal $d$-blocking sets.

\begin{restatable}{theorem}{caterpillarsMinimalBlockingSetsBounded}
    \label{thm:caterpillar_bounded_blocking_sets}
    Let $G$ be a caterpillar forest, $d$ a positive integer, and $X \subseteq V(G)$ a minimal $d$-blocking set of $G$.
    Then, $|X| \leq 2$.
\end{restatable}
\begin{proof}
    Toward a contradiction, assume that $X$ is a minimal blocking set of $G$, and that $|X| \geq 3$.
    Then, $X$ has to be part of a single connected component of $G$ (otherwise \cref{thm:blocking_set_connectivity_property} with $X' = \emptyset$ immediately yields a contradiction), and hence we can assume without loss of generality that $G$ is connected.
    Since $G$ is a caterpillar, $G$ contains spine vertices, and pendants.
    We show that the assumption $|X| \geq 3$ leads to a contradiction by considering multiple cases depending on which type of vertices $X$ contains.

    \proofsubparagraph*{Case 1: $X$ contains a pendant $v$ and the parent $p$ of $v$.}
    This implies that there is some minimum \dcocset{} that contains both $v$ and $p$, as $\{v,p\} \subsetneq X$.
    Let $S$ be such a set.
    Then, $S \setminus \{v\}$ would be a smaller set that is evidently still a \dcocset{}, contradicting that $S$ was of minimum size.

    \proofsubparagraph*{Case 2: $X$ contains two pendants $v_1, v_2$ that have the same parent $p$.}
    Again, this would imply that there is some minimum \dcocset{} $S$ that contains $v_1$ and $v_2$.
    Let $S$ be such a set.
    Then the set $(S \setminus \{v_1,v_2\}) \cup \{p\}$ is a \dcocset{} that is strictly smaller than $S$.
    So, also this case is not possible.

    \proofsubparagraph*{Case 3: No two vertices of $X$ are pendants of the same parent, and $X$ does not contain a pendant and the parent of the pendant.}
    Let the spine of $G$ be $v_1,\dots,v_n$.
    Each vertex of $X$ is either a vertex $v_i$ for some $i$, or a pendant of vertex $v_i$ for some $i$.
    Moreover, no two pendants of $X$ have the same parent.
    Hence, we can order the vertices of $X$ from left to right.
    Let $z_\ell$ be the leftmost vertex of $X$, and $z_r$ be the rightmost vertex of $X$.
    Let $v$ be a vertex of $X \setminus \{z_\ell,z_r\}$.
    Moreover, let $S_\ell$ be a minimum \dcocset{} of $G$ that contains all vertices of $X\setminus\{z_r \}$, and $S_r$ be a minimum \dcocset{} of $G$ that contains all vertices of $X \setminus \{z_\ell\}$.
    Importantly, $v \in S_\ell$ and $v \in S_r$.
    We consider two different cases based on whether $v$ is a vertex of the spine or a pendant.

    \proofsubparagraph*{Case 3.1: Vertex $v$ is a spine vertex.}
    In this case, we directly arrive at a contradiction by applying \cref{thm:blocking_set_connectivity_property} with $X' = \{v\}$.

    \proofsubparagraph*{Case 3.2: Vertex $v$ is a pendant.}
    Let $p$ be the parent of $v$, $G_\ell$ the connected component of $G - \{p\}$ that contains $z_\ell$, and $G_r$ the connected component of $G - \{p\}$ that contains $z_r$.
    Observe that  $S_\ell' = (S_\ell \setminus \{v\}) \cup \{p\}$ is also a minimum \dcocset{}.
    Furthermore, $S_\ell'$ has size exactly $\opt{G_\ell} + \opt{G_r} + 1$, and hence also $S_\ell$ has this size and $\opt{G} = \opt{G_\ell} + \opt{G_r} + 1$.
    Then, we must have that $S_\ell \cap V(G_\ell)$ is a \dcocset{} of $G_\ell$ of size $\opt{G_\ell}$, and that $S_\ell \cap V(G_r)$ is a \dcocset{} of $G_r$ of size $\opt{G_r}$.
    For the same reason, we know that $S_r \cap V(G_r)$ is a minumum \dcocset{} of $G_r$ and $S_r \cap V(G_\ell)$ is a minimum \dcocset{} of $G_\ell$.
    Moreover, we know that $S_r$ and $S_\ell$ contain no pendant $v'$ of $p$ apart from $v$: then, it would be strictly better to remove $v$ and $v'$ and instead select $p$.

    Now, we want to inspect solution $S_r$ in slightly more depth.
    This solution contains $v$ and is of minimum size, so it cannot contain $p$.
    So, $p$ is in some connected component of $G - S_r$.
    Let $C^r$ be the connected component of $G - S_r$ that contains $p$.
    Moreover, let $c_\ell^r = |V(C^r) \cap V(G_\ell)|$, and $c_r^r = |V(C^r) \cap V(G_r)|$.
    That is, $c_\ell^r$ tells us how many of the vertices of $C^r$ stem from the graph $G_\ell$, and similarly $c_r^r$ tells us how many stem from the graph $G_r$.
    Also, let $c_p$ be the number of pendants that $p$ has in $G$.
    As argued earlier both $S_\ell$ and $S_r$ cannot contain a pendant of $p$ apart from $v$, therefore $c_p - 1$ pendants of $p$ are part of $C^r$.

    Then, the connected component of $p$ contains exactly $1 + c_p - 1 + c_\ell^r + c_r^r$ vertices, showing that $c_p + c_\ell^r + c_r^r \leq d$ since $S_r$ is a \dcocset{}.
    Similarly, we define $C^\ell$ as the connected component of $G - S_\ell$ that contains $p$, and the quantities $c_\ell^\ell$ and $c_r^\ell$, and we have $c_\ell^\ell + c_r^\ell + c_p \leq d$.

    Toward a contradiction assume $c_\ell^\ell \leq c_\ell^r$.
    Set $\hat S = (S_r \cap V(G_r)) \cup (S_\ell \cap V(G_\ell)) \cup \{v\}$.
    Using our earlier observations about $S_r$ and $S_\ell$, we then know that $|\hat S| = \opt{G}$, and clearly $X \subseteq \hat S$.

    It remains to prove that $\hat S$ is a \dcocset{} of $G$.
    Observe that, as $(S_r \cap V(G_r))$ is a solution for $G_r$, and $(S_r \cap V(G_\ell))$ is a solution for $G_\ell$, the only problem that could occur is that $p$, the vertex connecting $G_\ell$ and $G_r$, is in a connected component that is too large.
    However, in $G - \hat S$ vertex $p$ is in a connected component that has $c_r^r + c_\ell^\ell + c_p$ vertices in total.
    We know that $c_r^r + c_\ell^r + c_p\leq d$, and moreover, by the assumption that $c_\ell^\ell \leq c_\ell^r$ we now obtain $c_r^r + c_\ell^\ell + c_p \leq d$, showing that $\hat S$ is a \dcocset{} of $G$ that selects all vertices of $X$, a contradiction.
    Therefore, $c_\ell^\ell > c_\ell^r$.

    Now, let us build another \dcocset{} $S^*$ by setting $S^* = (S_r \cap V(G_\ell)) \cup (S_\ell \cap V(G_r))$.
    Like in the previous paragraph, we must only show that $p$ is not in a connected component that is too large in $G - S^*$.
    But, $p$ is in a connected component with $c_\ell^r$ vertices of $G_\ell$, and $c_r^\ell$ vertices of $G_r$.
    This means the connected component of $p$ contains $1 + c_p + c_r^\ell + c_\ell^r < 1 + c_p + c_r^\ell + c_\ell^\ell \leq 1 + d$ vertices, so, $S^*$ is indeed a \dcocset{} of $G$.
    However, the size of this solution $S^*$ is $\opt{G_\ell} + \opt{G_r} < \opt{G}$, a contradiction.
\end{proof} 
\section{The Polynomial Kernel}
\label{sec:poly_kernel}
We now proceed to the kernelization algorithm for the \cocModToCaterpillarsPlusD{} problem.
Since we can use \cref{thm:remove_components_quickly_d_non_constant} to reduce the number of connected components, our main remaining challenge is reducing the size of the connected components.

Let us also want remark that we never change the value of $d$ in any reduction rule.
In particular, our kernel is essentially a kernel for the problem \dcoc{} parameterized by the size of a modulator to a pathwidth one graph, where additionally the running time and the kernel size are polynomial in $|M| + d$.

The main goal of this section is reducing the size of connected components of $G - M$ by decreasing the length of the spines of connected components of $G - M$.
Indeed, once we have bounded the length of the spine, we can use a known polynomial kernel for the parameter $k + d$ due to Xiao \cite{xiaoLinearKernelsSeparating2017a} to finish our kernelization routine.
First, we recall the following relation between \dcocset{}s of trees and maximum packings of disjoint connected subgraphs.

\thmPackingCOCDuality*{}

We remark that maximum packings of vertex-disjoint graphs of size $d+1$ are generally useful for \coc{}, as they always provide lower bounds on the solution size, and were, for example, also used by Casel et al. \cite{caselCombiningCrownStructures2024}.
Also observe that for $d = 1$ the set $\mathcal{P}$ is a maximum matching of the graph $T$.
Matchings were also exploited in the \vc{} kernel of Jansen and Bodlaender \cite{jansenVertexCoverKernelization2013b} which inspired our initial interest in these collections.

The duality between \dcoc{}  and the problem of packing vertex-disjoint connected subgraphs is convenient for arguing about solutions.
We now define the notion of solution-tight packings, which are such packings of the same size as an optimal \dcocset{} with some additional convenient properties.

\defSolutionTightPacking*{}

Observe that for any forest $F$ and solution-tight packing $\mathcal{P}$ of $F$ we have that $\opt{F} = \abs{\mathcal{P}}$.
Since we intend to use solution-tight packings directly in our kernel, we must be able to compute them quickly.

\begin{lemma}
    \label{thm:compute_sol_tight_packing_poly_time}
    Let $F$ be a caterpillar forest.
    Then, $F$ has a unique solution-tight packing that can be computed in polynomial-time.
\end{lemma}
\begin{proof}
    We can compute a solution-tight packing for each connected component $C$ of $F$ individually, and then just take the union of these packings.
    So, let $C$ be a caterpillar with spine $v_1,\dots,v_n$.
    We describe a simple recursive procedure to compute a solution-tight packing of $C$.
    If $C$ contains fewer than $d+1$ vertices, simply output the empty set, it is clear that this is the only solution-tight packing of $C$.

    Otherwise, let $i$ be the lowest integer such that the subgraph $L_i = C[\bigcup_{j \in \range{i}} \{v_i\} \cup \pendants{v_i}]$ contains at least $d+1$ vertices.
    Delete $L_i$ from $C$ to obtain the caterpillar $C'$.
    Recurse on $C'$, and let the obtained packing be $\mathcal{P}'$.
    Add $L_i$ to $\mathcal{P}'$, and return the resulting packing $\mathcal{P}$.
    It is not difficult to confirm that this greedy approach outputs a solution-tight packing.

    Let us now show that actually all solution-tight packings of $C$ contain $L_i$.
    For this purpose, let $\hat{\mathcal{P}}$ be an arbitrary solution-tight packing of $C$.
    By property 1 of the definition of solution-tight packings $\abs{\hat{\mathcal{P}}} = \abs{\mathcal{P}} > 0$.
    Let $\hat L$ be the leftmost graph contained in $\hat{\mathcal{P}}$.
    This graph must contain consecutive spine vertices of $C$, including the leftmost spine vertex of $C$ by property 3.
    Similarly, $L_i$ contains the leftmost spine vertex of $C$.
    Moreover, by property 2 both $L_i$ and $\hat L$ contain all pendants of their spine vertices.
    If $L_i$ was not a subgraph of $\hat L$, this would thus contradict our choice of the vertex set of $L_i$.
    However, if $\hat L$ is a proper supergraph of $L_i$, then $\hat L$ contains more spine vertices than $L_i$.
    This means that the graph obtained after deleting the rightmost spine vertex of $\hat L$ is still a (not necessarily proper) supergraph of $L_i$.
    But, then property 4 cannot hold for $\hat L$, because $L_i$ already contains $d+1$ vertices.
    Hence, we must have $L_i = \hat L$.
    Then, by induction, the computed solution-tight packing is the only solution-tight packing of $C$.
\end{proof}

Since the solution-tight packing of a caterpillar forest $F$ (with a fixed order of the spine) is unique by \cref{thm:compute_sol_tight_packing_poly_time}, we can refer to \emph{the} solution-tight packing of $F$.
For our purposes, we sometimes need to work with graphs which are significantly larger than those of a solution-tight packing.
To deal with this situation, while still being able to benefit of the properties of solution-tight packings, we introduce $\alpha$-merged packings.

\defMergedPacking*{}

Intuitively, an $\alpha$-merged packing is just obtained by taking the solution-tight packing $\mathcal{P}$ and merging every $\alpha$ consecutive graphs of $\mathcal{P}$ into a graph of the resulting packing.
Also observe that since we are working on acyclic graphs, any connected subgraph is also an induced subgraph, hence, the $1$-merged packing is exactly the solution-tight packing of a caterpillar or caterpillar forest.
Moreover, since each graph of the solution-tight packing is connected, and the solution-tight packing contains no spine vertex $v$ that is not part of a packed subgraph such that a spine vertex to the right of $v$ is part of a packed subgraph,
also each graph of the $\alpha$-merged packing is connected for any $\alpha$.
Finally, observe that each graph of an $\alpha$-merged packing itself has a solution-tight packing of size $\alpha$ that contains all of its vertices.

Next, we define the notion of merged graphs, observe that each graph of an $\alpha$-merged packing is a merged graph.
\defMergedGraphs*{}

The idea behind our reduction rule is that we identify some part $P$ of a caterpillar of $G - M$ for which we know that a good solution, that picks all neighbors of vertices of the modulator which are themselves not selected, exists.
Then, we want to replace $P$ with a smaller graph.
However, even if we identify such a part $P$ of a caterpillar, it is unclear with which graph we can replace it, as the interaction of $P$ with the rest of the caterpillar can be complex.
To capture the interaction of a part with the rest of the caterpillar we define the notion of an essence of a caterpillar.
An essence of a caterpillar is an object that can be computed in polynomial time for any caterpillar of the correct form.
The crucial idea is that the essence describes how the caterpillar can interact with other caterpillars that are attached to it from the left and right.

\defEssence*{}

Before getting into more details, we provide some basic properties of essences.
By a function $f: \range[0]{d+1} \rightarrow \range[0]{d+1}$ being non-decreasing we mean that $f(x) \geq f(y)$ for all $x,y \in \range[0]{d+1}$ with $x < y$.

\begin{lemma}
    \label{thm:basic_essence_properties}
    Let $C$ be a caterpillar such that the solution-tight packing of $C$ contains all vertices of $C$, and let $\gamma$ be the essence of $C$.
    Then, $\gamma$ is a non-decreasing function fulfilling $\gamma(0) = 0$ and $\gamma(d+1) = d+1$.
    Moreover, the essence $\gamma$ of $C$ can be computed in polynomial-time.
\end{lemma}
\begin{proof}
    By the definition of the essence, we have $\gamma(d+1) = d+1$.
    Regarding $\gamma(0)$, observe that $C_0 = C$, and there is a minimum \dcocset{} of $C$ that selects its rightmost spine vertex, so $\gamma(0)$ is indeed $0$.
    Moreover, as $C_x$ is a subgraph of $C_y$ for any $y > x$, the function $\gamma$ must be non-decreasing.

    Finally, we need to argue that we can compute $\gamma(x)$ in polynomial-time for any $x \in \range[0]{d}$.
    To achieve this, we can create $C_x$ in polynomial-time, and compute a solution-tight packing of $C_x$ in polynomial-time.
    If this packing is larger than $\opt{C}$, then $\gamma(x) = d+1$.
    Otherwise, we can build a solution $S$ for $C_x$ by selecting the rightmost spine vertex of each packed subgraph of the solution-tight packing of $C_x$.
    It is not difficult to confirm that this is the solution that minimizes the value $\beta_x$ of the definition of the essence among all optimal solutions.
    Hence, we can set $\gamma(x)$ to the number of vertices of the connected component of $C_x - S$ that contains the rightmost spine vertex of $C_x$, or to $0$ if that vertex is in $S$.
\end{proof}

Let us now formally define the action of replacing caterpillars, which is the central operation used in our reduction rule to decrease the spine length.

\begin{definition}[Replacing Caterpillars]
    Let $C_1$ be a connected subgraph of $G - M$ that contains at least one spine vertex of $G - M$.
    Set $C$ to be the connected component of $G - M$ that is a supergraph of $C_1$.
    Let $C_2$ be an arbitrary caterpillar.
    Let $v_\ell$ be the rightmost spine vertex of $C$ that is to the left of $C_1$ (or $\bot \notin V(G)$ if no such vertex exists), and $v_r$ be the leftmost spine vertex of $C$ that is to the right of $C_1$ (or $\bot \notin V(G)$ if no such vertex exists).

    Then, the operation of replacing $C_1$ with $C_2$ in $G$ is described by the following procedure:
    \begin{enumerate}
        \item Start with $G$, and delete $C_1$.
        \item Add $C_2$ to the graph.
        \item If $v_\ell \not = \bot$, add an edge between the leftmost spine vertex of $C_2$ and $v_\ell$.
        \item If $v_r \not = \bot$, add an edge between the rightmost spine vertex of $C_2$ and $v_r$. Output the resulting graph.
    \end{enumerate}
\end{definition}

Hence, the replacement operation corresponds to the natural operation of cutting out caterpillar $C_1$ and replacing it with $C_2$.
Next, we show that replacing a caterpillar with essence $\gamma$ with any caterpillar that also has essence $\gamma$ is well-behaved.

\thmReplaceWithEssenceSafeNoModConnections*{}
\begin{proof}
    Let $S$ be a \dcocset{} of $G$ of size at most $k$, and let $C$ be the connected component of $G - M$ that contains $P$.
    Furthermore, let $v_\ell$ be the rightmost spine vertex of $C$ that is to the left of $P$, or $\bot \notin V(G) \cup V(G')$ if no such vertex exists.
    Let the leftmost spine vertex of $P$ be $w_1$, and the rightmost spine vertex of $P$ be $w_r$, let the leftmost spine vertex of $P'$ be $w_1'$ and the rightmost spine vertex of $P'$ be $w_r'$.
    We consider two cases.

    \proofsubparagraph*{Case 1: $v_\ell = \bot$ or $v_\ell \in S$ or $|S \cap V(P)| > \opt{P}$:}
    Let $S_P'$ be a minimum \dcocset{} for $P'$ that selects the rightmost spine vertex of $P'$, note that such a set exists by the properties of the solution-tight packings and the fact that $P'$ has a solution-tight packing containing all of its vertices (and $P'$ must have such a solution-tight packing, as the essence is only defined for such caterpillars).
    Then $S' = (S \setminus V(P)) \cup \{v_\ell\} \cup S_P'$ is a \dcocset{} for $G'$, because we select the rightmost spine vertex of $P'$, as well as the rightmost spine vertex to the left of $P'$ if that vertex exists, and furthermore no vertex of $P'$ has a neighbor in $M$.
    Moreover, since $S_P'$ is optimal for $P'$, this set $S'$ contains at most $k' = k - \opt{P} + \opt{P'}$ vertices of $G'$.

    \proofsubparagraph*{Case 2: $v_\ell \not = \bot$ and $v_\ell \notin S$ and $|S \cap V(P)| = \opt{P}$:}
    We collect an initial fact about this case.

    \begin{claim}
        \label{claim:essence_replacement_path_1}
        $S$ selects a spine vertex of $P$, and there is no path from $w_1$ to $w_r$ in $G - S$ that only uses vertices of $P$.
    \end{claim}
    \begin{claimproof}
        Let $H$ be the leftmost graph of the solution-tight packing of $P$.
        If $S$ did not select a spine vertex of $P$, then because $S \cap V(P)$ is a minimum \dcocset{} of $P$, $S$ selects exactly one vertex of $H$ and that vertex is a pendant.
        Since $|V(H)| \geq d+1$, $S$ must select all neighbors of $H$, in particular $S$ must select $v_\ell$, a contradiction to $v_\ell \notin S$.
        Therefore, $S$ selects a spine vertex of $H$.
        Moreover, there is no path from $w_1$ to $w_r$ in $G - S$ that only uses vertices of $P$ as any such path would need to contain the selected spine vertex.
    \end{claimproof}

    Let $x$ be $d$ if $w_1 \in S$.
    Otherwise, if $w_1 \notin S$, let $L$ be the connected component of $G - S$ that contains $w_1$, and $L_P$ be the connected component of $P - S$ that contains $w_1$.
    Clearly $L_P$ is a subgraph of $L$, and therefore $\abs{L} = \abs{L_P} + \abs{L - L_P}$.
    Set $x = |L| - |L_P|$, and observe that $|L_P| = |L| - x \leq d$ since $|L| \leq d$ and $x \geq 0$.

    Let $y$ be $0$ if $w_r \in S$.
    Otherwise, if $w_r \notin S$, let $R$ be the connected component of $w_r$ in $G -S$, and $R_P$ the connected component of $w_r$ in $P - S$.
    We have that $R_P$ is a subgraph of $R$.
    Set $y = |R_P|$ and observe that we have $|R| = y + |R| - |R_P| \leq d$.

    \begin{claim}
        We have $\gamma(x) \leq y$.
    \end{claim}
    \begin{claimproof}
        Consider the graph $P_x$, which is the same as $P$ only that $w_1$ gets $x$ additional pendants.
        Then, $S \cap V(P)$ is a \dcocset{} of $P_x$ since in $P - S$ each connected component that does not contain $w_1$ has size at most $d$, and if there is a connected component $L_P$ that contains $w_1$, then it has size $\abs{L_P} = \abs{L} - x$.
        Therefore, $S \cap V(P)$ shows $\gamma(x) \leq y$.
    \end{claimproof}

    Since $P'$ has the same essence as $P$, there is a minimum-sized solution $S_{P'}$ for $P'_x$ such that $w_r'$ is selected (which must be the case if $y = 0$) or such that $w_r'$ is in a connected component of $P'_x - S_{P'}$ of size at most $y$.
    Set $S' = (S \setminus V(P)) \cup S_{P'}$.
    We clearly have $\abs{S'} \leq k'$, so it remains to shows that $S'$ is a \dcocset{} of $G'$.
    To show this, we first collect some more facts about paths.

    \begin{claim}
        \label{claim:essence_replacement_path_connecting_wr_w1}
        Without loss of generality there is no path from $w_r'$ to $w_1'$ in $P' - S'$.
    \end{claim}
    \begin{claimproof}
        Without loss of generality $S_{P'}$ contains no pendants of $P'$, and as $|V(P')| \geq d+1$ $S_{P'}$ must contain a spine vertex, which would destroy such a path.
    \end{claimproof}

    \begin{claim}
        \label{claim:essence_replacement_path_w1}
        There is no path going from a vertex $v \notin V(P')$ to $w_1'$ in $G' - S'$ that contains a vertex of $P' - \{w_1'\}$.
    \end{claim}
    \begin{claimproof}
        Assume that such a path exists.
        Then, the path reaches some vertex $w' \not = w_1'$ of $P'$ before reaching $w_1'$.
        Since the only vertices of $P'$ that may have neighbors outside $P'$ are $w_1'$ and $w_r'$, the first vertex of $P'$ that the path reaches must be $w_r'$.
        In particular, it must reach $w_r'$ using the neighbor $v_r$ of $w_r'$ that is a spine vertex that is not part of $P'$.
        However, there is no path from $w_r'$ to $w_1'$ that does not use $v_r$ since by \cref{claim:essence_replacement_path_connecting_wr_w1} there is no path from $w_r'$ to $w_1'$ that only uses vertices of $P'$.
        Therefore, the path contains $v_r$ twice, which is not possible.
    \end{claimproof}

    Since $w_1'$ and $w_r'$ are the only vertices of $P'$ that have neighbors outside $P'$, the only connected components of $G' - S'$ that could have size larger than $d$ are those containing $w_1'$ and $w_r'$.

    \proofsubparagraph*{Connected component containing $w_1'$.}
    Assume that there is a connected component $L'$ of $G' - S'$ that contains $w_1'$.
    Note that this implies that $w_1 \notin S$, as otherwise $S_{P'}$ being a \dcocset{} of $P'_x = P'_d$ implies that one can assume $w_1' \in S_{P'}$.
    Set $L'_{P'}$ to be the connected component of $P' - S'$ that contains $w_1'$.
    Since $S_{P'}$ is a \dcocset{} of $P'_x$ we have $|L'_{P'}| \leq d - x$.
    We will design an injective function $\delta$ that assigns each vertex $v \in L' - L'_{P'}$ a vertex $v' \in L - L_P$.
    This function $\delta$ then proves that $|L'| - |L'_{P'}| \leq |L| - |L_P|  = x$, and therefore $|L'| = |L'_{P'}| + |L'| - |L'_{P'}| \leq d - x + x = d$.

    Let $v$ be a vertex of $L' - L'_{P'}$.
    Then, there is some path from $v$ to $w_1'$ in $G' - S'$ such that this path does not exist in $P' - S'$.

    If $v \notin P'$, then \cref{claim:essence_replacement_path_w1} shows that the path from $v$ to $w_1'$ does not contain a vertex of $P' - \{w_1'\}$.
    Therefore, there is a path from $v$ to $w_1$ in $G - P$, showing that $v$ is also part of $L$.
    But clearly $v \notin V(L_P)$ since $V(L_P) \subseteq V(P)$.
    Hence, $v \in L - L_P$ and
    we can set $\delta(v) = v$.

    Otherwise, $v \in P'$.
    Therefore, the path from $v$ to $w_1'$ contains the vertex $w_r'$, otherwise $v$ would also be part of $L'_{P'}$.
    Let $R'_{P'}$ be the connected component of $w_r'$ in $P' - S'$.
    Since $v$ was arbitrary, the vertices of $P'$ which are part of $L' - L'_{P'}$ are exactly the vertices of $R'_{P'}$.
    By the properties of the essence $\gamma$ and our choice of $S_{P'}$ we have $\abs{R'_{P'}} \leq y = \abs{R_P}$.
    Moreover, since $v$ exists we have $y \geq 1$ and thus vertex $w_r \notin S$.
    Since there is a path from $w_r'$ to $w_1'$ in $G' - S'$, but by \cref{claim:essence_replacement_path_connecting_wr_w1} no path from $w_r'$ to $w_1'$ that lies within $P'$, there must be a path from $w_r$ to $w_1$ in $G - S$.
    Therefore, each vertex of $R_P$ is part of $L$.
    However, since there is no path from $w_r$ to $w_1$ in $P - S$ by \cref{claim:essence_replacement_path_1}, no vertex of $R_P$ is part of $L_P$.
    Since $\abs{R_P} \geq \abs{R'_{P'}}$, we can injectively map the vertices of $R'_{P'}$ to the vertices of $R_P$ with $\delta$.

    \proofsubparagraph*{Connected component containing $w_r'$.}
    Now, we assume that there is a connected component $R'$ of $G' - S'$ that contains $w_r'$.
    This implies $y \geq 1$ and therefore $w_r \notin S$.
    We may additionally assume that $R'$ does not contain $w_1'$, as we have already dealt with the component that contains $w_1'$.
    Let $R'_{P'}$ be the connected component of $w_r'$ in $P' - S'$.
    By our choice of solution $S_{P'}$ we have $\abs{R'_{P'}} \leq y$.
    Similarly to the proof for the component $L'$ we will map each vertex of $R' - R'_{P'}$ injectively to a vertex of $R - R_P$ using a function $\delta$.
    Indeed, if we had such a function $\delta$ then we would get $|R'| = |R'_{P'}| + |R'| - |R'_{P'}| \leq y + |R| - |R_P| \leq d$.

    Let $v$ be a vertex of $R' - R'_{P'}$.
    Then, $v$ has a path to $w_r'$, and this path does not exist in $P' - S'$.
    Moreover, since $w_1'$ is not part of $R'$ the path does not contain $w_1'$.
    So, $v$ cannot be part of $P'$.
    Therefore, the path from $v$ to $w_r'$ gives rise to a path from $v$ to $w_r$ in $G - S$.
    We can map $\delta(v) = v$ to finish the proof.
\end{proof}

To replace a graph with a certain essence $\gamma$ with a smaller graph with the same essence $\gamma$, we also must be able to quickly compute small graphs with essence $\gamma$.
This task is highly non-trivial, but, as the next lemma shows, possible.

\thmComputeCaterpillarEssence*{}

We defer the proof to \cref{sec:caterpillar_essence}, in which we also cover the interesting connection between the essences of caterpillars and abstract algebra.
An interesting aspect of the lemma is that for any conceptually possible essence, there is a small caterpillar that has this essence.

We intend to use \cref{thm:compute_caterpillars_with_essence,thm:replace_with_essence_safe_if_no_mod_connections} to replace some subgraphs of $G - M$ with smaller caterpillars that have the same essence, however, the connections to the modulator that these connected components can have cause some additional difficulties.
Before stating the actual reduction rule, we thus need to figure out which parts of $G - M$ we can actually safely replace, and for this purpose, we again utilize graph packings.

\begin{lemma}
    \label{thm:many_lefw_right_blocks}
    Let $C$ be a connected subgraph of $G - M$, and $\mathcal{P}$ be a packing of pairwise vertex-disjoint connected subgraphs of $C$, such that every graph of $\mathcal{P}$ contains at least one spine vertex of $C$.

    For any integers $c_1,c_2 \geq 1$ there exists a set $\partitionSeeingLotsBlocksLeftAndRight \subseteq \mathcal{P}$ with $\abs{\partitionSeeingLotsBlocksLeftAndRight} \geq \abs{\mathcal{P}} - \abs{M} \cdot (c_1 + c_2)$ such that for each $P \in \partitionSeeingLotsBlocksLeftAndRight$ and for each $v \in N(P) \cap M$ vertex $v$ has at least $c_1$ neighbors in graphs of $\mathcal{P}$ to the left of $P$, and at least $c_2$ neighbors in graphs of $\mathcal{P}$ to the right of $P$.
\end{lemma}
\begin{proof}
    Since each graph of $\mathcal{P}$ has at least one spine vertex, there is a natural order among them determined by the spine.
    For each $v \in M$ we perform the following marking procedure.
    Vertex $v$ marks the $c_1$ leftmost graphs of $\mathcal{P}$ that $v$ has edges to, and also marks the $c_2$ rightmost graphs of $\mathcal{P}$ that $v$ has edges to.
    Our set $\partitionSeeingLotsBlocksLeftAndRight$ consists of all unmarked graphs of $\mathcal{P}$.
    Since each vertex of $M$ marks at most $c_1 + c_2$ graphs of $\mathcal{P}$ we have $\abs{\partitionSeeingLotsBlocksLeftAndRight} \geq \abs{P} - \abs{M} \cdot (c_1 + c_2)$.

    Now, consider some $P \in \partitionSeeingLotsBlocksLeftAndRight$.
    Then, for each $v \in N(P) \cap M$ we know that $v$ did not mark $P$, and therefore $v$ has edges to $c_1$ graphs of $\mathcal{P}$ which are left of $P$, and $c_2$ graphs of $\mathcal{P}$ which are right of $P$.
\end{proof}

We can use this fact to find specific subgraphs of the $\smallPackingConstant$-merged packing.
For convenience, we define $\packingConstant = \smallPackingConstant \cdot (\abs{M} \cdot 2d + 1)$.
We will utilize the $\packingConstant$-merged packing.

\begin{restatable}{lemma}{thmfindAppropriateSubgraphs}
    \label{thm:find_appropriate_subgraphs}
    Let $C$ be a connected component of $G - M$, $\mathcal{P}_\packingConstant$ be the $\packingConstant$-merged packing of $C$, and $\mathcal{P}_{{\smallPackingConstant}}$ be the ${\smallPackingConstant}$-merged packing of $C$.
    Then, for each graph $G_1$ in $\mathcal{P}_{\packingConstant}$ there exists a subgraph $G_2$ of $G_1$ such that 
    \begin{enumerate}
        \item $G_2 \in \mathcal{P}_{{\smallPackingConstant}}$, and 
        \item each $v \in N(G_2) \cap M$ has $d$ neighbors in $G_1$ which are left of $G_2$,
        \item each $v \in N(G_2) \cap M$ has $d$ neighbors in $G_1$ which are right of $G_2$.
    \end{enumerate}
    Moreover, we can find $G_2$ in polynomial-time given $G_1$.
\end{restatable}
\begin{proof}
    Let $G_1$ be an arbitrary graph of $\mathcal{P}_\packingConstant$.
    Then, because $\packingConstant = {\smallPackingConstant} \cdot (|M| \cdot 2d + 1)$, it is the case that $G_1$ itself consists of $|M| \cdot 2d + 1$ graphs of $\mathcal{P}_{{\smallPackingConstant}}$.

    By applying \cref{thm:many_lefw_right_blocks} on component $C$ of $G' - M$ with these $|M| \cdot 2d+1$ packed subgraphs, we see that there exists a graph $G_2$ in $\mathcal{P}_{\smallPackingConstant}$ such that each $v \in N(G_2) \cap M$ has $d$ neighbors in $G_1$ which are left of $G_2$, and $d$ neighbors in $G_1$ which are right of $G_2$.
    We can also find $G_2$ in polynomial-time by simply iterating over all $|M| \cdot 2d+1$ packed subgraphs that make up $G_1$ and checking whether they fulfill the desired property.
\end{proof}

Now, we are finally ready to state the reduction rule we use to decrease the number of spine vertices of connected components of $G - M$.

\begin{reductionrule}
    \label{rule:replace_by_essence_graph}
    Let $(G,d,k,M)$ be an instance of \cocModToCaterpillarsPlusD{}.
    Let $\mathcal{X}$ be the set of chunks of the instance, and let $\mathcal{P}_\packingConstant$ be the $\packingConstant$-merged packing of $G - M$, and let $\mathcal{P}_{\smallPackingConstant}$ be the ${\smallPackingConstant}$-merged packing of $G-M$.
    The rule is applicable if $\abs{\mathcal{P}_\packingConstant} > \abs{M}^2 \cdot (\abs{M} + 2(d-1))$.

    Proceed with the following steps.
    \begin{enumerate}
        \item
              \label{rule_step:replace_by_essence_graph:marking}
              For each chunk $X \in \mathcal{X}$ set $\mathcal{P}^X_\packingConstant = \{\hat P \in \mathcal{P}_\packingConstant \mid \conflicts{X}{\hat P} > 0\}$.
              For each $X \in \mathcal{X}$ mark $\min(|M| + 2(d-1), \abs{\mathcal{P}^X_\packingConstant})$ arbitrary graphs of $\mathcal{P}^X_\packingConstant$.
        \item Fix an arbitrary unmarked $\hat P \in \mathcal{P}_\packingConstant$.
              \label{rule_step:replace_by_essence_graph:fix_unmarked_graph}
        \item
              \label{rule_step:replace_by_essence_graph:fix_subgraph}
              Fix an arbitrary subgraph $P$ of $\hat P$ with $P \in \mathcal{P}_\smallPackingConstant$ such that each $v \in N(P) \cap M$ has $d$ neighbors in $\hat P$ left of $P$ and $d$ neighbors in $\hat P$ right of $P$.
        \item
              \label{rule_step:replace_by_essence_graph:compute_essence}
              Compute the essence $\gamma$ of $P$.
        \item
              \label{rule_step:replace_by_essence_graph:compute_essence_graph}
              Compute a graph $P_\gamma$ with a solution-tight packing of size at most $d^3$ and essence $\gamma$.
        \item
              \label{rule_step:replace_by_essence_graph:replace_graphs}
              Create $G'$ by replacing $P$ with $P_\gamma$, and set $k' = k - \opt{P} + \opt{P_\gamma}$.
        \item
              \label{rule_step:replace_by_essence_graph:return_output_instance}
              Output the instance $(G',d,k',M)$.
    \end{enumerate}
\end{reductionrule}

Next, we argue that the rule can be executed in polynomial-time and that the steps are well-defined.

\begin{lemma}
    \label{thm:replace_by_essence_graph_rule_poly_time}
    \cref{rule:replace_by_essence_graph} can be executed in polynomial-time and each step is well-defined.
    Moreover, if $(G,d,k,M)$ is the instance before the rule is executed and $(G',d,k',M)$ is the instance after the rule is executed, then the solution-tight packing of $G' - M$ is strictly smaller than the solution-tight packing of $G - M$.
\end{lemma}
\begin{proof}
    It is clear that \cref{rule_step:replace_by_essence_graph:marking} can be performed in polynomial-time since we can compute $\conflicts{X}{\hat P}$ in polynomial-time for each $\hat P \in \mathcal{P}_\packingConstant$ and $X \in \mathcal{X}$.

    The graph $\hat P$ in \cref{rule_step:replace_by_essence_graph:fix_unmarked_graph} is well-defined since $\abs{\mathcal{P}_\packingConstant} > \abs{M}^2 \cdot (\abs{M} + 2(d-1))$ if the rule is applicable, for each $X \in \mathcal{X}$ we mark at most $\abs{M} + 2(d-1)$ graphs of $\mathcal{P}_\packingConstant$, and $\abs{\mathcal{X}} \leq \abs{M}^2$.

    The graph $P$ fixed in \cref{rule_step:replace_by_essence_graph:fix_subgraph} exists by \cref{thm:find_appropriate_subgraphs}.

    Since $P$ is part of the $\mathcal{P}_\smallPackingConstant$-merged packing the solution-tight packing of $P$ contains all vertices of $P$.
    Therefore, $P$ has an essence that we can compute in polynomial-time using \cref{thm:basic_essence_properties}.
    Hence, also \cref{rule_step:replace_by_essence_graph:compute_essence} can be done in polynomial-time.

    \cref{rule_step:replace_by_essence_graph:compute_essence_graph} can be executed quickly enough by applying \cref{thm:compute_caterpillars_with_essence}.

    It is also clear that \cref{rule_step:replace_by_essence_graph:replace_graphs,rule_step:replace_by_essence_graph:return_output_instance} can be done in polynomial-time since \coc{} can be solved in polynomial-time on caterpillars.

    To see that the rule reduces the size of the solution-tight packing observe that the graph $P$ has a solution-tight packing that contains $d^3 + 1$ graphs, whereas the solution-tight packing of $P_\gamma$ contains only $d^3$ graphs.
    The solution-tight packing of $G - M$ is identical to that of $G' - M$, apart from the packing of $P$ being exchanged for that of $P_\gamma$.
    Therefore, the size of the solution-tight packing is decreased by the rule application.
\end{proof}

Now, we proceed to showing that the rule is safe.
For this purpose, we first introduce two further auxiliary lemmas.

\begin{restatable}{lemma}{thmSimpleMarking1}
    \label{thm:simple_marking_1}
    Let $\hat{\mathcal{P}}$
    be a packing of pairwise vertex-disjoint connected merged subgraphs of $G - M$.
    For each chunk $X \in \mathcal{X}$ let $\hat{\mathcal{P}}^X = \{H \in \hat{\mathcal{P}} \mid \conflicts{H}{X} > 0\}$.
    Then, there is a minimum \dcocset{} $S$ of $G$ that intersects each chunk $X$ where $|\hat{\mathcal{P}}^X| \geq |M| + 2(d-1)$.
\end{restatable}
\begin{proof}
    Let $S$ be a minimum \dcocset{} of $G$ and $X$ be a chunk such that $|\hat{\mathcal{P}}^X| \geq |M| + 2(d-1)$ and assume that $S \cap X = \emptyset$.
    Since $|X| \leq 2$, at most $2(d-1)$ neighbors of $X$ are not selected by $S$.
    Thus, in $|M|$ graphs of $\hat{\mathcal{P}}^X$ all neighbors of $X$ are selected.
    Hence, in each of these graphs $\tilde H$ we have that $S \cap V(\tilde H) > \opt{\tilde H} $.
    In particular, since each such graph $\tilde H$ consists of some $\alpha$ graphs of the solution-tight packing $\mathcal{P}$ of $G - M$, there is a subgraph $\tilde H' \in \mathcal{P}$ of $\tilde H$ such that $S$ selects at least two vertices of $\tilde H'$.
    Overall, we can find $|M|$ vertex-disjoint graphs of $\mathcal{P}$ such that $S$ selects at least two vertices of each of them.
    This shows $S \geq |M| + \opt{G - M}$.
    Hence, the solution that selects all of $M$ and a minimum \dcocset{} of $G - M$ is also optimal and intersects all chunks.
\end{proof}

\begin{restatable}{lemma}{thmSimpleMarking2}
    \label{thm:simple_marking_2}
    Let $S \subseteq M$ be a set and $H$ a subgraph of $G - M$ such that $S \cap X \not = \emptyset$
    for all chunks $X \in \mathcal{X}$ with $\conflicts{X}{H} > 0$.
    Then, there is a minimum \dcocset{} $S_H$ of $H$ that selects all vertices of $N(M \setminus S) \cap V(H)$.
\end{restatable}
\begin{proof}
    Toward a contradiction assume that there is no minimum \dcocset{} of $H$ that selects all vertices of $B = N(M \setminus S) \cap V(H)$.
    Then, $B$ is a $d$-blocking set of $H$.
    Since minimal $d$-blocking sets of caterpillar forests have size at most two, we find a set $B' \subseteq B$ that is also a $d$-blocking set of $H$ and $|B'| \leq 2$.
    Let $B' = \{b_1,b_2\}$, we may have $b_1 = b_2$.
    Let $m_1, m_2$ be neighbors of $b_1$, respectively $b_2$ which are part of $M \setminus S$.
    Then $B'$ certifies that $\conflicts{H}{\{m_1,m_2\}} > 0$.
    But $\{m_1,m_2\}$ is a chunk, contradicting our assumption that $S$ hits all chunks that have nonzero conflict with $H$.
\end{proof}

Note that the central idea behind these two lemmas was first used by Jansen \& Bodlaender \cite{jansenVertexCoverKernelization2013b} to reduce the number of connected components, although we put a novel spin on it using the solution-tight packings to make \cref{thm:simple_marking_1} work for graphs that are not entire connected components.
Now, we are ready for the full proof of safety of the reduction rule.

\begin{lemma}
    \cref{rule:replace_by_essence_graph} is safe.
\end{lemma}
\begin{proof}
    Let $(G,d,k,M)$ be the instance before performing the rule, and let $(G',d,k',M)$ be the instance after performing the rule.
    Let $P$ be the caterpillar graph that is replaced with $P'$, where $\gamma$ is the essence of $P$ and of $P'$.
    Let $C$ be the connected component of $G - M$ that contains $P$, and let $C'$ be the connected component of $G' - M'$ that contains $P'$.
    We prove both directions of safety individually.

    \begin{claim}
        \label{claim:replace_by_essence_safety_forward}
        If the input instance is a yes-instance, then the output instance is a yes-instance.
    \end{claim}
    \begin{claimproof}
        Assume that the input instance is a yes-instance, and let $S$ be an \dcocset{} of $G$ of size at most $k$.
        Then it follows directly from \cref{thm:replace_with_essence_safe_if_no_mod_connections} that there is a \dcocset{} $S'$ of size at most $k'$ for $G'$.
    \end{claimproof}

    We can now proceed to the backward direction of the proof.
    In the forward direction, it was very convenient that $P'$ has no edges to the modulator, as that immediately allowed us to use \cref{thm:replace_with_essence_safe_if_no_mod_connections} however, as $P$ can have such edges, this direction is a bit more complicated to prove.

    \begin{claim}
        \label{claim:replace_by_essence_safety_backward}
        If the output instance is a yes-instance, then the input instance is a yes-instance.
    \end{claim}
    \begin{claimproof}
        Let $S'$ be a \dcocset{} of $G'$ of size at most $k'$.
        Set $M' = N_G(P) \cap (M \setminus S')$, that is, $M'$ is the set of unselected modulator vertices that $P$ has edges to.

        \proofsubparagraph*{Case 1: $M' = \emptyset$:}
        Set $\tilde G' = G' - (S' \cap M)$, and $\tilde G = G - (S' \cap M)$.
        Clearly, $S' \cap V(\tilde G')$ is a \dcocset{} of $\tilde G'$ of size at most $k' - |S' \cap M'|$.
        Moreover, since $M' = \emptyset$ graph $P$ has no neighbors in $\tilde G$.
        Hence, we can create $\tilde G$ by replacing $P'$ with $P$ in $\tilde G'$.
        Observe that the solution-tight packing of $\tilde G' - M = G' - M$ is almost identical to the solution-tight packing of $\tilde G - M = G - M$, the only difference is that the solution-tight packing of $P$ is replaced by the solution-tight packing of $P'$.
        Therefore, $P'$ is a connected merged subgraph of $\tilde G' - M$.
        By \cref{thm:replace_with_essence_safe_if_no_mod_connections} the graph $\tilde G$ then has a \dcocset{} $\tilde S$ of size at most $k - |S' \cap M'|$.
        Hence, $S = \tilde S \cup (S' \cap M)$ is a \dcocset{} for $G$ of size at most $k$.

        \proofsubparagraph*{Case 2: $M' \not = \emptyset$:}
        Let $\mathcal{P}_\packingConstant$ be the $\packingConstant$-merged packing of $G - M$, and $\mathcal{P}_\smallPackingConstant$ the $\smallPackingConstant$-merged packing of $G - M$.
        Moreover, let $\hat P$ be the graph of $\mathcal{P}_\packingConstant$ that contains $P$, and recall that $P$ is actually a graph of the packing $\mathcal{P}_{\smallPackingConstant}$.

        As $M'$ is nonempty, and each vertex of $M'$ has at least $d$ neighbor in $\hat P$ to the left and right of $P$, there is a rightmost spine vertex of $\hat P$ that is to the left of $P$ in $G$, call this spine vertex $v_\ell$, and a leftmost spine vertex of $\hat P$ that is to the right of $P$ in $G$, call this spine vertex $v_r$.
        Also denote the leftmost spine vertex of $P$ as $w_1$, and the leftmost spine vertex of $P'$ as $w_1'$.
        The rightmost spine vertex of $P$ is denoted as $w_r$, and the rightmost spine vertex of $P'$ as $w_r'$.

        Let $\mathcal{M}$ be the set of graphs of $\mathcal{P}_\packingConstant$ which were marked by the reduction rule.
        Observe that the solution-tight packing of $G' - M$ is obtained from the solution-tight packing of $G - M$ by replacing the solution-tight packing of $P$ with the solution-tight packing of $P'$.
        Therefore, all graphs of $\mathcal{M}$ are merged subgraphs of $G'$.
        In the reduction rule the graph $\hat P$ was not marked, and hence for each chunk $X$ with $\conflicts{X}{\hat P} > 0$ we have that we marked at least $|M| + 2(d-1)$ graphs of $\mathcal{P}_\packingConstant^X = \{T \in \mathcal{P}_\packingConstant \mid \conflicts{X}{T} > 0\}$ which are not $\hat P$.
        Therefore, we have that $|\mathcal{P}^X| > |M| + 2(d-1)$ for each chunk $X$ such that $\conflicts{X}{\hat P} > 0$.
        We now apply \cref{thm:simple_marking_1} for the packing $\mathcal{M}$ and the output instance.
        This yields that we can assume that $S'$ intersects all chunks $X$ with $\conflicts{X}{\hat P} > 0$.
        Then, \cref{thm:simple_marking_2} applied to the set $S' \cap M$ and $\hat P$ on the input instance yields that there is a minimum \dcocset{} $S_{\hat P}$ of $\hat P$ that selects all vertices of $N(M') \cap V(\hat P)$.

        \proofsubparagraph*{Building a small solution for $G$.}
        Recall that each vertex of $M'$ has neighbors to at least $d$ vertices of $\hat P$ which are to the left of $P$, and at least $d$ vertices of $\hat P$ which are to the right of $P$.
        For an arbitrary $m \in M'$, let $L_m$ be any $d$ neighbors of $m$ to the left of $P$ which are part of $\hat P$, and $R_m$ be any $d$ neighbors of $m$ of $\hat P$ which are to the right of $P$.
        Since $m$ is not selected, $S'$ must select at least one vertex of $L_m$, and at least one vertex of $R_m$.
        Let $n_\ell$ be a vertex of $L_m$ that $S'$ selects, and $n_r$ a vertex of $R_m$ that $S'$ selects.

        We know that both $n_\ell$ and $n_r$ are also part of $S_{\hat P}$, as $S_{\hat P}$ selects all vertices of $N(M') \cap V(\hat P)$.
        This common ground between $S'$ and $S_{\hat P}$ is what we exploit to build a small solution $S$ for the input instance.

        Let the solution-tight packing of $\hat P$ be $\{P_1,\dots,P_t\}$, where $P_i$ is to the left of all graphs $P_j$ with $j > i$ for any $i$.
        Recall that the solution-tight packing of $\hat P$ contains all vertices of $\hat P$.
        Then, for any vertex $v$ of $\hat P$, let $\mu(v)$ be the integer such that $v \in V(P_{\mu(v)})$.
        We are now interested in $P^* = \hat P\left[\bigcup_{i \in \range[\mu(n_\ell) + 1]{\mu(n_r) - 1}} V(P_i)\right]$.
        In other words, $P^{*}$ is the subgraph of $\hat P$ consisting of all graphs of $P_1,\dots,P_t$ which are in between the graph containing $n_\ell$ and the graph containing $n_r$.
        Observe that $P^*$ contains $P$ as subgraph.
        Define $P^{*\prime}$ as the graph corresponding to $P^*$ in $G'$, that is, the graph we obtain when replacing $P$ of $P^{*}$ with $P'$.
        Clearly $P^{*\prime}$ is a subgraph of $G'$.
        Because $P$,$P'$,$P^{*}$ all have solution-tight packings that contain all of their vertices, also $P^{*\prime}$ has such a packing.
        In particular, $\opt{P^{*\prime}} = \opt{P^{*}} - \opt{P} + \opt{P'}$.

        We set $S = (S' \setminus V(P^{*\prime})) \cup (V(P^{*}) \cap S_{\hat P})$, and want to now show that $S$ is a \dcocset{} of $G$ of size at most $k$.
        We first argue that $S$ has size at most $k$.
        Observe that $S_{\hat P} \cap V(P^{*})$ is a minimum \dcocset{} of $P^{*}$, so $\abs{V(P^{*}) \cap S_{\hat P}} = \opt{P^{*}}$.
        Furthermore, we know that $\abs{S' \cap V(P^{*\prime})} \geq \opt{P^{*}} - \opt{P} + \opt{P'}$.
        This directly yields that $\abs{S} \leq k$.
        So, all that remains to conclude the proof is showing that $S$ is actually a \dcocset{} of $G$.

        Recall that $P_1,\dots,P_t$ are the graphs of the solution-tight packing of $\hat P$, and therefore any minimum \dcocset{} of $\hat P$ contains exactly one vertex of each such graph $P_i$.
        Using the fact that $S_{\hat P}$ is a minimum \dcocset{} of $\hat P$ that selects all neighbors of $M'$, we see that $S_{\hat P}$ contains only vertex $n_\ell$ of $P_{\mu(n_\ell)}$, implying that $n_\ell$ is the only vertex of $P_{\mu(n_\ell)}$ that has a neighbor in $M'$.
        Similarly, $n_r$ is the only vertex of $P_{\mu(n_r)}$ that has a neighbor in $M'$, and $S_{\hat P}$ contains only vertex $n_r$ of $P_{\mu(n_r)}$.
        Denote the leftmost spine vertex of $P^*$ as $v_\ell^*$ and the rightmost spine vertex of $P^*$ as $v_r^*$.

        Now, consider the case that $n_\ell$ is a pendant.
        Then vertex $v_\ell^*$ must be selected by $S_{\hat P}$, because otherwise $v_\ell^*$ and the at least $d$ unselected vertices of $P_{\mu(n_\ell)}$ already lead to a connected component of size larger than $d$ that lies solely within $\hat P$.
        For the same reason, $S_{\hat P}$ contains $v_r^*$ of $P^{*}$ if $n_r$ is a pendant.

        By construction, we have that $S \cap V(P^{*})$ is a \dcocset{} of $P^{*}$, and also $S \cap V(G - P^{*})$ is a \dcocset{} of $G - P^{*}$.
        So, any connected component of $G - S$ that has size larger than $d$ would have to contain vertices of both $P^{*}$ and $G - P^{*}$.
        Let $Z$ be a connected component of $G - S$ that contains both vertices of $P^{*}$ and of $G - P^{*}$.
        Since $S_{\hat P}$ selects all neighbors of $M'$ in $\hat P$, we have that $S$ selects all vertices of $P^{*}$ that have a neighbor in $M'$.
        This means that $Z$ must contain $v_\ell^*$ or $v_r^*$, as no other vertex of $P^{*}$ can have a neighbor outside $P^{*}$ in $G - S$.

        Assume that $v_\ell^*$ is not in $S$ and consider a path in $G - S$ that leaves $v_\ell^*$ to the left.
        If $n_\ell$ is a pendant, then $v_\ell^*$ is selected, which is not the case by assumption.
        Otherwise, $n_\ell \in S$ is a spine vertex, and $n_\ell$ is the only vertex of $P_{\mu(n_\ell)}$ that has a neighbor in $M'$.
        Therefore,  the path can only contain vertices of $P_{\mu(n_\ell)}$ and vertex $v_\ell^*$: the path cannot go further to the left because spine vertex $n_\ell$ is selected, and the path cannot reach $M'$.
        Similar reasoning shows that a path starting in $v_r^*$ that goes to the right can only contain vertices of $P_{\mu(n_r)}$ and vertex $v_r^*$.
        Since any vertex $w$ of $Z$ has a path to $v_\ell^*$ or $v_r^*$, this shows that all vertices of $Z$ are part of $V(P^*) \cup P_{\mu(n_\ell)} \cup P_{\mu(n_r)}$.

        Now, let $\tilde P = \hat P\left[\bigcup_{i \in \range[\mu(n_\ell)]{\mu(n_r)}} V(P_i) \right]$.
        Graph $\tilde P$ is similar to $P^{*}$ but also contains the graphs $P_{\mu(n_\ell)}$ and $P_{\mu(n_r)}$.
        By the elaborations above, we find that each vertex of $Z$ is a vertex of $\tilde P$.
        Furthermore, we have that $S_{\hat P} \cap V(\tilde P) \subseteq S \cap V(\tilde P)$ by our choice of $S$ and the fact that $n_\ell,n_r \in S_{\hat P} \cap S'$.
        Therefore, the vertices of $Z$ also belong to the same connected component in $\tilde P - S_{\hat P}$ and since $S_{\hat P} \cap V(\tilde P)$ is a \dcocset{} of $\tilde P$ we have $\abs{Z} \leq d$.
    \end{claimproof}
    The lemma follows by combining \cref{claim:replace_by_essence_safety_backward,claim:replace_by_essence_safety_forward}.
\end{proof}

\subsection{The Kernelization Algorithm}

As a last ingredient before providing our kernelization algorithm, we state a result by Xiao \cite{xiaoLinearKernelsSeparating2017a} for \coc{} parameterized by the sum of the solution size $k$ and the integer $d$, which we denote as \cocSolSizePlusD{}.

\begin{theorem}[{Follows from \cite[Theorem 2]{xiaoLinearKernelsSeparating2017a}}]
    \label{thm:xiao_kernel_solution_size}
    \cocSolSizePlusD{} admits a kernel with $\Oh(dk)$ vertices.
\end{theorem}

We are now ready to describe the kernelization algorithm and prove that reduced instances are small.

\mainThmPolyKernelCaterpillarMod*{}
\begin{proof}
    We can clearly solve \coc{} in polynomial-time on caterpillars, $\mathcal{G}$ is hereditary, and minimal $d$-blocking sets of caterpillar forests have size at most two by \cref{thm:caterpillar_bounded_blocking_sets}, so we can apply the algorithm of \cref{thm:remove_components_quickly_d_non_constant} to reduce the number of connected components.
    Similarly, by \cref{thm:replace_by_essence_graph_rule_poly_time} we can execute \cref{rule:replace_by_essence_graph} in polynomial-time.

    Let $(G,d,k,M)$ be the input instance.
    If $|M| = 0$, then $G$ is a caterpillar forest, so we can solve \coc{} in polynomial-time, and output a trivial yes-instance or a trivial-no instance depending on whether $(G,d,k,M)$ is a yes-instance or a no-instance.

    Otherwise, we have $|M| \geq 1$.
    We first apply \cref{rule:replace_by_essence_graph} repeatedly until it is no longer applicable.
    Then we apply the algorithm of \cref{thm:remove_components_quickly_d_non_constant} (which is essentially \cref{rule:remove_connected_components_expansion_lemma}).
    By \cref{thm:replace_by_essence_graph_rule_poly_time} we have that \cref{rule:replace_by_essence_graph} is run a polynomial number of times since each application of \cref{rule:replace_by_essence_graph} reduces the size of the solution-tight packing.
    Neither rule changes the value $d$ or the modulator $M$.
    Let $(G_2,d,k_2,M)$ be the instance at this point.
    If $k_2 < 0$, we output a trivial no-instance.
    Otherwise, let $\mathcal{C}$ be the connected components of $G_2 - M$.
    From \cref{thm:remove_components_quickly_d_non_constant}, we obtain that $\abs{\mathcal{C}} \leq ((d-1) \cdot 2 + 1) \cdot \abs{M}^2 \leq 2d\abs{M}^2$.
    Let $\mathcal{P}_\packingConstant$ be the $\packingConstant$-merged packing of $G_2 - M$, and observe that we have $\abs{\mathcal{P}_\packingConstant} \leq \abs{M}^2 \cdot (\abs{M} + 2(d-1)) \leq \abs{M}^3 + 2d \abs{M}^2$, as otherwise \cref{rule:replace_by_essence_graph} would have still been applicable given that the algorithm of \cref{thm:remove_components_quickly_d_non_constant} does not increase the size of the solution-tight packing.
    Also recall that $\packingConstant = \smallPackingConstant \cdot (\abs{M} \cdot 2d + 1) \leq 6d^4\abs{M}$.

    Let $C$ be an arbitrary component of $\mathcal{C}$.
    Observe that fewer than $\packingConstant \cdot (d+1)$ spine vertices of $C$ are not in some graph of $\mathcal{P}_\packingConstant$, as otherwise, one could strictly increase the size of $\mathcal{P}_\packingConstant$, which is not possible.
    Furthermore, observe that each graph in $\mathcal{P}_\packingConstant$ contains at most $(d+1) \cdot \packingConstant$ spine vertices.
    Therefore, the number of vertices of the spine of $G_2 - M$ is bounded by
    \begin{align*}
        (|\mathcal{C}| + |\mathcal{P}_\packingConstant|)  \cdot (d+1) \cdot \packingConstant & \leq  (2d\abs{M}^2 + \abs{M}^3 + 2d\abs{M}^2)  \cdot (d+1) \cdot 6d^4 \abs{M} \\&\leq 48 d^6 \abs{M}^3 + 12 d^5 \abs{M}^4.
    \end{align*}

    Now, let us set $M' = M \cup V(\spine{G_2 - M})$, that is, we create $M'$ by adding the spine of $G_2 - M$ to the modulator.
    Since $G_2 - M'$ is an independent set, we have that $M'$ is a \dcocset{} of $G_2$.
    Thus, the instance is a yes-instance if $k \geq 48 d^6 \abs{M}^3 + 13 d^5 \abs{M}^4 \geq \abs{M'}$, in which case we output a trivial yes-instance.
    Otherwise, we have that $k$ is at most $48 d^6 \abs{M}^3 + 13 d^5 \abs{M}^4$, therefore we can apply the kernel of \cref{thm:xiao_kernel_solution_size} to create an equivalent instance with $\Oh(d^7 \abs{M}^3 + d^6 \abs{M}^4)$ vertices that our kernel outputs.

    Finally, we observe that since none of the used preprocessing routines change the value $d$ (this also holds for the kernel by Xiao of \cref{thm:xiao_kernel_solution_size}), the algorithm also works for the problem \dcocModToCaterpillars{} where $d$ is a fixed constant.
\end{proof} 
\section{The Essence of Caterpillars and Monoids}
\label{sec:caterpillar_essence}
We now cover the proof of \cref{thm:compute_caterpillars_with_essence}, which states that for any possible essence of a caterpillar, we can quickly compute a small caterpillar of low degree that has such an essence.
Interestingly, using the essence of caterpillars, we can define a certain monoid, and phrase the type of object we seek in algebraic terms.
We begin by defining the essence monoid.
\defEssenceMonoid*{}

It is not difficult to confirm that $(\mathcal{M}_d,\circ, \idFunction)$ is a monoid: Function composition is associative, the identity function is clearly non-decreasing and $\idFunction(0) = 0$ and $\idFunction(d+1) = d+1$, and $\mathcal{M}_d$ is closed under function composition, because the function composition of any non-decreasing function is again non-decreasing.

Regarding the connections to caterpillars, we have already argued that each essence of a caterpillar is a function of $\mathcal{M}_d$ in \cref{thm:basic_essence_properties}.
As we will see, it also holds that for any function $\gamma$ of $\mathcal{M}_d$ there is a caterpillar that has $\gamma$ as its essence that can be computed quickly.
For now, we will disregard this connection to caterpillars, and establish algorithmic properties of the essence monoid itself.

A set $\mathcal{M}_d' \subseteq \mathcal{M}_d$ with $\idFunction \in \mathcal{M}_d'$ is a generator of $(\mathcal{M}_d, \circ, \idFunction)$ if the smallest monoid
$(\hat{\mathcal{M}}_d, \circ, \idFunction)$ with $\mathcal{M}_d' \subseteq \hat{\mathcal{M}}_d$ fulfills $\hat{\mathcal{M}}_d = \mathcal{M}_d$.
In other words, if we only have the generator $\mathcal{M}_d'$, then we can already obtain all of $\mathcal{M}_d$ by utilizing functional composition.

Monoids related to $\mathcal{M}_d$ are well-studied in the mathematical literature.
Of large importance is the monoid over the non-decreasing functions $f: \range{n} \rightarrow \range{n}$ denoted by $\mathcal{O}_n$ \cite{
    howieProductsIdempotentOrderpreserving1973,ahmadubellouniversityQuasiidempotentsFiniteSemigroup2023,gomesRanksCertainSemigroups1992,higginsCombinatorialResultsSemigroups1993,higginsIdempotentDepthSemigroups1994,howieProductsIdempotentsCertain1971}.
Observe that our essence monoid $\mathcal{M}_d$ is isomorphic to the monoid on the function set $\{f \in \mathcal{O}_{d+2} \mid f(1) = 1, f(d+2) = d+2\}$.
It is known that $\mathcal{O}_n$ it can be generated by composing idempotent functions \cite{howieProductsIdempotentsCertain1971,howieProductsIdempotentOrderpreserving1973,gomesRanksCertainSemigroups1992,higginsIdempotentDepthSemigroups1994},
or quasi-idempotent functions (and the identity function) \cite{ahmadubellouniversityQuasiidempotentsFiniteSemigroup2023}.\footnote{A function $f$ is idempotent if $f = f \circ f$, and quasi-idempotent if it is not idempotent but $f \circ f$ is.}
A monoid that is isomorphic to ours has appeared in work by Kudryavtseva and Mazorchuk \cite[Section 6.7]{kudryavtsevaPartializationCategoriesInverse2008}, who also exhibit a small generator for it.
Next, we define a small set of simple functions, that we will then show to be a generator of the essence monoid.

\defBasicFunction*{}

Now, we not only show that the set of basic functions is a generator of the essence monoid, but that we can actually obtain any function of the monoid by composing a few basic functions.
For this purpose, we first introduce a set of \emph{advanced functions}, these are still relatively simple functions which are later built from basic functions.

\begin{definition}[Advanced Functions]
    Let $d$ be a positive integer. For any $i \in \range[1]{d}$ let $\incFunction_i: \range[0]{d+1} \rightarrow \range[0]{d+1}$ be the function
    \begin{align*}
        \incFunction_i(x) = \begin{cases}
                                x+1 & \text{if $x = i$}, \\
                                x   & \text{otherwise.}
                            \end{cases}
    \end{align*}

    For any $i,j \in \range[1]{d+1}$ with $i < j$ define $\incFunction_{i,j}:  \range[0]{d+1} \rightarrow \range[0]{d+1}$ as
    \begin{align*}
        \incFunction_{i,j}(x) = \begin{cases}
                                    j & \text{if $i \leq x \leq j$}, \\
                                    x & \text{otherwise.}
                                \end{cases}
    \end{align*}

    For any $i,j \in \range[0]{d}$ with $i < j$ define $\decFunction_{i,j}:  \range[0]{d+1} \rightarrow \range[0]{d+1}$ as
    \begin{align*}
        \decFunction_{i,j}(x) = \begin{cases}
                                    i & \text{if $i \leq x \leq j$}, \\
                                    x & \text{otherwise.}
                                \end{cases}
    \end{align*}
    We call the functions $\incFunction_i$ $(i \in \range[1]{d})$, $\incFunction_{i,j}$ $(i,j \in \range[1]{d+1}, i < j)$, and $\decFunction_{i,j}$ $(i,j \in \range[0
        ]{d}, i < j)$ advanced functions (relative to $d$).
\end{definition}

The approach we use to compose the functions of the essence monoid is based on the techniques of the abstract algebra literature \cite{higginsIdempotentDepthSemigroups1994,gomesRanksCertainSemigroups1992}.
Recall that $\mathcal{O}_n$ is the monoid of non-decreasing functions over the domain $\range[1]{n}$.
In Higgins \cite{higginsCombinatorialResultsSemigroups1993} it is shown that any function $g \in \mathcal{O}_n$ can be composed from $\Oh(n^2)$ idempotents of $\mathcal{O}_n$ of defect one.
These are functions $g \in \mathcal{O}_n$ such that $g = g \circ g$ and such that there is only a single input value $x$ with $g(x) \not = x$.
This implies that for such a function $g$ with $g(x) \not = x$ we must have $g(x) \in \{x-1,x+1\}$.
Therefore, the functions $\incFunction_i$ and $\decFunction_i$ have exactly the same form as these functions.
Although Higgins gives the result for non-decreasing functions $g$ where $0, d+1$ are not necessarily fixed-points, the central techniques work fine even in our setting.
We provide formal proofs here to keep the paper self-contained and to reconfirm the correctness in our setting.

However, our set of basic functions does not include the functions $\incFunction_i$, thus we first need to compose these functions using basic functions.
For this purpose we utilize the insight Gomes and Howie \cite{gomesRanksCertainSemigroups1992} use in the proof of their Lemma 2.6 to obtain the following \cref{thm:compose_incfunction_i}.

\begin{lemma}
    \label{thm:compose_incfunction_i}
    Let $d$ be a positive integer, and $i \in \range[1]{d}$.
    Then, there are $\ell \leq d$ basic functions $g_1,\dots,g_\ell$ relative to $d$ such that $\incFunction_i = g_1 \circ g_2 \dots \circ g_\ell$.
\end{lemma}
\begin{proof}
    Set
    \begin{align*}
        h_1 & = \begin{cases}
                    \decFunction_{d} \circ \dots \circ \decFunction_{i+1} & \text{if $i + 1 \leq d$}, \\
                    \idFunction                                           & \text{otherwise,}
                \end{cases} \\
        h_2 & = \incFunction,                                                                     \\
        h_3 & = \begin{cases}
                    \decFunction_{i} \circ \dots \circ \decFunction_2 & \text{if $i \geq 2$}, \\
                    \idFunction                                       & \text{otherwise.}
                \end{cases}
    \end{align*}
    Observe that $h_1$ decreases all inputs $x \in \range[i+1]{d}$ by exactly one, $h_2$ increases all inputs between $1$ and $d$, and $h_3$ decreases all $x \in \range[2]{i}$ by exactly one.
    All other inputs are mapped to themselves by these functions.
    Set $h^* = h_3 \circ h_2 \circ h_1$.

    Fix any $x \in \range[0]{d+1}$.
    If $x \in \{0,d+1\}$, then $h^*(x) = x$ as desired since $0,d+1$ are fixed-points of all composed functions.

    If $1 \leq x \leq i-1$ we have that $h_1(x) = x$, $(h_2 \circ h_1)(x) = x+1$ and $h^*(x) = x$.

    If $x = i$ we have $h_1(x) = i$, $(h_2 \circ h_1)(x) = i+1 = h^*(x)$.

    If $i +1 \leq x \leq d$ we have $h_1(x) = x-1$, $(h_2 \circ h_1)(x) = x = h^*(x)$.
    Thus, $\incFunction_i = h^*$.
\end{proof}

Now, we continue with the approach of Higgins \cite{higginsIdempotentDepthSemigroups1994} which forms the backbone of the next observations, corollaries and lemmas up until \cref{thm:function_composition_alg}.
First, we can observe that it is easy to compose the functions $\incFunction_{i,j}$ and $\decFunction_{i,j}$ using at most $d$ functions which are either basic, or the function $\incFunction_{i'}$ for some values of $i'$.

\begin{observation}
    \label{obs:compose:inc_i_j}
    Let $d$ be a positive integer.
    For any $i,j \in \range[1]{d+1}$ with $i < j$
    we have $\incFunction_{i,j} = \incFunction_{j-1} \circ \dots \circ \incFunction_{i}$.
    For any $i,j \in \range[0]{d}$ with $i < j$ we have $\decFunction_{i,j} = \decFunction_{i+1} \circ \dots \circ \decFunction_j$.
\end{observation}

We can summarize the results about the advanced functions in the following simple algorithmic corollary.
\begin{corollary}
    \label{corollary:compose_advanced_functions}
    Let $d$ be a positive integer.
    There is an algorithm that takes an advanced function $f: \range[0]{d+1} \rightarrow \range[0]{d+1}$ relative to $d$ as input and outputs at most $d^2$ basic functions $g_1,\dots,g_\ell$ relative to $d$ such that $f = g_1 \circ g_2 \dots \circ g_\ell$.
    The algorithm runs in time polynomial in $d$.
\end{corollary}

We first focus on functions $f$ which map any input $x$ to a value that is at least $x$.
For a function $f$ let $\fixedPoints(f)$ be the number of fixed-points of $f$.

\begin{lemma}
    \label{lemma:compose_plus_function}
    Let $d$ be a positive integer.
    There is an algorithm that takes a non-decreasing function $f: \range[0]{d+1} \rightarrow \range[0]{d+1}$ with
    \begin{enumerate}
        \item $f \not = \idFunction$, and
        \item $f(x) \geq x$ for all $x \in \range[0]{d+1}$, and
        \item $f(0) = 0, f(d+1) = d + 1$,
    \end{enumerate}
    as input and outputs $\ell \leq d+2 - \fixedPoints(f)$ advanced functions $g_1,\dots,g_\ell$ relative to $d$ such that $f = g_1 \circ \dots \circ g_\ell$.
    The algorithm runs in time polynomial in $d$.
\end{lemma}
\begin{proof}
    Let $\mathcal{X}$ be a partition of $\range[0]{d+1}$ such that for all $x,y \in \range[0]{d+1}$ we have that $x,y$ are in the same partition of $\mathcal{X}$ if and only if $f(x) = f(y)$.
    For $X \subseteq \range[0]{d+1}$ we define $f(X) = y$ if $f(x) = y$ for all $x \in X$.

    Now, remove those sets $X$ from $\mathcal{X}$ where $\min X = f(X)$ to create $\mathcal{X'}$, that is, $\mathcal{X}' = \{X \in \mathcal{X} \mid \min X \neq f(X)\}$.
    We will briefly argue that $|\mathcal{X}'| \neq \emptyset$.
    For this purpose, consider some $X \in \mathcal{X} \setminus \mathcal{X}'$, that is, some set $X \in \mathcal{X}$ with $f(X) = \min X$.
    We argue that we must have $|X| = 1$.
    Otherwise, there is some $x \in X \setminus \{\min X\}$.
    However, then $f(x) = f(X) = \min X < x$, which contradicts that $f(y) \geq y$ for all domain elements $y$.
    Therefore, we have $|X| = 1$.
    Now, if all sets of $\mathcal{X}$ have size one, then we must have $f = \idFunction$.
    As this is not the case, some set of $\mathcal{X}$ has size at least two, and is part of $\mathcal{X}'$.

    Next, we argue that we do not have $\min X = 0$ for any $X \in \mathcal{X}'$.
    Assume otherwise.
    Then, from $f(0) = 0$ we would obtain $\min X = f(X) = 0$, contradicting that $X \in \mathcal{X}'$.

    Now, order $\mathcal{X}'$ by decreasing values of $f$, that is, $\mathcal{X}' = \{X_1,\dots,X_t\}$ such that $f(X_i) > f(X_j)$ for all $i < j$.
    Set $h = \incFunction_{\min X_t,f(X_t)} \circ \dots \circ \incFunction_{\min X_1,f(X_1)}$, and observe that this function is well-defined by the arguments of the previous paragraph.

    Fix some $x \in \range[0]{d+1}$.
    Function $\incFunction_{j_1,j_2}$ only changes the input $x$ if $x \in \range[j_1]{j_2 - 1}$.
    Using our choice of $h$ this would imply that for some set $X_j$ we have that $\min X_j \leq x < f(X_j)$ if the function $ \incFunction_{\min X_j,f(X_j)}$ used for $X_j$ in the composition of $h$ changes the input value $x$.

    We will show that $h(x) = f(x)$.
    First, consider that $f(x) > x$, then $x \in X_i$ for some $X_i \in \mathcal{X}'$ and $f(x) = f(X_i)$.
    Consider an arbitrary set $X_j$ with $j < i$.
    Since we have sorted $\mathcal{X}'$ by decreasing values of $f$, we have that $f(X_j) > f(X_i) = f(x)$.
    Moreover, since $f$ is non-decreasing, we obtain that $\min X_j > x$.
    Therefore, the function $\incFunction_{\min X_j,f(X_j)}$ used for $X_j$ in the function composition does not change the input value $x$.
    So, \[(\incFunction_{\min X_i, f(X_i)} \circ \dots \circ \incFunction_{\min X_1, f(X_1)})(x) = \incFunction_{\min X_i,f(X_i)}(x).\]
    Because $f(y) \geq y$ for all domain elements, we have that $\min X_i \leq x \leq f(X_i)$, which yields $\incFunction_{\min X_i,f(X_i)}(x) = f(X_i) = f(x)$.
    Now, consider some set $X_j$ with $j > i$.
    Then, $f(X_j) < f(X_i)$, and therefore the function $\incFunction_{\min X_j, f(X_j)}$ does not change the input value $f(X_i)$.
    This yields $h(x) = (\incFunction_{\min X_t, f(X_t)} \circ \dots \circ \incFunction_{\min X_i, f(X_i)})(x) = f(x)$.

    Now, consider that $f(x) = x$.
    Then, there is no $X_j \in \mathcal{X}'$ such that the corresponding function $\incFunction_{\min X_j, f(X_j)}$ changes the input value $x$.
    Indeed, this would imply that $\min X_j \leq x = f(x) < f(X_j) = f(\min X_j)$.
    Then, we would have $x \geq \min X_j$ but $f(x) < f(\min X_j)$, which contradicts that $f$ is non-decreasing.
    Thus, $h(x) = f(x)$ as none of the composed functions change the value $x$.

    Finally, we need to argue that we compose at most $d+2 - \fixedPoints(f)$ functions to create $h$.
    To see this, we argue that each fixed-point $x$ of $f$ is either not part of any set of $\mathcal{X}'$, or is part of a set of $\mathcal{X}'$ that contains an additional element $x' \in \range[0]{d+1}$ which is not a fixed-point.
    Assume that there is a fixed-point $x \in X$ for some $X \in \mathcal{X}'$.
    Then, $X$ cannot contain another fixed-point other than $x$ since all elements of $X$ have the same image.
    Moreover, if $\abs{X} = 1$, then we would find that $\min X = x = f(X)$, contradicting that $X \in \mathcal{X}'$.
    Therefore, we have that $X$ contains some further integer $x'$ which is not a fixed-point.
    Thus, any set of $\mathcal{X}'$ contains an integer of $\range[0]{d+1}$ that is not a fixed-point, showing that the size of $\mathcal{X}'$ is at most $d+2 - \fixedPoints(f)$.
\end{proof}

We can obtain an analogous lemma for functions $f$ such that $f(x) \leq x$ for all domain-elements $x$.
The proof is symmetric to the previous one and therefore mostly identical, but we still provide it here for completeness.

\begin{lemma}
    \label{lemma:compose_minus_function}
    Let $d$ be a positive integer.
    There is an algorithm that takes a non-decreasing function $f: \range[0]{d+1} \rightarrow \range[0]{d+1}$ with
    \begin{enumerate}
        \item $f \not = \idFunction$, and
        \item $f(x) \leq x$ for all $x \in \range[0]{d+1}$, and
        \item $f(0) = 0, f(d+1) = d + 1$,
    \end{enumerate}
    as input and outputs $\ell \leq d+2 - \fixedPoints(f)$ advanced functions $g_1,\dots,g_\ell$ relative to $d$ such that $f = g_1 \circ \dots \circ g_\ell$.
    The algorithm runs in time polynomial in $d$.
\end{lemma}
\begin{proof}
    Let $\mathcal{X}$ be a partition of $\range[0]{d+1}$ such that for all $x,y \in \range[0]{d+1}$ we have that $x,y$ are in the same partition of $\mathcal{X}$ if and only if $f(x) = f(y)$.
    For $X \subseteq \range[0]{d+1}$ we define $f(X) = y$ if $f(x) = y$ for all $x \in X$.

    Now, remove those sets $X$ from $\mathcal{X}$ where $\max X = f(X)$ to create $\mathcal{X'}$, that is, $\mathcal{X}' = \{X \in \mathcal{X} \mid \max X \neq f(X)\}$.
    We will briefly argue that $|\mathcal{X}'| \neq \emptyset$.
    For this purpose, consider some $X \in \mathcal{X} \setminus \mathcal{X}'$, that is, some set $X \in \mathcal{X}$ with $f(X) = \max X$.
    We argue that we must have $|X| = 1$.
    Otherwise, there is some $x \in X \setminus \{\max X\}$.
    However, then $f(x) = f(X) = \max X > x$, which contradicts that $f(y) \leq y$ for all domain elements $y$.
    Therefore, we have $|X| = 1$.
    Now, if all sets of $\mathcal{X}$ have size one, then we must have $f = \idFunction$.
    As this is not the case, some set of $\mathcal{X}$ has size at least two, and is part of $\mathcal{X}'$.

    Next, we argue that we do not have $\max X = d+1$ for any $X \in \mathcal{X}'$.
    Assume otherwise.
    Then, from $f(d+1) = d+1$ we would obtain $\max X = f(X) = d+1$, contradicting that $X \in \mathcal{X}'$.

    Now, order $\mathcal{X}'$ by increasing values of $f$, that is, $\mathcal{X}' = \{X_1,\dots,X_t\}$ such that $f(X_i) < f(X_j)$ for all $i < j$.
    Set $h = \decFunction_{f(X_t),\max X_t} \circ \dots \circ \decFunction_{f(X_1), \max X_1}$, and observe that this function is well-defined by the arguments of the previous paragraph.

    Fix some $x \in \range[0]{d+1}$.
    Function $\decFunction_{j_1,j_2}$ only changes the input $x$ if $x \in \range[j_1+1]{j_2}$.
    Using our choice of $h$ this would imply that for some set $X_j$ we have that $f(X_j) < x \leq \max X_j$ if the function $\decFunction_{f(X_j), \max X_j}$ used for $X_j$ in the composition of $h$ changes the input value $x$.

    We will show that $h(x) = f(x)$.
    First, consider that $f(x) < x$, then $x \in X_i$ for some $X_i \in \mathcal{X}'$ and $f(x) = f(X_i)$.
    Consider an arbitrary set $X_j$ with $j < i$.
    Since we have sorted $\mathcal{X}'$ by increasing values of $f$, we have that $f(X_j) < f(X_i) = f(x)$.
    Moreover, since $f$ is non-decreasing, we obtain that $\max X_j < x$.
    Therefore, the function $\decFunction_{f(X_j),\max X_j}$ used for $X_j$ in the function composition does not change the input value $x$.
    So, \[(\decFunction_{f(X_i),\max X_i} \circ \dots \circ \decFunction_{f(X_1), \max X_1})(x) = \decFunction_{f(X_i),\max X_i}(x).\]
    Because $f(y) \leq y$ for all domain elements, we have that $f(X_i) \leq x \leq \max X_i$, which yields $\decFunction_{f(X_i),\max X_i}(x) = f(X_i) = f(x)$.
    Now, consider some set $X_j$ with $j > i$.
    Then, $f(X_j) > f(X_i)$, and therefore the function $\decFunction_{f(X_j), \max X_j}$ does not change the input value $f(X_i)$.
    This yields $h(x) = (\decFunction_{f(X_t), \max X_t} \circ \dots \circ \decFunction_{f(X_i),\max X_i})(x) = f(x)$.

    Now, consider that $f(x) = x$.
    Then, there is no $X_j \in \mathcal{X}'$ such that the corresponding function $\decFunction_{f(X_j),\max X_j}$ changes the input value $x$.
    Indeed, this would imply that $f(\max X_j) = f(X_j) < x = f(x) \leq \max X_j$.
    Then, we would have $x \leq \max X_j$ and $f(x) > f(\max X_j)$, which contradicts that $f$ is non-decreasing.
    Thus, $h(x) = f(x)$ as none of the composed functions change the value $x$.

    Finally, we need to argue that we compose at most $d+2 - \fixedPoints(f)$ functions to create $h$.
    To see this, we argue that each fixed-point $x$ of $f$ is either not part of any set of $\mathcal{X}'$, or is part of a set of $\mathcal{X}'$ that contains an additional element $x' \in \range[0]{d+1}$ which is not a fixed-point.
    Assume that there is a fixed-point $x \in X$ for some $X \in \mathcal{X}'$.
    Then, $X$ cannot contain another fixed-point other than $x$ since all elements of $X$ have the same image.
    Moreover, if $\abs{X} = 1$, then we would find that $\max X = x = f(X)$, contradicting that $X \in \mathcal{X}'$.
    Therefore, we have that $X$ contains some further integer $x'$ which is not a fixed-point.
    Thus, any set of $\mathcal{X}'$ contains an integer of $\range[0]{d+1}$ that is not a fixed-point, showing that the size of $\mathcal{X}'$ is at most $d + 2 - \fixedPoints(f)$.
\end{proof}

Now, we can use \cref{lemma:compose_minus_function,lemma:compose_plus_function} to obtain the main result about composing functions of the essence monoid.

\thmFunctionCompositionAlg*{}
\begin{proof}
    If $f = \idFunction$ we output the basic function $\idFunction$.

    Otherwise, we create the following two functions $f^+$ and $f^-$ from $f$:
    \begin{align*}
        f^+(x) = \begin{cases}
                     f(x) & \text{if $f(x) > x$}, \\
                     x    & \text{otherwise.}
                 \end{cases} & \qquad & f^-(x) = \begin{cases}
                                                     f(x) & \text{if $f(x) < x$}, \\
                                                     x    & \text{otherwise.}
                                                 \end{cases}
    \end{align*}

    We argue that $f = f^+ \circ f^-$.
    If $f(x) > x$ this is obvious since $f^-(x) = x$.
    When $f(x) < x$, notice that we must have $f(f(x)) \leq f(x)$ since $f$ is non-decreasing.
    Therefore, $f^+(f^-(x)) = f^+(f(x)) = f(x)$.

    We can apply \cref{lemma:compose_plus_function,lemma:compose_minus_function} to obtain function compositions of $f^+$ and $f^-$ in terms of advanced functions, if one of $f^+$ or $f^-$ is $\idFunction$ we can simply omit the function in the composition.
    Since each $x \in \range[1]{d}$ is a fixed-point of $f^+$ or of $f^-$, and additionally $0,d+1$ are fixed-points of both $f^+$ and $f^-$ we find that the total number of used advanced functions is at most $2(d+2) - \fixedPoints(f^+) - \fixedPoints(f^-) \leq d$.
    By \cref{corollary:compose_advanced_functions} we can compose any advanced function from at most $d^2$ basic functions, which concludes the proof.
\end{proof}

Next, we want to create the connection between these functions and the caterpillars we use.
Concretely, we first show that for any $d$ and any basic function (relative to $d$), there is a simple caterpillar whose essence is the basic function.

\thmBasicFunctionCaterpillarAlg*{}
\begin{proof}
    There are three different types of basic functions, $\idFunction$, $\incFunction$, and $\decFunction_i$ (for some $i \in \range[1]{d}$).

    \proofsubparagraph*{The algorithm.}
    Let $f$ be the input basic function.
    Moreover, for any integers $i,j$ with $i \leq j$, let $P^i_j$ be the caterpillar created by taking a path $v_1,\dots,v_j$ on $j$ vertices, and adding a pendant to $v_i$.
    The spine of $P_i$ is $v_1,\dots,v_j$.
    The algorithm proceeds in three cases depending on $f$.
    \begin{itemize}
        \item If $f = \idFunction$, the algorithm simply outputs the path on $d+1$ vertices, the spine of the caterpillar is the whole path.
        \item If $f = \incFunction$, the algorithm outputs the caterpillar $P^{d+1}_{d+1}$.
        \item If $f = \decFunction_i$ for some $i \in \range[1]{d}$, then the algorithm outputs the caterpillar $P^{d-i+1}_d$.
    \end{itemize}
    It is easy to see that one can compute these caterpillars in time linear in $d$, and also check which of the basic functions $f$ is in time linear in $d$.
    Hence, the algorithm has linear running time in $d$ overall.

    \proofsubparagraph*{Correctness of the algorithm.}

    It remains to proof that the caterpillars the algorithm outputs actually have the correct essence.
    For any essence $\gamma$ we have $\gamma(d+1) = d+1$ by definition, so we must only show that the essences match the function $f$ for inputs between $0$ and $d$.

    \begin{description}
        \item[$f = \idFunction$:] Then, the algorithm output is a path $v_1,\dots,v_{d+1}$ on $d+1$ vertices.
              Observe that the solution-tight packing of $P$ contains a single graph containing all vertices of $P$.
              Let $\gamma$ be the essence of $P$.
              We can select vertex $v_{d+1}$ to show that $\gamma(0) = 0 = f(0)$.
              If we consider the graph $P_d$, which is $P$ in which $v_1$ has $d$ additional pendants, then selecting $v_1$ shows that $\gamma(d) = d$.

              Let $x \in \range[1]{d-1}$.
              Consider the graph $P_x$ obtained by adding $x$ additional pendants to the leftmost spine vertex of $P$.
              Then, $S = \{v_{d+1 - x}\}$ is a \dcocset{} of $P$, and there is no \dcocset{} of size one that selects a vertex $v_j$ with $j > d+1 - x$.
              Indeed, $P - S$ contains a connected component containing $v_1,\dots,v_{d - x}$, and $x$ pendants of $v_1$, and this components has size exactly $d$.
              This also shows that there cannot be a \dcocset{} of size at most one containing a vertex $v_j$ with $j > d + 1 - x$.
              Furthermore, $ P - S$ contains a connected component containing $v_{d - x + 2}, \dots, v_{d+1}$, this component has size exactly $(d + 1) - (d - x + 2) + 1 = x$, showing that  $\gamma(x) = x = f(x)$.

        \item[$f = \incFunction$:] Then, the algortihm outputs $P = P^{d+1}_{d+1}$ with spine $v_1,\dots,v_{d+1}$.
              Once again, the solution-tight packing of $P$ has size exactly one and selects all vertices of $P$.
              Let $\gamma$ be the essence of $P$.
              We then have $\gamma(0) = 0 = f(0)$, because $\{v_{d+1}\}$ is a \dcocset{} of $P$.

              Consider graph $P_d$ obtained by adding $d$ pendants to $v_1$.
              We show that $\opt{P_d} > 1$.
              Any minimum \dcocset{} $S$ can be assumed to select $v_1$ due to the $d$ pendants of $v_1$.
              However, if $S$ selects no other vertex, then $v_2,\dots,v_{d+1}$ and the pendant of $v_{d+1}$ form a connected component of size $d+1$, so $S$ would not be a \dcocset{} at all.
              As $\opt{P_d} > \opt{P}$ we have $\gamma(d) = d+1 = f(d)$, as desired.

              Now, consider any $x \in \range[1]{d-1}$.
              Then, $S = \{v_{d+1 - x}\}$ is a \dcocset{} of $P_x$ of size one.
              Indeed, $P_x - S$ contains a connected component containing $v_1,\dots,v_{d-x}$ and $x$ pendants of $v_1$, this component has size exactly $d$.
              This also shows that there is no \dcocset{} of size at most one selecting a vertex $v_j$ with $j > d + 1 - x$.
              The other connected component of $P_x - S$ contains $v_{d - x + 2},\dots,v_{d+1}$, as well as the additional pendant of $v_{d+1}$.
              This connected component then has size exactly $x + 1 = \gamma(x) = f(x)$.

        \item[$f = \decFunction_i$ for some $i \in \range{d}$:]
              In this case the output was $P = P^{d-i + 1}_d$ with spine $v_1,\dots,v_d$.
              Again, it is not difficult to see that the solution-tight packing of $P$ contains all its vertices in a single graph.
              For any $x \geq 0$ the graph $P_x$ is the same as $P$ where $v_1$ get $x$ additional pendants.
              Let $\gamma$ be the essence of $P$.
              Then, we can select $v_d$ to see that $\gamma(0) = 0$.

              \proofsubparagraph*{$i = d$:}
              Next, we want to cover the case where $i = d$, in which case $d - i + 1 = 1$, and so $v_1$ has a pendant in graph $P$.
              Then, we can select $v_1$ of the graph $P_d$ to show that $\gamma(d) = d-1$, indeed, $P_d - \{v_1\}$ only contains connected components of size one, namely the pendants of $v_1$, as well as the connected component containing $v_2,\dots,v_d$ of size $d-1$.
              For $x = d-1$, we also select $v_1$.
              Since this results in essentially the same connected components as in the case where $x = d$, and $f(d-1) = d-1$, it remains to argue that there is no solution of size one that selects $v_2$.
              And indeed, $v_1$ is in a connected component of size $d+1$ in $P_{d-1} - \{v_2\}$, as $v_1$ has $d$ pendants in $P_{d-1}$.
              Hence, $\gamma(d-1) = f(d-1)$.
              For any $x \in \range[1]{d-2}$, we can select the vertex $v_{d - x}$ of $P_x$, observe that $v_{d - x} \not = v_1$.
              Then, $P_x - \{v_{d - x}\}$ has a connected component containing $v_1,\dots,v_{d -x  -1}$, as well as the $x + 1$ pendants of $v_1$.
              This component has size $d$, so there is no minimum \dcocset{} of $P_x$ selecting a vertex $v_j$ with $j > d - x$.
              Furthermore, $P_x - \{v_{d-  x}\}$ also has a connected component containing $v_{d - x + 1},\dots,v_d$ of size $x = \gamma(x) = f(x)$.

              \proofsubparagraph*{$i < d$:}
              Now, we consider the cases where $i < d$.
              We have already established that $\gamma(0) = 0$.
              Now, consider an arbitrary $x \in \range{i-2}$.
              Then, set $S = \{v_{d - x}\}$, and observe that $d-x > 1$.
              The graph $P_x - S$ contains a connected component containing the vertices $v_1,\dots,v_{d - x - 1}$, and, as $d - x - 1 \geq d - i + 1$, also the additional pendant of $v_{d - i +1}$, as well as the $x$ further pendants of $v_1$.
              This connected component then has size exactly $d$, which also shows there is no solution for $P_x$ of size at most one selecting a vertex $v_j$ with $j > d - x -1$.
              Furthermore, $P_x - S$ contains a connected component containing $v_{d - x + 1},\dots,v_{d}$ of size exactly $x = \gamma(x) = f(x)$.

              Next, consider that $x = i - 1$.
              In this case, we again set $S = \{v_{d-x}\} = \{v_{d - i + 1}\}$.
              Again, observe that $d - x > 1$ as $i < d$.
              Hence, there is a connected component in $P_x - S$. containing the vertices $v_1,\dots,v_{d - i}$, as well as the $i-1$ additional pendants of $v_1$.
              This connected component has size $d - 1$.
              So, we need to argue that there no solution $S'$ selecting a vertex $v_j$ with $j > d - i + 1$.
              And indeed, $S' = \{v_{d - i + 2}\}$ is not a \dcocset{} of $P_x$, as $P_x - S'$ contains a connected component containing $v_1,\dots,v_{d - i +1}$, as well as the $i - 1$ additional pendants on $v_1$, and the additional pendant of $v_{d - i + 1}$, yielding a connected component of size $d+1$.
              Graph $P_x - S$ also contains a connected component of size exactly one, namely the one containing the pendant of $v_{d - i +1}$.
              Finally, $P_x - S$ contains a connected component containing $v_{d - x + 1},\dots,v_d$, and hence $v_d$ is in a connected component of size exactly $x = \gamma(x) = f(x)$.

              Now, consider that $x = i$.
              In this case, we can again set $S = \{v_{d - x + 1}\}$, which means that we select exactly the same vertex as in the previous case, namely vertex $v_{d - i + 1} \not = v_1$.
              It is clear from the elaboration for the previous case that there cannot be a solution for $P_x$ that selects a vertex $v_j$ with $j > d- x + 1$.
              Hence, it remains to argue that $S$ is actually a \dcocset{} of $P_x$.
              And indeed, $P_x$ contains a connected component containing $v_1,\dots,v_{d - i}$, as well as the $x = i$ additional pendants of $v_1$, and that component has size exactly $d$.
              The other connected components are exactly like in the case when we had $x = i-1$, so, we in particular have that $v_d$ is in a connected component of size $i - 1 = x - 1 = \gamma(x) = f(x)$.

              Next, consider the case that $x \in \range[i+1]{d-1}$.
              Then, we set $S = \{v_{d - x + 1}\}$.
              Observe once again that we have $d - x + 1 > 1$.
              In this case, $P_x - S$ has a connected component containing $v_1,\dots,v_{d - x}$, as well as $x$ additional pendants of $v_1$.
              This connected component has size exactly $d$, which also shows that there is no minimum \dcocset{} selecting $v_j$ for any $j > d - x + 1$.
              Moreover, $P_x - S$ has a connected component containing $v_{d - x + 2},\dots,v_d$, and, as $d - x + 2 \leq d - i + 1$, also the additional pendant of $v_{d - i + 1}$.
              Hence, $v_d$ is in a connected component of size exactly $x = \gamma(x) = f(x)$.

              Finally, consider that $x = d$.
              Then, we set $S = \{v_1\}$.
              In this case $P_{d+1} - S$ has connected component of size one (containing the pendants of $v_1$), as well as a connected component containing $v_2,\dots,v_d$, as well as the additional pendant of $v_{d - i + 1}$.
              This connected component has size exactly $d = \gamma(d) = f(d)$.
              It remains to argue that there is no solution of size one selecting $v_2$.
              And indeed, $v_1$ is in a connected component of size $d+1$ in $P_d - \{v_2\}$, as $v_1$ has $d$ pendants in that graph.
    \end{description}
\end{proof}

Now, we show that attaching two caterpillar together by their spine vertices results in a caterpillar whose essence is exactly the function composition of the essences of the individual caterpillars.

\thmCaterpillarSignatureFunctionComposition*{}
\begin{proof}
    Observe that we obtain the solution-tight packing for $C$ by simply taking the union of the solution-tight packings of $C_1$ and $C_2$.
    This directly yields that $\opt{C} = \opt{C_1} + \opt{C_2}$.
    Furthermore, we have $\gamma(d+1) = d+1 = \gamma_1(d+1) = \gamma_2(d+1) = \gamma_2 \circ \gamma_1 (d+1)$, because any essence maps $d+1$ to $d+1$.
    Let $x \in \range[0]{d}$ be an integer, we show that $\gamma(x) = \gamma_2(\gamma_1(x))$ by showing $\gamma(x) \leq \gamma_2(\gamma_1(x))$ and $\gamma(x) \geq \gamma_2(\gamma_1(x))$ in two distinct steps.
    Recall that $(C_i)_x$ (where $i \in \{1,2\}$) is the graph $C_i$ where the leftmost spine vertex of $C_i$ gets an additional $x$ pendants.

    \proofsubparagraph*{$\gamma(x) \leq \gamma_2(\gamma_1(x))$:}
    If $\gamma_1(x) = d+1$, then we also have that $\gamma_2(\gamma_1(x)) = d+1$.
    Moreover, if $\gamma_2(\gamma_1(x)) = d+1$, then we clearly have that $\gamma(x) \leq \gamma_2(\gamma_1(x))$.

    So, now consider that $\gamma_1(x) \leq d$ and $\gamma_2(\gamma_1(x)) \leq d$.
    If $\gamma_1(x) = 0$, let $S_1$ be the minimum solution of $(C_1)_x$ that selects $v_n$, if $\gamma_1(x) > 0$, let $S_1$ be the solution of $(C_1)_x$ in which $v_n$ is in a connected component of $(C_1)_x - S_1$ of size $\gamma_1(x)$.
    Similarly, if $\gamma_2(\gamma_1(x)) = 0$, let $S_2$ be the minimum solution of $(C_2)_{\gamma_1(x)}$ which selects $w_m$, otherwise let $S_2$ be the solution of $(C_2)_{\gamma_1(x)}$ in which $w_m$ is in a connected component of $(C_2)_{\gamma_1(x)} - S_2$ of size $\gamma_2(\gamma_1(x))$.
    We also have that $|S_1| = \opt{C_1}$ and $|S_2| = \opt{C_2}$.
    Now, we claim that $S = S_1 \cup S_2$ is a solution for $C_x$ such that if $\gamma_2(\gamma_1(x)) = 0$, then $w_m$ is selected, and when $\gamma_2(\gamma_1(x)) > 0$, then $w_m$ is in a connected component of size $\gamma_2(\gamma_1(x))$ in $C_x - S$.

    First, we show that $S$ is a \dcocset{} of $C_x$.
    We can assume that $S_1$ and $S_2$ only contain spine vertices of $C_1$ and $C_2$.
    Hence, $S_1$ is a solution of $(C_1)_x$, and $S_2$ is a solution of $C_2$.
    Thus, any connected component of $C_x - S$ that is too large would have to contain both vertices of $(C_1)_x$ and $C_2$.
    Consider a connected component $Z$ of $C_x - S$ containing both vertices of $C_1$ and $C_2$.
    Then, this connected component contains exactly $\gamma_1(x)$ vertices of $C_1$, as it must contain vertex $v_n$.
    Furthermore, $Z$ must contain $w_1$.
    However, as $S_2$ is actually solution for the graph $(C_2)_{\gamma_1(x)}$, $Z$ can contain at most $d - \gamma_1(x)$ vertices of $C_2$, so component $Z$ has size at most $d$ overall.
    It follows that $S$ is a \dcocset{} of $C_x$.

    Vertex $w_m$ is selected only if $\gamma_2(\gamma_1(x)) = 0$, which shows that we also have $\gamma(x) = 0$ in that case.
    Otherwise, $w_m$ is not selected, and $w_m$ is in a connected component of $C_x - S$.
    As we can assume that $S_2$ selects a spine vertex of $C_2$, this connected component is exactly the same connected component that contains $w_m$ in the graph $(C_2)_{\gamma_1(x)} - S_2$, and that component has size $\gamma_2(\gamma_1(x))$.
    So, we have that $\gamma(x) \leq \gamma_2(\gamma_1(x))$ as evidenced by solution $S$.

    \proofsubparagraph*{$\gamma(x) \geq \gamma_2(\gamma_1(x))$:}
    Clearly, if $\gamma(x) = d+1$, we have $\gamma(x) \geq \gamma_2(\gamma_1(x))$, so we shall assume that $\gamma(x) \leq d$ in the rest of the proof.
    Then, there is a solution $S$ with $|S| = \opt{C_1} + \opt{C_2}$ for $C_x$ such that either $S$ selects $w_m$ if $\gamma(x) = 0$, or otherwise such that $w_m$ is in a connected component of $C_x - S$ of size $\gamma(x)$.

    We have that $S \cap V(C_1)$ is a solution for $(C_1)_x$ of size $\opt{C_1}$, because $(C_1)_x$ and $C_2$ are subgraphs of $C_x$, and the size of $S$ is $\opt{C_1} + \opt{C_2}$.
    Hence, $S$ can select $v_n$ only if $\gamma_1(x) = 0$, and otherwise, $v_n$ is in a connected component of $C_x - S$ of size at least $\gamma_1(x)$.

    Now, we claim that it must be the case that $S \cap V(C_2)$ is a solution for $(C_2)_{\gamma_1(x)}$ of size $\opt{C_2}$.
    That $S \cap V(C_2)$ has size $\opt{C_2}$ follows directly from the fact that $|S| = \opt{C_1} + \opt{C_2}$.
    So, all that remains is showing that $S \cap V(C_2)$ is actually a solution for $(C_2)_{\gamma_1(x)}$.
    We show this property by contradiction.
    Assume that $S \cap V(C_2)$ is not a solution for $(C_2)_{\gamma_1(x)}$.
    Then, the connected component $Z$ of $(C_2)_{\gamma_1(x)}$ that contains $w_1$ has size larger than $d$.
    However, in $C_x - S$ there is also a connected component having the same size as $Z$, because the connected component of $w_1$ contains $\gamma_1(x)$ vertices of $C_1$, and these are exactly as many as the additional number of pendants that $w_1$ has in $(C_2)_{\gamma_1(x)}$ compared to $C_2$.
    Hence, $S$ would not be a solution of $C_x$ at all.

    So, we in fact have that $S \cap V(C_2)$ is an optimal solution for $(C_2)_{\gamma_1(x)}$, and the connected component of $C_2 - (S \cap V(C_2))$ that contains $w_m$ has size $\gamma(x)$.
    This shows that $\gamma_2(\gamma_1(x)) \leq \gamma(x)$, as desired.
    We conclude that $\gamma(x) = \gamma_2(\gamma_1(x))$ for all $x \in \range[0]{d+1}$.
\end{proof}

Now, we can easily compute a caterpillar with a specific essence.

\thmComputeCaterpillarEssence*{}
\begin{proof}
    We first utilize \cref{thm:function_composition_alg} to find basic functions $g_1,g_2,\dots,g_\ell$ such that $\gamma = g_\ell \circ g_{\ell-1} \circ \dots \circ g_1$, where $\ell \leq d^3$.

    Then, we use \cref{thm:basic_function_caterpillar_alg} to compute the caterpillar $C_i$ with essence $g_i$ for any $i$.
    Finally, we take the disjoint union of all of these caterpillars, and connect the rightmost spine vertex of $C_i$ to the leftmost spine vertex of $C_{i+1}$ for all $i \in \range{\ell-1}$.
    The spine of the resulting caterpillar is the path starting with the leftmost spine vertex of $C_1$ ending with the rightmost spine vertex of $C_\ell$.

    It is not difficult to confirm that this can be done in time polynomial in $d$, as the algorithms of \cref{thm:function_composition_alg,thm:basic_function_caterpillar_alg} also run in time polynomial in $d$.
    Finally, the resulting caterpillar $C$ indeed has signature $\gamma$ by \cref{thm:caterpillar_signatures_function_composition}.
\end{proof} 
\section{A Kernelization Dichotomy for Small Degree Modulators}
\label{sec:kernels_small_degree_mod}
We now shift our attention to the problem \dcocModToG{} where $\mathcal{G}$ is the class of graphs with maximum degree at most $c$, for some fixed constant $c$.
A modulator to $\mathcal{G}$ is called a degree-$c$ modulator.
For \vc{}, polynomial-kernels are known when $c \leq 2$, and are unlikely to exist for $c \geq 3$ \cite{majumdarPolynomialKernelsVertex2018}.

Any graph with maximum degree at most one only contains connected components of size at most two.
Hence, a degree-$1$ modulator is a \dcocset{} for any $d \geq 2$.
This means that when $d \geq 2$, the kernels for solution-size parameterization already do the trick.
When $d = 1$, the problem \dcoc{}  is \vc{}, and here it is already known that a polynomial-kernel exists when parameterizing by a degree-$c$ modulator for $c \leq 1$ \cite{jansenVertexCoverKernelization2013b,majumdarPolynomialKernelsVertex2018}.
Another argument would be that a degree-$c$ modulator is also a modulator to the class of caterpillars when $c \leq 1$, so our \cref{main_thm:poly_kernel_caterpillar_mod} can be directly employed.

To handle the case of degree-$c$ modulators when $c \geq 3$,
we show that \dcoc{}  is \NP{}-complete on graphs with maximum degree at most three for any value of $d \geq 1$, and thus any kernels parameterized by a degree-$c$ modulator are implausible.

So, especially the case $c = 2$ is of interest.
A graph with maximum degree two only contains connected components which are either paths (and therefore caterpillars), or cycles.
When $d \leq 3$ any \dcocset{} only leaves connected components that have maximum degree at most two, and hence, the size of a minimum degree-$2$ modulator is at most the size of a minumum \dcocset{}.
When $d \geq 4$, the parameters solution size and the size of a smallest degree-$2$ modulator are generally incomparable.

We now show that the problem \dcocModToG{} admits a polynomial kernel for all $d \geq 1$ when $\mathcal{G}$ is the class of graphs with maximum degree at most two.
In fact, we show that even the problem \cocModToGPlusD{} admits a polynomial kernel, and we slightly generalize the results by actually defining $\mathcal{G}$ as the class of graphs whose connected components are cycles or caterpillars.
As a first step, we show that graphs of $\mathcal{G}$ have bounded minimal $d$-blocking set size.
Note that a statement bounding the size of minimal $1$-blocking sets on cycles was (independently) proven in \cite[Theorem 2]{majumdarPolynomialKernelsVertex2018}.

\begin{lemma}
    \label{thm:graphs_with_max_degree_two_bounded_blocking_sets}
    Let $G$ be a graph whose connected components are caterpillars or cycles, and $d \geq 1$ be an integer.
    Then, any minimal $d$-blocking set of $G$ has size at most three.
\end{lemma}
\begin{proof}
    Let $X$ be a blocking set of $G$ of size at least four.
    If the vertices of $X$ are part of more than one connected component of $G$, then, \cref{thm:blocking_set_connectivity_property} with $X' = \emptyset$ shows that $X$ is not a minimal $d$-blocking set.
    Otherwise, $X$ is a blocking set of a single connected component of $G$.
    Hence, without loss of generality, we assume that $G$ is connected.
    If $G$ is a caterpillar, then \cref{thm:caterpillar_bounded_blocking_sets} shows that $X$ is not a minimal $d$-blocking set.

    Otherwise, $G = v_1,\dots,v_n$ is a cycle.
    Let $v_{\lambda_1},v_{\lambda_2},v_{\lambda_3},v_{\lambda_4}$ be four vertices of $X$ such that $\lambda_i < \lambda_{i+1}$ for all $i \in \range{3}$.
    Then, if we choose $X' = \{v_{\lambda_1},v_{\lambda_3}\}$, observe that removing $X'$ from $G$ results in two distinct connected components, one of which contains $v_{\lambda_2}$, and one of which contains $v_{\lambda_4}$.
    Hence, by \cref{thm:blocking_set_connectivity_property}, $X$ is not minimal.
\end{proof}

We remark that \cref{thm:graphs_with_max_degree_two_bounded_blocking_sets} is optimal in the sense that for any $d$ there are graphs in $\mathcal{G}$ that have a minimal $d$-blocking set of size three.
If $G$ is a cycle on $d+2$ vertices, then, any subset of $V(G)$ of size exactly three is a minimal $d$-blocking set, because any subset of $V(G)$ of size exactly two is an optimal solution of $G$.

Now, we proceed to the polynomial kernel for \cocModToGPlusD{}.
To obtain the kernel, we employ a clever trick which is often used to obtain polynomial kernels for \vc{} \cite{bougeretHowMuchDoes2019,fominVertexCoverStructural2016,holsEliminationDistancesBlocking2022,majumdarPolynomialKernelsVertex2018}.
The trick is that reducing the number of connected components of $G - M$ ($M$ is a modulator to some graph class) is sometimes sufficient for finding a polynomial kernel, because one can then lift a number of vertices to the modulator that is polynomial in the parameter, and obtain an instance of a problem for which a polynomial kernel is already known.

\begin{lemma}
    \label{thm:deg_2_mod_kernel}
    Let $\mathcal{G}$ be the class of graphs whose connected components are cycles or caterpillars.
    Then, \cocModToGPlusD{} admits a polynomial-kernel.
    For any $d \geq 1$ \dcocModToG{} admits a polynomial-kernel.
\end{lemma}
\begin{proof}
    Let $(G,d,k,M)$ be the input instance.
    Observe that the problem \dcoc{} can be solved in polynomial-time on graphs of $\mathcal{G}$, because it can be solved in polynomial-time on caterpillars, and a polynomial-time algorithm for cycles is also trivial (just select an arbitrary vertex and then optimally solve the remaining caterpillar).
    Moreover, the graph class $\mathcal{G}$ is hereditary.
    Furthermore, minimal $d$-blocking sets of the considered graph class have size at most three by \cref{thm:graphs_with_max_degree_two_bounded_blocking_sets}.
    Hence, we can apply \cref{thm:remove_components_quickly_d_non_constant}
    to obtain an equivalent instance $(G_2,d,k_2,M)$, such that $G_2 - M$ has at most $((d - 1) \cdot 3 + 1) \cdot \abs{M}^3$ connected components, and
    $k_2 \leq k$.

    Next, from each connected component of $G_2 - M$ that is a cycle, we move an arbitrary vertex of the cycle into the modulator $M$, and call the resulting instance $(G_3,d,k_3,M_3)$.
    This means that $\abs{M_3} \leq ((d - 1) \cdot 3 + 1) \cdot \abs{M}^3 + \abs{M}$, which is polynomial in $\abs{M} + d$.
    Now, each connected component of $G_3 - M_3$ is a caterpillar, and hence $(G_3,d,k_3,M_3)$ is an instance of \cocModToCaterpillarsPlusD{}.
    We apply \cref{main_thm:poly_kernel_caterpillar_mod} to obtain an equivalent instance $(G_4,d,k_4,M_4)$ of
    \cocModToCaterpillarsPlusD{}, where $\abs{V(G_4)}$ is polynomial in $\abs{M_3} + d$, and hence $\abs{V(G_4)}$ is polynomial in $\abs{M} + d$.
    Now, we can simply use $(G_4,d,k_4,M_4)$ as the output instance of \cocModToGPlusD{}.
    The preprocessing algorithm never changes the value of $d$, therefore it also works for \dcocModToG{}.
\end{proof}

On the lower bound side, we utilize the fact that \vc{} is \NP{}-complete on planar graphs with maximum degree three \cite[Lemma 1]{gareyRectilinearSteinerTree1977} to show that \dcoc{}  is also \NP{}-hard on planar graphs with maximum degree three for any $d \geq 2$.
Note that for some values of $d$ results in this direction were already known.
For example, for $d = 2$ \NP{}-hardness on cubic planar graphs \cite{tuVertexCoverP_32013}, and bipartite graphs with maximum degree three \cite{boliacComputingDissociationNumber2004} are known.
However, to the best of our knowledge, \NP{}-hardness on graphs with maximum degree three has not yet been explicitly shown in the literature for \dcoc{} when $d \geq 3$.
We show that the problem is hard, even on planar graphs with maximum degree three by exhibiting a simple polynomial-time reduction.

\begin{lemma}
    \label{thm:d_coc_hard_max_deg_three}
    For any $d \geq 1$, \dcoc{}  is \textnormal{\NP{}}-complete on planar graphs with maximum degree at most three.
\end{lemma}
\begin{proof}
    Clearly, \dcoc{}  is in \NP{}.
    So, we must only show that it is \NP{}-hard on planar graphs with maximum degree at most three.
    For $d = 1$ the problem is \vc{}, for which \NP{}-hardness on planar graphs with maximum degree $3$ is known \cite[Lemma 1]{gareyRectilinearSteinerTree1977}.

    For $d \geq 2$, we reduce from \vc{} on planar graphs with maximum degree $3$.
    Let $(G,k)$ be the input instance.
    For each edge $e = uv \in E(G)$, we remove $uv$ from $G$, and replace the edge with the path $u,u_e,v_e,v$, where $u_e$ and $v_e$ are fresh vertices not part of $G$.
    Furthermore, we add a fresh path $p_{u,e}^1,\dots,p_{u,e}^{d-1}$ and a fresh path $p_{v,e}^1,\dots,p_{v,e}^{d-1}$ to the graph, and add edges between $p_{u,e}^1$ and $u_e$, and between $p_{v,e}^1$ and $v_e$.
    Let $G'$ be the resulting graph.

    We set $k' = k + |E(G)|$, and the output instance of \dcoc{}  is $(G',k')$.
    Observe that each vertex of $V(G)$ has the same degree in $G$ and $G'$, and each vertex of $V(G') \setminus V(G)$ has degree at most three.
    Moreover, $G'$ is planar since the used operation to replace edges preserves planarity.
    We now prove that the instances are equivalent, starting with the forward direction.

    \begin{claim}
        \label{claim:thm:d_coc_hard_max_deg_three_forward}
        If $(G,k)$ is a yes-instance of \vc{}, then $(G',k')$ is a yes-instance of \dcoc{}.
    \end{claim}
    \begin{claimproof}
        Let $S$ be a vertex cover of $G$ of size at most $k$.
        We show how to build a \dcocset{} $S'$ of $G'$ of size at most $k'$.

        Initially, set $S' = S$.
        Now, iterate over all edges $e = uv$ of $G$.
        Edge $e$ was replaced by the path $u,u_e,v_e,v$ in $G'$.
        If we have that $u$ is in $S$, then add the vertex $v_e$ to $S'$, otherwise, add the vertex $u_e$ to $S'$.
        Clearly, the set $S'$ has size at most $k'$.

        Now, we argue that $S'$ is a \dcocset{} of $G'$.
        Consider an arbitrary connected component $C$ of $G' - S'$.
        If $C$ contains a vertex $v$ of $V(G)$, then $v$ was not part of $S$.
        Consider an edge $e = uv \in E(G)$, since $v$ was not part of $S$, vertex $u \in S$, and hence vertex $v_e$ is in $S'$.
        So, all neighbors of $v$ are part of $S'$, and $C$ has size exactly one.

        Otherwise, $C$ contains no vertex of $V(G)$.
        This means that $C$ must contain vertex $p_{u,e}^1$ for some edge $e = uv$ and some vertex $u \in V(G)$.
        If $C$ does not contain vertex $u_e$, then $C$ only contains the $d-1$ vertices $p_{u,e}^1,\dots,p_{u,e}^{d-1}$.
        Otherwise, if $C$ contains $u_e$, then, we know that both $v_e$ and $u$ are in $S'$, so $C$ contains exactly the vertices $p_{u,e}^1,\dots,p_{u,e}^{d-1}$ and vertex $u_e$, and has size $d$.
    \end{claimproof}

    We proceed to the backward direction of equivalence.
    \begin{claim}
        \label{claim:thm:d_coc_hard_max_deg_three_backward}
        If $(G',k')$ is a yes-instance of \dcoc{}, then $(G,k)$ is a yes-instance of \vc{}.
    \end{claim}
    \begin{claimproof}
        Let $S'$ be a \dcocset{} of $G'$ of size at most $k' = k + |E(G)|$.
        For any edge $e = uv \in E(G)$, we replaced the edge with a path $u,u_e,v_e,v$ in $G'$, such that $u_e$ and $v_e$ also have fresh paths of length $d-1$ attached to them.
        Observe that, if $S'$ contains a vertex of the path of length $d-1$ that is attached to $u_e$, we can replace the selected vertex with $u_e$.
        Similarly, if $S'$ contains a vertex of the path of length $d-1$ that is attached to $v_e$, we can replace the selected vertex with $v_e$.

        So, without loss of generality, $S'$ contains no vertices of the attached paths of length $d-1$.
        If $S'$ contains both $u_e$ and $v_e$, we can remove $u_e$ from the solution, and instead add vertex $u$.
        Because both $u$ and $v_e$ are selected by the resulting set, this will only yield a connected component of size $d$ that contains $u_e$.
        Hence, we can even assume that $S'$ contains at most one of $u_e$ and $v_e$.
        Moreover, $S'$ must select at least one vertex of $\{u_e,v_e\}$,
        as otherwise there would be a connected component in $G' - S'$ containing $u_e$, $v_e$, and the two paths of length $d-1$ attached to $u_e$, respectively $v_e$, and that component would have size larger than $d$.

        Now, consider the case that $S'$ does not contain vertex $u_e$.
        Then, $S'$ must contain $u$, as otherwise there would be a connected component of $G' - S'$ containing vertex $u_e$, the path of length $d+1$ attached to $u_e$, and vertex $u$.
        Similarly, if $S'$ does not contain $v_e$, then $S'$ must contain $v$.
        So, for any edge $uv$ replaced by $u,u_e,v_e,v$, we can assume that $S'$ contains exactly one vertex of  $\{u_e,v_e\}$, so also at least one of these vertices is not selected by $S'$.
        Then, at least one of $\{u,v\}$ must be selected by $S'$.
        In particular, $S' \cap V(G)$ contains one vertex of each edge of $G$, so $S' \cap V(G)$ is a vertex cover of $G$.
        Moreover, the size of $S' \cap V(G)$ is at most $k$, because $S'$ must select at least one vertex of $\{u_e,v_e\}$ for each edge $e = uv$.
    \end{claimproof}
    The lemma now follows from \cref{claim:replace_by_essence_safety_backward,claim:replace_by_essence_safety_forward}.
\end{proof}

Now, we can provide a kernelization dichotomy for \dcoc{}  parameterized by the size of degree-$c$ modulators.

\mainThmDegreeModDichotomy*{}
\begin{proof}
    If $\mathcal{G}$ is the class of graphs with maximum degree at most $2$, then a polynomial-kernel exists by \cref{thm:deg_2_mod_kernel} (note that a graph with maximum degree at most two has connected components that are either cycles or paths, and paths are catpillars).
    If $\mathcal{G}$ is the class of planar graphs with maximum degree $3$, then \cref{thm:d_coc_hard_max_deg_three} shows that the problem \dcoc{}  is \NP{}-hard on $\mathcal{G}$.
    In particular, for any instance $(G,k)$ of \dcoc{}  on graphs of $\mathcal{G}$, we have that $(G,k,\emptyset)$ is an instance of \dcocModToG{} with a parameter value of $0$.
    A kernel (of any size) would then yield that we can reduce $(G,k)$ to an equivalent instance of constant size in polynomial-time, which allows us to solve \dcoc{}  in polynomial-time on the graph class $\mathcal{G}$.
    Hence, the problem does not admit a kernel unless $\PClass = \NP$.
\end{proof} 
\section{The Necessity of d in the Parameter}
\label{sec:necessity_of_d}
In our positive results in \cref{main_thm:poly_kernel_caterpillar_mod,main_thm:kernel_degree_dichotomy} we parameterize by the distance to some graph class $\mathcal{G}$ and $d$ is either part of the parameter, or even a constant.
This begs the question whether it is possible to still obtain polynomial kernels for \coc{} when parameterizing by either $d$ or the distance to $\mathcal{G}$.
Clearly, parameterizing just by $d$ is not suitable, since the problem is already \NP{}-hard when $d$ is the constant $1$.
So, in this section we investigate the complexity of the problem when parameterizing solely by the distance to some graph class $\mathcal{G}$.
It turns out that \coc{} parameterized by the distance to a caterpillar forest is \WOne{}-hard.
This is in contrast to \cref{main_thm:poly_kernel_caterpillar_mod} which shows the existence of a polynomial kernel when $d$ is also part of the parameter.

\subsection{Distance to a Caterpillar Forest}
Let us now formally prove the hardness result for \cocModToCaterpillarsWithoutD{}, which is \coc{} parameterized by the distance to a caterpillar forest.
The core of our reduction is the same as a reduction by Ganian et al. \cite[Theorem 2]{ganianStructuralParameterizationsBoundedDegree2021} for the \boundedDegreeDeletion{} problem, apart from a key difference we elaborate on later.
The reduction is from \MRSSFullAbbrv{}.
Here, the input consists of an integer $k$, a set $S = \{s_1,\dots,s_n\}$ of vectors, where each $s_i \in \mathbb{N}^k$, a target vector $t \in \mathbb{N}^k$, and an integer $k'$.
All integers are given in unary.
The task is to decide whether there exists a subset $S' \subseteq S$ with $|S'| \leq k'$ such that $\sum_{s \in S'} s \geq t$.
\MRSS{} is \WOne{}-hard parameterized by $k + k'$ \cite[Lemma 4]{ganianStructuralParameterizationsBoundedDegree2021}.

\begin{theorem}
    \label{main_thm:w_one_hardness}
    \cocModToCaterpillarsWithoutD{} is \WOne{}-hard.
\end{theorem}
\begin{proof}
    We provide a parameterized reduction from \MRSS{}.
    Let $(k,S,t,k')$ be the input instance of \MRSS{}.
    Set $\beta = \max\{\max s \mid s \in S\}$, that is, $\beta$ is the largest integer appearing in any vector of the set $S$.
    Set $\gamma = (k' + 1) \cdot k \cdot \beta$.
    We begin by creating a graph that contains vertices $u_i^1,\dots,u_i^{k' + 1}$ for each $1 \leq i \leq k$.
    Next, for each $s \in S$, we introduce $\gamma$ vertices $s_1,\dots,s_{\gamma}$.
    We make each vertex $u_i^j$ adjacent to $s[i]$ vertices of $s_1,\dots,s_\gamma$, such that no vertex of $s_1,\dots,s_\gamma$ is adjacent to two vertices $u^i_j$.
    Observe that this is always possible by the way we choose $\gamma$.
    Let $\hat G$ be the graph constructed so far.
    We set $d = \Delta(\hat G) + 1$, that is, $d$ is the maximum degree of $\hat G$ plus one.

    Next, for each $s \in S$, and each $1 \leq i < \gamma$, we add a fresh star with $d-1$ leaves and center vertex $\hat s_i$ to the graph.
    We make $s_i$ and $s_{i+1}$ adjacent to $\hat s_i$.
    Finally, we add enough pendants to each vertex $u_i^j$ such that the vertex has $d + t[i] - 1$ neighbors overall.
    Observe that this is always possible since $u^i_j$ has at most $d - 1$ neighbors in the graph $\hat G$ and $t[i] \geq 0$.
    Call the resulting graph $G$.
    Observe that the set $M = \{u_i^j \mid i \in \range{k}, j \in \range{k'+1}\}$ is a modulator to a caterpillar forest, and thus the parameter of the output instance is $k \cdot (k' + 1)$.
    We set the solution size $\ell$ of the output \cocModToCaterpillarsWithoutD{} instance to $(\gamma - 1) \cdot |S| + k'$.
    The output instance is then $(G,d,\ell,M)$.
    Observe that the reduction can be performed quickly enough since all integers of the input instance are encoded in unary.
    We now show the correctness of the reduction in two separate claims.

    \begin{claim}
        \label{claim:main_thm:w_one_hardness_forward}
        If the input instance $(k,S,t,k')$ is a yes-instance of \MRSS{}, then $(G,d,\ell,M)$ is a yes-instance of \cocModToCaterpillarsWithoutD{}.
    \end{claim}
    \begin{claimproof}
        Let $S' \subseteq S$ be a set of vectors such that $\sum_{s \in S'} s \geq t$, and $|S'| \leq k'$.
        We create a selection $S^*$ of $V(G)$ as follows.
        For each $s \in S'$, we select vertex $s_i$ for all $i \in \range{\gamma}$, and for each $s \in S \setminus S'$ we select vertex $\hat s_i$ for all $i \in \range{\gamma - 1}$.
        Observe that the resulting selection $S^*$ has size exactly $|S'| \cdot \gamma + |S \setminus S'| \cdot (\gamma - 1) \leq |S| \cdot (\gamma - 1) + k'$.
        So, all that remains is arguing that $G - S^*$ has no connected component of size larger than $d$.

        Consider an arbitrary connected component $Z$ of $G - S^*$, if $Z$ contains no vertex $u^i_j$, then $Z$ must either be a singleton vertex $s_{i'}$, or $Z$ contains a vertex $\hat s_{i'}$ for some $s \in S'$.
        In the latter case, since both $s_{i'}$ and $s_{i' + 1}$ are in $S^*$, the connected component $Z$ has size exactly $d$.
        So, now consider that $Z$ contains a vertex $u^i_j$.
        We have that $u^i_j$ has exactly $d + t[i] - 1$ neighbors in $G$.
        Moreover, for each neighbor of $u^i_j$ that is not selected by $S^*$, it holds that the neighbor has degree exactly one in $G - S^*$.
        Indeed, neighbors of $u^i_j$ are either pendants of $u^i_j$, or vertices $s_{i'}$ (for some $i'$), and in the latter case the neighbors of $s_{i'}$ apart from $u^i_j$ are part of $S^*$.
        Then, since $\sum_{s \in S'} s[i] \geq t[i]$, and we have that $u^i_j$ has exactly $s[i]$ neighbors in the set $s_1,\dots,s_\gamma$ in the graph $G$ which are all selected by $S^*$ for each $s \in S'$, we see that at least $t[i]$ neighbors of $u^i_j$ are selected.
        Thus, in $G - S^*$, the vertex $u^i_j$ has at most $d - 1$ neighbors, and hence $Z$ has size at most $d$.
    \end{claimproof}

    \begin{claim}
        \label{claim:main_thm:w_one_hardness_backward}
        If instance $(G,d,\ell,M)$ is a yes-instance of \cocModToCaterpillarsWithoutD{}, then $(k,S,t,k')$ is a yes-instance of \MRSS{}.
    \end{claim}
    \begin{claimproof}
        Let $S^* \subseteq V(G)$ be a \dcocset{} of $G$ of size at most $\ell$.
        Consider an arbitrary $s \in S$.
        Then, for each $i \in \range{\gamma - 1}$ vertex $\hat s_i$ has $d-1$ pendants, and thus $S^*$ must select at least one of $\hat s_i$, $s_i$, or a pendant of $\hat s_i$.
        In particular, this means that $S^*$ must select at least $(\gamma - 1) \cdot |S|$ vertices which are not part of $M$.
        Hence, $S^*$ selects at most $k'$ vertices which are not part of $M$, and thus there exists a $j$ such that no vertex of $u_1^j,\dots,u_k^j$ is selected by $S^*$.
        Clearly, $S^*$ could select $u_i^j$ instead of a pendant of $u_i^j$, so we can also assume that no pendant of $u_i^j$ in $G$ is selected by $S^*$.

        Before we continue, we want to analyze which vertices of $s_1,\dots,s_\gamma$ and $\hat s_1,\dots,\hat s_{\gamma - 1}$ may be selected by $S^*$ for any $s \in S$.
        Clearly, if $S^*$ selects a pendant of a vertex $\hat s_{i'}$ (for some $i'$), we can remove this pendant from $S^*$ and select $\hat s_{i'}$ instead.
        So, we can assume that $S^*$ selects none of the pendants of the vertices $\hat s_1,\dots, \hat s_{\gamma - 1}$.
        If we have that $S^*$ selects $\gamma$ vertices of $s_1,\dots,s_\gamma$,$\hat s_1,\dots,\hat s_{\gamma - 1}$, then we can assume that $S^*$ selects exactly $s_1,\dots,s_\gamma$, as no vertex of $\hat s_1, \dots \hat s_{\gamma -1}$ has a neighbor in $M$.
        Otherwise, we have that $S^*$ selects $\gamma - 1$ vertices of $s_1,\dots,s_\gamma, \hat s_1,\dots, \hat s_{\gamma - 1}$, we now want to argue that then, $S^*$ must select exactly the vertices $\hat s_1,\dots, \hat s_{\gamma - 1}$.
        Otherwise, assume that $S^*$ selects vertex $s_{i'}$ for some $i'$.
        Then, for all $j < i'$, we have that $\hat s_j$, the $d-1$ pendants of $\hat s_j$, and vertex $s_j$ form a connected subgraph of size $d+1$ in $G$.
        Moreover, for all $j \geq i'$, we have that $\hat s_j$, the pendants of $\hat s_j$, and the vertex $s_{j + 1}$ form a connected subgraph of size $d+1$ in $G$.
        All of these subgraphs are vertex-disjoint.
        This yields that $S^*$, that selects $s_{i'}$, must select $\gamma -1$ further vertices of $s_1,\dots, s_\gamma$ and $\hat s_1,\dots, \hat s_{\gamma - 1}$, contradicting that at most $\gamma - 1$ of these are selected.

        So, we can without loss of generality assume that for any $s \in S$, $S^*$ selects either all vertices $s_1,\dots,s_\gamma$, or all vertices $\hat s_1,\dots \hat s_{\gamma - 1}$.
        We build our solution for \MRSS{} by choosing a vector $s \in S$ if and only if $S^*$ selects the $\gamma$ vertices $s_1,\dots,s_\gamma$, and we call the resulting set of vectors $S'$.
        By the choice of $\ell$, we see that the size of $S'$ is at most $k'$.
        Now, consider any $i \in \range{k}$.
        We have that $u_i^j$ has $d + t[i] - 1$ neighbors in $G$, and furthermore, neither $u_i^j$ nor a pendant of $u_i^j$ in $G$ is selected by $S^*$.
        Hence, $S^*$ must select at least $t[i]$ neighbors of $u_i^j$ which are of the form $s_i$ for some $s \in S'$.
        Since $S^*$ contains $s[i]$ vertices of $s_1,\dots,s_\gamma$ that are neighbors of $u_i^j$ if and only if $s \in S'$, we see that indeed $\sum_{s \in S'} s[i] \geq t[i]$.
    \end{claimproof}
    By \cref{claim:main_thm:w_one_hardness_forward,claim:main_thm:w_one_hardness_backward} the reduction is correct and the theorem follows.
\end{proof}

The main difference between our reduction and proof and the one by Ganian et al. \cite{ganianStructuralParameterizationsBoundedDegree2021} lies in the fact that their reduction does not work for the problem \cocModToCaterpillarsWithoutD{}, mainly because the gadgets they use to encode sets of $S$ are not caterpillars (but trees with bounded depth).
Hence, our contribution here is the creation of new gadgets to encode the sets of $S$, which make the reduction work for our problem.

\subsection{Vertex Cover Parameterization}

By the previous result, the parameter $d$ cannot be dropped from \cref{main_thm:poly_kernel_caterpillar_mod}.
However, it may still be possible that polynomial kernels exist for different structural parameters, even when not considering $d$ in the parameter.
Next, we show that \coc{} does not admit a polynomial-kernel even when the parameter is the size of a \emph{vertex cover} plus the \emph{solution size} (unless $\NP \subseteq \coNPPoly$).\footnote{Recall that while \dcoc{} has a polynomial kernel parameterized by solution size, the problem \coc{} is \WOne{}-hard.}
We denote the problem \coc{} parameterized by the size of a vertex cover (that is given as part of the input) plus the solution size as \cocModToVCPlusK{}.
The result means that \coc{} does not have a polynomial kernel for ``almost all'' regular structural parameterizations.

We show our kernelization lower bound by a polynomial parameter transformation from the \exactSetCover{} problem.
In this problem, the input consists of a universe $\mathcal{U}$, a set family $\mathcal{F} \subseteq \binom{\mathcal{U}}{t}$ for some integer $t$, and an integer $k$.
The question is whether there is a subset $\mathcal{F}' \subseteq \mathcal{F}$ of size $k$ such that each element of $\mathcal{U}$ is in exactly one set of $\mathcal{F}'$.
This problem was considered by Dell and Marx \cite{dellKernelizationPackingProblems2012} for fixed constant $t$ under the name \perfectTSetMatching{}.
They show that the problem with $t \geq 3$ does not have a compression of size $\Oh(k^{t - \varepsilon})$ for any $\varepsilon > 0$, unless $\NP \subseteq \coNPPoly$ \cite[Theorem 1.2]{dellKernelizationPackingProblems2012}.
Note that they phrase their result in terms of kernels, however the underlying machinery also rules out compressions.
A compression algorithm for the \exactSetCover{} problem can also be used for the problem with fixed $t$.
Since any yes-instance of the problem must fulfill that $|\mathcal{U}| = k \cdot t = O(k)$ when $t$ is fixed, their result implies that the problem \exactSetCover{} parameterized by the universe size $|\mathcal{U}|$, which we denote as \exactSetCoverUniverse{}, does not have a polynomial compression, unless $\NP \subseteq \coNPPoly$.

\begin{lemma}[Follows from {\cite[Theorem 2]{dellKernelizationPackingProblems2012}}]
    \label{thm:exact_set_cover_universe_lb}
    \exactSetCoverUniverse{} does not admit a polynomial compression, unless $\NP \subseteq \coNPPoly$.
\end{lemma}

Let us now use \cref{thm:exact_set_cover_universe_lb} as the starting point of a polynomial parameter transformation.
On the technical level, our reduction is inspired by a reduction of Giannopoulou et al. \cite{giannopoulouUniformKernelizationComplexity2017}.
They reduce from the same problem (with fixed $t$) and show compression lower bounds for three different \fdeletion{} problems parameterized by vertex cover.
The problem we consider is not a \fdeletion{} problem (because $d$ is not constant), but the same general structure of their construction still extends fairly well to \cocModToVCPlusK{}.

To ease the presentation, let us first introduce an annotated variant of \cocModToVCPlusK{}.
In the \cocModToVCAnnotatedPlusK{} problem the input consists of a graph $G$, integers $k$, $d$, a vertex cover $M$ of $G$, and annotations $A \subseteq \binom{M}{2}$.
The goal is to decide whether there exists a set $S \subseteq V(G)$ of size at most $k$ such that (1) each connected component of $G - S$ has size at most $d$, and (2) for each $\{u,v\} \in A$ we have that $S \cap \{u,v\} \not = \emptyset$.
Intuitively, a solution needs to solve two problems at the same time: \dcoc{} on the graph $G$, and \vc{} on the graph with edge set $A$.
The parameter is $|M| + k$.

\begin{lemma}
    \label{thm:kernel_lower_bound_annotated_coc}
    The problem \cocModToVCAnnotatedPlusK{} does not have a polynomial compression, unless $\NP \subseteq \coNPPoly$.
\end{lemma}
\begin{proof}
    Let the input instance of \exactSetCoverUniverse{} be $(\mathcal{U},\mathcal{F},k)$, where each set of $\mathcal{F}$ has size $t \leq |\mathcal{U}|$.
    For convenience let us denote $\mathcal{U} = \{u_1,\dots,u_n\}$.
    If we do not have $k = \frac{n}{t}$ we clearly have a no-instance, and hence we can output a constant-size no-instance of \cocModToVCAnnotatedPlusK{}.

    Otherwise, we have that $k = \frac{n}{t}$ and thus $n = k \cdot t$.
    We now describe the output graph $G$ of our \coc{} instance, begin with the empty graph, and initially set $A = \emptyset$.
    For each $i \in \range{n}$ and $j \in \range{k}$ we create a vertex $u_i^j$.
    Intuitively, each vertex $u_i^j$ represents the universe element $u_i$, however for technical reasons we have $k$ copies of each such representative.
    For all $j \in \range{k}$ set $C^j =  \{u_i^j \mid i \in \range{n}\}$, and for all $i \in \range{n}$ set $R_i = \{u_i^j \mid j \in \range{k}\}$.
    We refer to $C^j$ as the $j$th column and to $R_i$ as the $i$th row of the construction.
    For each $j \in \range{k}$ we make the vertices $C^j$ into a clique.
    Moreover, for each $i \in \range{n}$ and each vertex pair $\{u_i^{j},u_i^{j'}\}$ where $j \not = j'$ we add the annotation $\{u_i^{j},u_i^{j'}\}$ to the set $A$.
    Set $M = \bigcup_{j \in \range{k}} C^j = \bigcup_{i \in \range{n}} R_i$.

    Now, for each $j \in \range{k}$ and each $F \in \mathcal{F}$ we add a vertex $v_F^j$.
    We make $v_F^j$ adjacent to all vertices of the set $\{u^j_i \mid i \in \range{n}, u_i \notin F\}$.
    In other words, $v_F^j$ is adjacent to all vertices of $C^j$ which do \emph{not} represent an element of the set $F$.
    We set $d = |\mathcal{F}| + t - 1$, and $k' = n \cdot (k-1)$.
    We can assume that $t \geq 2$ as the input problem is easily solvable in polynomial-time if $t = 1$, and thus $d \geq 1$.
    The output instance of \cocModToVCAnnotatedPlusK{} is $(G,d,k,M,A)$.
    Observe that the set $M$ is a vertex cover of $G$ of size exactly $n \cdot k = n \cdot \frac{n}{t} \in \Oh{(n^2)}$.
    Moreover, $k' = n \cdot (k-1) = n \cdot (\frac{n}{t} - 1) \in \Oh{(n^2)}$.
    Hence, the parameter $k' + |M|$ of the output instance is bounded by a polynomial in $n$.
    We now show the forward direction of correctness.

    \begin{claim}
        \label{claim:thm:kernel_lower_bound_annotated_coc_forward}
        If $(\mathcal{U},\mathcal{F},k)$ is a yes-instance of \exactSetCoverUniverse{}, then $(G,d,k',M,A)$ is a yes-instance of \cocModToVCAnnotatedPlusK{}.
    \end{claim}
    \begin{claimproof}
        Let $S = \{F_1,\dots,F_k\} \subseteq \mathcal{F}$ be an exact set cover of $\mathcal{U}$, that is, a set family such that each element of $\mathcal{U}$ appears in exactly one set of $S$.
        Then, define \[S' = \{u_i^j \mid i \in \range{n}, j \in \range{k}, u_i \notin F_j\}.\]
        As each universe element $u_i$ is in exactly one set $F_j$, we have that for each $i \in \range{n}$ there is exactly one $j$ such that $u_i^j$ is not selected by $S'$.
        Hence, $|S'| = k'$.
        For the annotations in the set $A$, observe that these are fulfilled since all apart from exactly one vertex of each row $R_i$ is selected.

        It remains to argue that  $S'$ is also a \dcocset{} of $G$.
        Consider an arbitrary connected component $C$ of $G - S'$.
        If $C$ has size one we are done since we have $d \geq 1$, so consider that $C$ has size larger than one.
        Then, $C$ must contain a vertex of $M$ since $G - M$ is an independent set.
        Let $u_i^j$ be a vertex of $M$ contained in $C$.
        Now, observe that a vertex of $M$ is only connected to other vertices of $M$ of the same column, and that vertices of $G - M$ are also only adjacent to vertices of the same column.
        So, $C$ can only consist of vertices of $M$ stemming from column $C^j$, as well as of vertices $v_F^j$ for some $F \in \mathcal{F}$.
        Observe that $S'$ selects all vertices of column $C^j$ which do not correspond to elements of the set $F_j$, so at most $t$ vertices of $C^j$ are not selected.
        The neighborhood of $v_{F_j}^j$ consists exactly of those vertices which do not correspond to elements of $F_j$, and hence all neighbors of $v_{F_j}^j$ are selected by $S'$.
        Finally, this means that $C$ contains at most $|\mathcal{F}| - 1$ vertices of $G - M$, so the size of $C$ is at most $d = |\mathcal{F}| - 1 + t$.
    \end{claimproof}

    We now show the backward direction of correctness.

    \begin{claim}
        \label{claim:thm:kernel_lower_bound_annotated_coc_backward}
        If $(G,d,k,M,A)$ is a yes-instance of \cocModToVCAnnotatedPlusK{} then $(\mathcal{U},\mathcal{F},k')$ is a yes-instance of \exactSetCoverUniverse{}.
    \end{claim}
    \begin{claimproof}
        Let $S'$ be a solution to the output instance of \cocModToVCAnnotatedPlusK{}.
        Fix an arbitrary $i \in \range{n}$.
        Since each pair of distinct vertices of row $R_i$ has an annotation in $A$, we have that $S'$ must select all apart from at most one vertex of $R_i$.
        Given that $k' = n \cdot (k-1)$, and there are $n$ rows, we observe that $S'$ must select exactly $k-1$ vertices of $R_i$, and hence there is exactly one $j$ such that $u_i^j$ is not part of $S'$.

        Now, assume that there is some column $C^j$ such that $S'$ selects fewer than $n-t$ vertices of $C^j$, and recall that $C^j$ contains exactly $n$ vertices.
        Let $u_i^j$ be one of the vertices of $C^j$ that is not part of $S'$.
        Recall that $C^j$ is a clique in $G$, so the connected component of $u_i^j$ in $G - S'$ contains at least $t+1$ vertices of $C^j$.
        Next, consider an arbitrary vertex $v_{F}^j$ for $F \in \mathcal{F}$.
        Since $v_F^j$ has $n - t$ neighbors in $C^j$, and $|C^j| = n$, $v_F^j$ has an unselected neighbor in $C^j$, and is thus part of the connected component of $u_i^j$ in $G- S'$.
        Hence, the component of $u_i^j$ in $G' - S'$ contains at least $|\mathcal{F}| + t + 1 > d$ vertices, a contradiction.

        If some column $C^j$ is such that $S'$ selects more than $n-t$ vertices of $C^j$, observe that since $k' = n \cdot (k-1) = nk - kt = k(n-t)$, there would have to be a column $C^{j'}$ in which fewer than $n - t$ vertices are selected by $S'$, which cannot be the case by the previous paragraph.

        Hence, $S'$ selects exactly $n - t$ vertices of each column $C^j$.
        Since for each universe element $u_i$ additionally there exists exactly one column $C^j$ such that $u_i^j$ is not selected, it suffices to show that those vertices of $C^j$ that are not selected by $S'$ are a set of $\mathcal{F}$.
        Indeed, if that was the case then if we set $F_j$ to be $\{u_i \mid u_i^j \in C^j \setminus S'\}$ for each $j \in \range{k}$ and $S = \{F_1,\dots,F_k\}$, we would have that each element $u_i \in \mathcal{U}$ appears in exactly one set of $S$, and so the input instance would be a yes-instance.

        So, fix an arbitrary $j \in \range{k}$, we shall show that $F_j = \{u_i \mid u_i^j \in C^j \setminus S'\}$ is in $\mathcal{F}$.
        Let $u_i^j$ be a vertex of $C^j \setminus S'$, and $C$ the connected component of $u_i^j$ in $G - S'$.
        Recall that $S'$ selects exactly $n - t$ vertices of each column, and thus $C$ contains exactly $t$ vertices of $C^j$.
        Toward a contradiction, assume that $F_j \notin \mathcal{F}$.
        Consider an arbitrary $F \in \mathcal{F}$.
        Then, the vertex $v_F^j$ is adjacent to all $n-t$ vertices of $C^j$ that do not represent an element of $F$.
        In particular, since $F \not = F_j$ we have that $v_F^j$ is adjacent to a vertex of $C^j \setminus S'$, and thus $v_F^j$ is part of $C$.
        So, $C$ has size $|\mathcal{F}| + t > d$, which is a contradiction.
    \end{claimproof}
    By \cref{claim:thm:kernel_lower_bound_annotated_coc_backward,claim:thm:kernel_lower_bound_annotated_coc_forward} we see that the polynomial parameter transformation is correct and \cref{thm:exact_set_cover_universe_lb} concludes the proof.
\end{proof}

Finally, we need to provide a polynomial parameter transformation from the annotated problem \cocModToVCAnnotatedPlusK{} to \cocModToVCPlusK{}.
For this, we can use the idea at the core of \cref{thm:d_coc_hard_max_deg_three}.

\begin{theorem}
    \label{main_thm:vc_kernelization_lower_bound}
    The problem \cocModToVCPlusK{} does not admit a polynomial compression unless $\NP \subseteq \coNPPoly$.
\end{theorem}
\begin{proof}
    We provide a polynomial parameter transformation from \cocModToVCAnnotatedPlusK{}.
    Let $(G,k,d,M,A)$ be the input instance.
    Note that we can assume that $k,d \leq |V(G)|$, and hence the magnitude of these values is polynomial in the input size.

    Initially set $M' = M$.
    For each $a = \{u,v\} \in A$, we add two vertices $a_u$ and $a_v$ to the graph.
    We make $a_v$ adjacent to $v$, and $a_u$ adjacent to $u$, and add an edge between $a_u$ and $a_v$.
    Moreover, we add $d-1$ pendants each to $a_u$ and $a_v$.
    We also add $a_u$ and $a_v$ to $M'$.
    Denote the resulting graph as $G'$.

    It is clear that $M'$ is a vertex cover of $G'$ because we added $a_u$ and $a_v$ to the vertex cover when handling annotation $a$.
    Set $k' = k + |A|$.
    Let the output instance be $(G',k',d,M')$.
    The size of $M'$ is at most $|M| + 2|A| \leq |M| + 2|M|^2 \in \Oh{(|M|^2)}$ (recall that we only allowed annotations between pairs of vertices in $M$).
    Moreover, $k' = k + |A| \leq k + |M|^2$.
    Thus, $|M'| + k'$ is bounded polynomially in $|M| + k$.
    It is also clear that the parameter transformation runs in polynomial-time since we only add $2d \cdot |A| \leq 2|V(G)| \cdot |V(G)|^2$ new vertices to the graph.
    So, all that remains is showing the correctness.
    We begin with the forward direction.

    \begin{claim}
        \label{claim:main_thm:vc_kernelization_lower_bound_forward}
        If $(G,k,d,M,A)$ is a yes-instance of \cocModToVCAnnotatedPlusK{}, then $(G',k',d,M')$ is a yes-instance of \cocModToVCPlusK{}.
    \end{claim}
    \begin{claimproof}
        Let $S$ be a solution of the input instance, and initially set $S' = S$.
        For each $a = \{u,v\}\in A$ we have that $S$ selects at least one of $u$ or $v$.
        If $S$ selects $u$, then we add $a_v$ to $S'$, otherwise if $S$ does not select $u$, we add $a_u$ to $S'$.
        The resulting set $S'$ has size at most $k + |A|$ because we change $S$ by adding exactly one vertex for each $a \in A$.

        Clearly, each connected component of $G' - S'$ that only contains vertices of $V(G)$ has size at most $d$.
        So, consider some connected component $C$ of $G' - S'$ that contains a vertex of $V(G') \setminus V(G)$.
        Every such component that is not just an isolated vertex also contains a vertex $a_u$ or $a_v$ for some annotation $\{u,v\} \in A$.
        So, consider an arbitrary annotation $\{u,v\} \in A$.

        If $u \in S$, then we have that $a_v \in S'$, and $a_u \notin S'$.
        Since the only neighbors of $a_u$ are $a_v$, the pendants of $a_u$, and vertex $u$, we have that the connected component of $a_u$ in $G' - S'$ contains only $a_u$ and $d-1$ pendants of $a_u$.

        Similarly, if $u \notin S$, then we have that $a_u \in S'$.
        Moreover, because $\{u,v\} \in A$ we have that $v \in S$.
        Hence, neighbors $v$ and $a_u$ of $a_v$ are selected, and the connected component of $G' - S'$ that contains $a_v$ has size at most $d$.
    \end{claimproof}

    \begin{claim}
        \label{claim:main_thm:vc_kernelization_lower_bound_backward}
        If $(G',k',d,M')$ is a yes-instance of \cocModToVCPlusK{}, then $(G,k,d,M,A)$ is a yes-instance of \cocModToVCAnnotatedPlusK{}.
    \end{claim}
    \begin{claimproof}
        Let $S'$ be a solution of the output instance.
        For each $a = \{u,v\} \in A$ the output graph $G'$ contains $a_u$, $a_v,$ and $d-1$ pendants each for $a_u$ and $a_v$.
        Hence, $S'$ must select $a_u,a_v$ or a pendant of $a_u$ or $a_v$.
        It is not difficult to see that, if $S'$ selects a pendant of $a_v$, we can instead select $a_v$.
        So, we can assume that $S'$ contains at least one of $\{a_u,a_v\}$, and that $S'$ contains no pendant of $a_u$ or $a_v$.

        Next, assume that $S'$ contains both $a_u$ and $a_v$.
        In this case, we can select $v$ instead of $a_v$ (or equivalently $u$ instead of $a_u$).
        Thus, we can assume that $S'$ selects exactly one of $\{a_u,a_v\}$.
        Then, given that one of $\{a_u,a_v\}$, say $a_u$, is not selected, $S'$ must select vertex $u$, as otherwise $a_u$, the $d-1$ pendants of $a_u$ and vertex $u$ would lead to a connected component of size $d+1$ in $G' - S'$.
        Thus, we see that $S' \cap V(G)$ selects one vertex of each annotation in $A$.
        Moreover, $S' \cap V(G)$ must be a \dcocset{} of $G$ because $G$ is a subgraph of $G'$.
        Finally, the fact that $S'$ selects one vertex of $\{a_u,a_v\}$ for each annotation $a = \{u,v\} \in A$ shows that $|S' \cap V(G)| \leq k$.
    \end{claimproof}

    By combining \cref{claim:main_thm:vc_kernelization_lower_bound_backward,claim:main_thm:vc_kernelization_lower_bound_forward} and \cref{thm:kernel_lower_bound_annotated_coc} the theorem is proven.
\end{proof}

In \cref{main_thm:w_one_hardness} we were able to show that dropping the parameter $d$ results in a \WOne{}-hard problem.
In contrast to this, \cref{main_thm:vc_kernelization_lower_bound} only rules out polynomial compressions.
It turns out that this is not a weakness of our result: the problem \cocModToVCWithoutK{} (\coc{} parameterized by the size of a vertex cover given in the input) is \FPT{}.
Hence, the result of \cref{main_thm:vc_kernelization_lower_bound} cannot be strengthened to show \WOne{}-hardness of the problem.
We present the simple branching algorithm that shows this next.

\begin{theorem}
    \label{thm:coc_mod_to_vc_fpt}
    \cocModToVCWithoutK{} is \FPT{}.
\end{theorem}
\begin{proof}
    Let the input instance be $(G,d,k,M)$.
    We branch over all (non-proper) colorings of $M$ using the colors $\range[0]{|M|}$.

    Let $\Gamma$ be the coloring of a fixed branch.
    \begin{itemize}
        \item  Set $S_1$ to be those vertices $v \in M$ with $\Gamma(v) = 0$, which indicates that $v$ is part of the solution set.
        \item Set $S_2 = S_1$.
        \item Iterate over all $w \in G - M$.
              If there are $m_1,m_2 \in N(w)$ such that $\Gamma(m_1) \not = 0$ and $\Gamma(m_2) \not = 0$ and $\Gamma(m_1) \not = \Gamma(m_2)$, then add $w$ to $S_2$.
        \item If $S_2$ is a \dcocset{} of $G$ of size at most $k$ output yes.
    \end{itemize}
    If no branch outputs that the instance is a yes-instance, output that it is a no-instance.
    Let us now proceed to the correctness proof of the algorihm.

    \begin{claim}
        \label{claim:thm:coc_mod_to_vc_fpt_correctness_1}
        The algorithm outputs no if $(G,d,k,M)$ is a no-instance.
    \end{claim}
    \begin{claimproof}
        Since the algorithm can only output yes if a \dcocset{} of $G$ of size at most $k$ is found, if the input instance is a no-instance, the algorithm will always output no.
    \end{claimproof}

    Finally, we need to show that the algorithm outputs yes if the input instance is a yes-instance.
    \begin{claim}
        \label{claim:thm:coc_mod_to_vc_fpt_correctness_2}
        If $(G,d,k,M)$ is a yes-instance, then the algorithm outputs yes.
    \end{claim}
    \begin{claimproof}
        Let $S$ be a \dcocset{} of $G$ of size at most $k$.
        Let the connected components of $G - S$ that contain vertices of $M$ be $C_1,\dots,C_t$ (clearly $t \leq |M|$).

        We now show that we can assume that $S$ selects no vertex $w$ of $G - M$ such that $N(w) \setminus S \subseteq V(C_i)$ for some $i$.
        Towards a contradiction, assume there is such a $w \in S$.
        If $N(w) \subseteq S$ selecting $w$ is not necessary.
        Otherwise, we can select an arbitrary vertex of $N(w) \setminus S$ instead of selecting $w$.
        This cannot increase the size of a connected component because all vertices of $N(w) \setminus S$ are already in the same connected component in $G-S$.
        Hence, we can assume that all selected vertices of $G - M$ have unselected neighbors in two different connected components of $C_1,\dots,C_t$.

        Clearly, there is a coloring $\Gamma$ of $M$ with the colors $\range[0]{|M|}$ such that all vertices in $S \cap M$ get color $0$, and all vertices in $C_i$ get color $i$ for all $i \in \range{t}$.
        Let us inspect the algorithm for this branch $\Gamma$.
        First, it creates set $S_1$ which clearly fulfills $S_1 = S \cap M$.
        Then, it creates $S_2$ by selecting those vertices of $G - M$ which have neighbors with different colors, which thus correspond to vertices having neighbors in different connected components of $G - S$.
        It is not difficult to confirm that we have $S_2 \subseteq S$.
        However, since $S$ selects no vertex $w$ of $G - M$ that has unselected neighbors in only a single connected component of $G - S$, we obtain that $S = S_2$, and the algorithm outputs yes.
    \end{claimproof}

    It is clear that the sketched algorithm terminates and runs in time $f(|M|) \cdot |V(G)|^{\Oh(1)}$ for a computable function $f$, and it is correct by \cref{claim:thm:coc_mod_to_vc_fpt_correctness_1,claim:thm:coc_mod_to_vc_fpt_correctness_2}.
\end{proof} 
\section{Problem Definitions}
\label{sec:problem_definitions}
In this section we collect and formally restate the definitions of the problems we consider in the paper.

\subsection{Component Order Connectivity Variants}
Recall that for an integer $d \geq 1$ and graph $G$, a \dcocset{} of $G$ is a set $S$ such that each connected component of $G - S$ has size at most $d$.

\defproblem{\phantomsection\cocFullAbbreviated{}}{Graph $G$, integer $d \geq 1$, integer $k$}{Is there a \dcocset{} $S$ of $G$ with $|S| \leq k$?}
\label{problem:coc}

If $d$ is a fixed-constant, we instead obtain the problem \dcoc{}. 
For each fixed $d$ we have a different problem.

\defproblem{\phantomsection\dcoc{}}{Graph $G$, integer $k$}{Is there a \dcocset{} $S$ of $G$ with $|S| \leq k$?}
\label{problem:dcoc}
\label{problem:vertex_cover}

The problem \vcFullAbbreviated{} is the same as $1$-\coc{}.
Now, we proceed to the basic parameterized problems.
First, we have the parameterization by $d + k$.

\defparproblem{\phantomsection\cocSolSizePlusD{}}{Graph $G$, integer $d \geq 1$, integer $k \geq 0$}{$d+k$}{Is there a \dcocset{} $S$ of $G$ with $|S| \leq k$?}
\label{problem:coc_sol_size_plus_d}

Next, we consider structural parameters.
In the following, $\mathcal{G}$ always denotes a fixed graph class, and we parameterize by the size of a modulator $M$ to $\mathcal{G}$.
For the problem \coc{} we additionally consider $d$ as part of the parameter.

\defparproblem{\phantomsection\cocModToGPlusD}{Graph $G$, integer $d \geq 1$, integer $k$, set $M$ such that $G - M \in \mathcal{G}$}{$d + \abs{M}$}{Is there a \dcocset{} $S$ of $G$ with $|S| \leq k$?}
\label{problem:coc_mod_to_g_plus_d}

When we consider the problem \dcoc{} instead of \coc{} $d$ is a constant, and thus it is naturally no longer part of the parameter.

\defparproblem{\phantomsection\dcocModToG}{Graph $G$, integer $k$, set $M$ such that $G - M \in \mathcal{G}$}{$\abs{M}$}{Is there a \dcocset{} $S$ of $G$ with $|S| \leq k$?}
\label{problem:dcoc_mod_to_g}

Our result in \cref{main_thm:poly_kernel_caterpillar_mod} is for the problem \cocModToCaterpillarsPlusD{}, which is \cocModToGPlusD{} where $\mathcal{G}$ is fixed to be the class of graphs with pathwidth at most $1$, that is, the class of caterpillar forests.

\defparproblem{\phantomsection\cocModToCaterpillarsPlusD}{Graph $G$, integer $d \geq 1$, integer $k$, set $M$ such that $G - M$ is a caterpillar forest}{$d + \abs{M}$}{Is there a \dcocset{} $S$ of $G$ with $|S| \leq k$?}
\label{problem:coc_mod_to_caterpillars_plus_d}

The result also applies to the problem where $d$ is fixed.

\defparproblem{\phantomsection\dcocModToCaterpillars{}}{Graph $G$, integer $k$, set $M$ such that $G - M$ is a caterpillar forest}{$\abs{M}$}{Is there a \dcocset{} $S$ of $G$ with $|S| \leq k$?}
\label{problem:dcoc_mod_to_caterpillars}

We additionally provide results for the following problems where $d$ is not part of the parameter.

\defparproblem{\phantomsection\cocModToVCWithoutK}{Graph $G$, integer $d \geq 1$, integer $k$, set $M$ such that $G - M$ is an independent set}{$\abs{M}$}{Is there a \dcocset{} $S$ of $G$ with $|S| \leq k$?}
\label{problem:coc_mod_to_vc_without_k}

\defparproblem{\phantomsection\cocModToVCPlusK}{Graph $G$, integer $d \geq 1$, integer $k$, set $M$ such that $G - M$ is an independent set}{$k + \abs{M}$}{Is there a \dcocset{} $S$ of $G$ with $|S| \leq k$?}
\label{problem:coc_mod_to_vc_plus_k}

\defparproblem{\phantomsection\cocModToVCAnnotatedPlusK}{Graph $G$, integer $d \geq 1$, integer $k$, set $M$ such that $G - M$ is an independent set, set $A \subseteq \binom{M}{2}$}{$k + \abs{M}$}{Is there a set $S \subseteq V(G)$ of size at most $k$ that is a \dcocset{} of $G$ and for each $\{u,v\} \in A$ we have $S \cap \{u,v\} \not = \emptyset$?}
\label{problem:coc_mod_to_vc_annotated_plus_k}

\defparproblem{\phantomsection\cocModToCaterpillarsWithoutD}{Graph $G$, integer $d \geq 1$, integer $k$, set $M$ such that $G - M$ is a caterpillar forest}{$\abs{M}$}{Is there a \dcocset{} $S$ of $G$ with $|S| \leq k$?}
\label{problem:coc_mod_to_caterpillars_without_d}

\subsection{Further Problems}

For our lower bound in \cref{main_thm:w_one_hardness} we reduce from \MRSS{} which is \WOne{}-hard \cite[Lemma 4]{ganianStructuralParameterizationsBoundedDegree2021}.

\defparproblem{\phantomsection\MRSSFullAbbrv{}}{Integer $k$, set $S = \{s_1,\dots,s_n\}$ of vectors where each $s_i \in \mathbb{N}^k$, target vector $t \in \mathbb{N}^k$, integer $k'$, all integers of the instance are encoded in unary.}{$k + k'$}{Is there a set $S' \subseteq S$ with $|S'| \leq k'$ and $\sum_{s \in S'} s \geq t$?}
\label{problem:mrss}

For our kernelization lower bound in \cref{main_thm:vc_kernelization_lower_bound} we reduce from \exactSetCover{} parameterized by the universe size, which we denote by \exactSetCoverUniverse{}.
A lower bound for the problem follows from the work of Dell and Marx \cite{dellKernelizationPackingProblems2012}.

\defproblem{\phantomsection\exactSetCover}{Universe $\mathcal{U}$, set family $\mathcal{F} \subseteq \binom{\mathcal{U}}{t}$ for some integer $t$, integer $k$}{Is there a set $\mathcal{F}' \subseteq \mathcal{F}$ of size $k$ such that each element of $\mathcal{U}$ is in exactly one set of $\mathcal{F}'$?}
\label{problem:exact_set_cover}

\defparproblem{\phantomsection\exactSetCoverUniverse{} }{Universe $\mathcal{U}$, set family $\mathcal{F} \subseteq \binom{\mathcal{U}}{t}$ for some integer $t$, integer $k$}{$|\mathcal{U}|$}{Is there a set $\mathcal{F}' \subseteq \mathcal{F}$ of size $k$ such that each element of $\mathcal{U}$ is in exactly one set of $\mathcal{F}'$?}
\label{problem:exact_set_cover_universe}

Note that in our definition of \exactSetCoverUniverse{} each set has the same size $t$, but $t$ is not a constant and also not part of the parameter.

We provide no results for the following problems, but provide their definitions here nevertheless since they are mentioned in our paper.

\defproblem{\phantomsection\fvs{}}{Graph $G$, integer $k$}{Is there a set $S \subseteq V(G)$ of size at most $k$ such that $G - S$ is acyclic?}
\label{problem:fvs}

\defproblem{\phantomsection\boundedDegreeDeletion{}}{Graph $G$, integers $k$,$t$}{Is there a set $S \subseteq V(G)$ of size at most $k$ such that the maximum degree of $G - S$ is at most $t$?}
\label{problem:bounded_degree_deletion}

In the \fdeletion{} problem $\mathcal{F}$ is a (usually finite) set of graphs.

\defproblem{\phantomsection\fdeletion{}}{Graph $G$, integer $k$}{Is there a set $S \subseteq V(G)$ of size at most $k$ such that $G - S$ has no graph of $\mathcal{F}$ as minor?}
\label{problem:fdeletion} 
\bibliography{bib}

@article{holsEliminationDistancesBlocking2022,
  title     = {Elimination Distances, Blocking Sets, and Kernels for Vertex Cover},
  author    = {Eva{-}Maria C. Hols and Stefan Kratsch and Astrid Pieterse},
  year      = {2022},
  journal   = {{SIAM} J. Discret. Math.},
  volume    = {36},
  number    = {3},
  pages     = {1955--1990},
  doi       = {10.1137/20m1335285},
  bibsource = {dblp computer science bibliography, https://dblp.org},
  url       = {https://doi.org/10.1137/20m1335285}
}

@article{bougeretBridgedepthCharacterizesWhich2022,
  title     = {Bridge-Depth Characterizes which Minor-Closed Structural Parameterizations of Vertex Cover Admit a Polynomial Kernel},
  author    = {Marin Bougeret and Bart M. P. Jansen and Ignasi Sau},
  year      = {2022},
  journal   = {{SIAM} J. Discret. Math.},
  volume    = {36},
  number    = {4},
  pages     = {2737--2773},
  doi       = {10.1137/21m1400766},
  bibsource = {dblp computer science bibliography, https://dblp.org},
  url       = {https://doi.org/10.1137/21m1400766}
}

@article{jansenVertexCoverKernelization2013b,
  title     = {Vertex Cover Kernelization Revisited - Upper and Lower Bounds for a Refined Parameter},
  author    = {Bart M. P. Jansen and Hans L. Bodlaender},
  year      = {2013},
  journal   = {Theory Comput. Syst.},
  volume    = {53},
  number    = {2},
  pages     = {263--299},
  doi       = {10.1007/s00224-012-9393-4},
  bibsource = {dblp computer science bibliography, https://dblp.org},
  url       = {https://doi.org/10.1007/s00224-012-9393-4}
}

@article{bougeretHowMuchDoes2019,
  title     = {How Much Does a Treedepth Modulator Help to Obtain Polynomial Kernels Beyond Sparse Graphs?},
  author    = {Marin Bougeret and Ignasi Sau},
  year      = {2019},
  journal   = {Algorithmica},
  volume    = {81},
  number    = {10},
  pages     = {4043--4068},
  doi       = {10.1007/s00453-018-0468-8},
  bibsource = {dblp computer science bibliography, https://dblp.org},
  url       = {https://doi.org/10.1007/s00453-018-0468-8}
}

@inproceedings{greilhuberComponentOrderConnectivity2024,
  title     = {Component Order Connectivity Admits No Polynomial Kernel Parameterized by the Distance to Subdivided Comb Graphs},
  booktitle = {19th International Symposium on Parameterized and Exact Computation, {IPEC} 2024, September 4-6, 2024, Royal Holloway, University of London, Egham, United Kingdom},
  author    = {Jakob Greilhuber and Roohani Sharma},
  editor    = {{\'{E}}douard Bonnet and Pawel Rzazewski},
  year      = {2024},
  series    = {LIPIcs},
  volume    = {321},
  pages     = {21:1--21:17},
  publisher = {Schloss Dagstuhl - Leibniz-Zentrum f{\"{u}}r Informatik},
  doi       = {10.4230/LIPIcs.IPEC.2024.21},
  url       = {https://doi.org/10.4230/LIPIcs.IPEC.2024.21},
  urldate   = {2025-03-18},
  bibsource = {dblp computer science bibliography, https://dblp.org}
}

@book{cyganParameterizedAlgorithms2015,
  title     = {Parameterized Algorithms},
  author    = {Marek Cygan and Fedor V. Fomin and Lukasz Kowalik and Daniel Lokshtanov and D{\'{a}}niel Marx and Marcin Pilipczuk and Michal Pilipczuk and Saket Saurabh},
  year      = {2015},
  publisher = {Springer},
  doi       = {10.1007/978-3-319-21275-3},
  isbn      = {978-3-319-21274-6},
  bibsource = {dblp computer science bibliography, https://dblp.org},
  url       = {https://doi.org/10.1007/978-3-319-21275-3}
}

@inproceedings{bhyravarapuDifferenceDeterminesDegree2023,
  title     = {Difference Determines the Degree: Structural Kernelizations of Component Order Connectivity},
  booktitle = {18th International Symposium on Parameterized and Exact Computation, {IPEC} 2023, September 6-8, 2023, Amsterdam, The Netherlands},
  author    = {Sriram Bhyravarapu and Satyabrata Jana and Saket Saurabh and Roohani Sharma},
  editor    = {Neeldhara Misra and Magnus Wahlstr{\"{o}}m},
  year      = {2023},
  series    = {LIPIcs},
  volume    = {285},
  pages     = {5:1--5:14},
  publisher = {Schloss Dagstuhl - Leibniz-Zentrum f{\"{u}}r Informatik},
  doi       = {10.4230/LIPIcs.IPEC.2023.5},
  bibsource = {dblp computer science bibliography, https://dblp.org},
  url       = {https://doi.org/10.4230/LIPIcs.IPEC.2023.5}
}

@article{majumdarPolynomialKernelsVertex2018,
  title     = {Polynomial Kernels for Vertex Cover Parameterized by Small Degree Modulators},
  author    = {Diptapriyo Majumdar and Venkatesh Raman and Saket Saurabh},
  year      = {2018},
  journal   = {Theory Comput. Syst.},
  volume    = {62},
  number    = {8},
  pages     = {1910--1951},
  doi       = {10.1007/s00224-018-9858-1},
  url       = {https://doi.org/10.1007/s00224-018-9858-1},
  bibsource = {dblp computer science bibliography, https://dblp.org}
}

@article{gareyRectilinearSteinerTree1977,
  title     = {The Rectilinear Steiner Tree Problem is {NP} Complete},
  author    = {M. R. Garey and David S. Johnson},
  year      = {1977},
  journal   = {{SIAM} Journal of Applied Mathematics},
  volume    = {32},
  pages     = {826--834},
  bibsource = {dblp computer science bibliography, https://dblp.org}
}

@article{tuVertexCoverP_32013,
  title     = {The Vertex Cover {$P_3$} Problem in Cubic Graphs},
  author    = {Jianhua Tu and Fengmei Yang},
  year      = {2013},
  volume    = {113},
  number    = {13},
  pages     = {481--485},
  doi       = {10.1016/j.ipl.2013.04.002},
  url       = {https://doi.org/10.1016/j.ipl.2013.04.002},
  bibsource = {dblp computer science bibliography, https://dblp.org},
  journal   = {Inf. Process. Lett.}
}

@article{boliacComputingDissociationNumber2004,
  title     = {On computing the dissociation number and the induced matching number of bipartite graphs},
  author    = {Rodica Boliac and Kathie Cameron and Vadim V. Lozin},
  year      = {2004},
  journal   = {Ars Comb.},
  volume    = {72},
  bibsource = {dblp computer science bibliography, https://dblp.org}
}

@inproceedings{fominVertexCoverStructural2016,
  title     = {Vertex Cover Structural Parameterization Revisited},
  booktitle = {Graph-Theoretic Concepts in Computer Science - 42nd International Workshop, {WG} 2016, Istanbul, Turkey, June 22-24, 2016, Revised Selected Papers},
  author    = {Fedor V. Fomin and Torstein J. F. Str{\o}mme},
  editor    = {Pinar Heggernes},
  year      = {2016},
  series    = {Lecture Notes in Computer Science},
  volume    = {9941},
  pages     = {171--182},
  doi       = {10.1007/978-3-662-53536-3_15},
  url       = {https://doi.org/10.1007/978-3-662-53536-3\_15},
  bibsource = {dblp computer science bibliography, https://dblp.org}
}

@inproceedings{karpReducibilityCombinatorialProblems1972,
  title     = {Reducibility Among Combinatorial Problems},
  booktitle = {Proceedings of a symposium on the Complexity of Computer Computations, held March 20-22, 1972, at the {IBM} Thomas J. Watson Research Center, Yorktown Heights, New York, {USA}},
  author    = {Richard M. Karp},
  editor    = {Raymond E. Miller and James W. Thatcher},
  year      = {1972},
  series    = {The {IBM} Research Symposia Series},
  pages     = {85--103},
  publisher = {Plenum Press, New York},
  doi       = {10.1007/978-1-4684-2001-2_9},
  url       = {https://doi.org/10.1007/978-1-4684-2001-2\_9},
  bibsource = {dblp computer science bibliography, https://dblp.org}
}

@article{donkersTuringKernelizationDichotomy2021,
  title     = {A Turing kernelization dichotomy for structural parameterizations of F-Minor-Free Deletion},
  author    = {Huib Donkers and Bart M. P. Jansen},
  year      = {2021},
  journal   = {J. Comput. Syst. Sci.},
  volume    = {119},
  pages     = {164--182},
  doi       = {10.1016/j.jcss.2021.02.005},
  url       = {https://doi.org/10.1016/j.jcss.2021.02.005},
  bibsource = {dblp computer science bibliography, https://dblp.org}
}

@article{lewisNodeDeletionProblemHereditary1980,
  title     = {The Node-Deletion Problem for Hereditary Properties is NP-Complete},
  author    = {John M. Lewis and Mihalis Yannakakis},
  year      = {1980},
  journal   = {J. Comput. Syst. Sci.},
  volume    = {20},
  number    = {2},
  pages     = {219--230},
  doi       = {10.1016/0022-0000(80)90060-4},
  bibsource = {dblp computer science bibliography, https://dblp.org},
  url       = {https://doi.org/10.1016/0022-0000(80)90060-4}
}

@inproceedings{holsSmallerParametersVertex2017,
  title     = {Smaller Parameters for Vertex Cover Kernelization},
  booktitle = {12th International Symposium on Parameterized and Exact Computation, {IPEC} 2017, September 6-8, 2017, Vienna, Austria},
  author    = {Eva{-}Maria C. Hols and Stefan Kratsch},
  editor    = {Daniel Lokshtanov and Naomi Nishimura},
  year      = {2017},
  series    = {LIPIcs},
  volume    = {89},
  pages     = {20:1--20:12},
  publisher = {Schloss Dagstuhl - Leibniz-Zentrum f{\"{u}}r Informatik},
  doi       = {10.4230/LIPIcs.IPEC.2017.20},
  url       = {https://doi.org/10.4230/LIPIcs.IPEC.2017.20},
  bibsource = {dblp computer science bibliography, https://dblp.org}
}

@article{dekkerKernelizationFeedbackVertex2024,
  author    = {David Dekker and
               Bart M. P. Jansen},
  title     = {Kernelization for feedback vertex set via elimination distance to
               a forest},
  journal   = {Discret. Appl. Math.},
  volume    = {346},
  pages     = {192--214},
  year      = {2024},
  url       = {https://doi.org/10.1016/j.dam.2023.12.016},
  doi       = {10.1016/J.DAM.2023.12.016},
  bibsource = {dblp computer science bibliography, https://dblp.org}
}

@article{grossSurveyComponentOrder2013,
  title     = {A Survey of Component Order Connectivity Models of Graph Theoretic Networks},
  author    = {Gross, Daniel and Heinig, Monika and Iswara, Lakshmi and Kazmierczak, L William and Luttrell, Kristi and Saccoman, John T and Suffel, Charles},
  year      = {2013},
  journal   = {WSEAS Transactions on Mathematics},
  volume    = {12},
  number    = {9},
  pages     = {895--910},
  publisher = {WSEAS}
}

@inproceedings{bougeretKernelizationDichotomiesHitting2024a,
  title     = {Kernelization Dichotomies for Hitting Subgraphs Under Structural Parameterizations},
  booktitle = {51st International Colloquium on Automata, Languages, and Programming, {ICALP} 2024, July 8-12, 2024, Tallinn, Estonia},
  author    = {Marin Bougeret and Bart M. P. Jansen and Ignasi Sau},
  editor    = {Karl Bringmann and Martin Grohe and Gabriele Puppis and Ola Svensson},
  year      = {2024},
  series    = {LIPIcs},
  volume    = {297},
  pages     = {33:1--33:20},
  publisher = {Schloss Dagstuhl - Leibniz-Zentrum f{\"{u}}r Informatik},
  doi       = {10.4230/LIPIcs.ICALP.2024.33},
  url       = {https://doi.org/10.4230/LIPIcs.ICALP.2024.33},
  urldate   = {2025-04-02},
  bibsource = {dblp computer science bibliography, https://dblp.org}
}

@inproceedings{arnborgMonadicSecondOrder1990,
  title     = {Monadic Second Order Logic, Tree Automata and Forbidden Minors},
  booktitle = {Computer Science Logic, 4th Workshop, {CSL} '90, Heidelberg, Germany, October 1-5, 1990, Proceedings},
  author    = {Stefan Arnborg and Andrzej Proskurowski and Detlef Seese},
  editor    = {Egon B{\"{o}}rger and Hans Kleine B{\"{u}}ning and Michael M. Richter and Wolfgang Sch{\"{o}}nfeld},
  year      = {1990},
  series    = {Lecture Notes in Computer Science},
  volume    = {533},
  pages     = {1--16},
  publisher = {Springer},
  doi       = {10.1007/3-540-54487-9_49},
  url       = {https://doi.org/10.1007/3-540-54487-9\_49},
  urldate   = {2025-04-09},
  bibsource = {dblp computer science bibliography, https://dblp.org}
}

@article{DBLP:journals/mst/CyganLPPS14,
  author    = {Marek Cygan and Daniel Lokshtanov and Marcin Pilipczuk and Michal Pilipczuk and Saket Saurabh},
  title     = {On the Hardness of Losing Width},
  journal   = {Theory Comput. Syst.},
  volume    = {54},
  number    = {1},
  pages     = {73--82},
  year      = {2014},
  url       = {https://doi.org/10.1007/s00224-013-9480-1},
  doi       = {10.1007/s00224-013-9480-1},
  bibsource = {dblp computer science bibliography, https://dblp.org}
}

@article{DBLP:journals/tcs/JansenP20,
  author    = {Bart M. P. Jansen and Astrid Pieterse},
  title     = {Polynomial kernels for hitting forbidden minors under structural parameterizations},
  journal   = {Theor. Comput. Sci.},
  volume    = {841},
  pages     = {124--166},
  year      = {2020},
  url       = {https://doi.org/10.1016/j.tcs.2020.07.009},
  doi       = {10.1016/j.tcs.2020.07.009},
  bibsource = {dblp computer science bibliography, https://dblp.org}
}

@article{DBLP:journals/siamcomp/BussG93,
  author    = {Jonathan F. Buss and Judy Goldsmith},
  title     = {Nondeterminism Within {P}},
  journal   = {{SIAM} J. Comput.},
  volume    = {22},
  number    = {3},
  pages     = {560--572},
  year      = {1993},
  url       = {https://doi.org/10.1137/0222038},
  doi       = {10.1137/0222038},
  bibsource = {dblp computer science bibliography, https://dblp.org}
}

@inproceedings{DBLP:conf/wg/ChorFJ04,
  author    = {Benny Chor and Mike Fellows and David W. Juedes},
  editor    = {Juraj Hromkovic and Manfred Nagl and Bernhard Westfechtel},
  title     = {Linear Kernels in Linear Time, or How to Save k Colors in O(n\({}^{\mbox{2}}\)) Steps},
  booktitle = {Graph-Theoretic Concepts in Computer Science, 30th International Workshop,WG 2004, Bad Honnef, Germany, June 21-23, 2004, Revised Papers},
  series    = {Lecture Notes in Computer Science},
  volume    = {3353},
  pages     = {257--269},
  publisher = {Springer},
  year      = {2004},
  url       = {https://doi.org/10.1007/978-3-540-30559-0\_22},
  doi       = {10.1007/978-3-540-30559-0_22},
  bibsource = {dblp computer science bibliography, https://dblp.org}
}

@inproceedings{DBLP:conf/wg/Fellows03,
  author    = {Michael R. Fellows},
  editor    = {Hans L. Bodlaender},
  title     = {Blow-Ups, Win/Win's, and Crown Rules: Some New Directions in {FPT}},
  booktitle = {Graph-Theoretic Concepts in Computer Science, 29th International Workshop, {WG} 2003, Elspeet, The Netherlands, June 19-21, 2003, Revised Papers},
  series    = {Lecture Notes in Computer Science},
  volume    = {2880},
  pages     = {1--12},
  publisher = {Springer},
  year      = {2003},
  url       = {https://doi.org/10.1007/978-3-540-39890-5\_1},
  doi       = {10.1007/978-3-540-39890-5_1},
  bibsource = {dblp computer science bibliography, https://dblp.org}
}

@inproceedings{DBLP:conf/iwpec/KumarL16,
  author    = {Mithilesh Kumar and Daniel Lokshtanov},
  editor    = {Jiong Guo and Danny Hermelin},
  title     = {A 2lk Kernel for l-Component Order Connectivity},
  booktitle = {11th International Symposium on Parameterized and Exact Computation, {IPEC} 2016, August 24-26, 2016, Aarhus, Denmark},
  series    = {LIPIcs},
  volume    = {63},
  pages     = {20:1--20:14},
  publisher = {Schloss Dagstuhl - Leibniz-Zentrum f{\"{u}}r Informatik},
  year      = {2016},
  url       = {https://doi.org/10.4230/LIPIcs.IPEC.2016.20},
  doi       = {10.4230/LIPIcs.IPEC.2016.20},
  bibsource = {dblp computer science bibliography, https://dblp.org}
}

@article{DBLP:journals/jal/ChenKJ01,
  author    = {Jianer Chen and Iyad A. Kanj and Weijia Jia},
  title     = {Vertex Cover: Further Observations and Further Improvements},
  journal   = {J. Algorithms},
  volume    = {41},
  number    = {2},
  pages     = {280--301},
  year      = {2001},
  url       = {https://doi.org/10.1006/jagm.2001.1186},
  doi       = {10.1006/jagm.2001.1186},
  bibsource = {dblp computer science bibliography, https://dblp.org}
}

@article{DBLP:journals/sigact/Khuller02,
  author    = {Samir Khuller},
  title     = {Algorithms column: the vertex cover problem},
  journal   = {{SIGACT} News},
  volume    = {33},
  number    = {2},
  pages     = {31--33},
  year      = {2002},
  url       = {https://doi.org/10.1145/564585.564598},
  doi       = {10.1145/564585.564598},
  bibsource = {dblp computer science bibliography, https://dblp.org}
}

@article{drangeComputationalComplexityVertex2016,
  title     = {On the Computational Complexity of Vertex Integrity and Component Order Connectivity},
  author    = {P{\aa}l Gr{\o}n{\aa}s Drange and Markus S. Dregi and Pim van 't Hof},
  year      = {2016},
  journal   = {Algorithmica},
  volume    = {76},
  number    = {4},
  pages     = {1181--1202},
  doi       = {10.1007/s00453-016-0127-x},
  bibsource = {dblp computer science bibliography, https://dblp.org},
  url       = {https://doi.org/10.1007/s00453-016-0127-x}
}

@book{fominKernelizationTheoryParameterized2019,
  title     = {Kernelization: theory of parameterized preprocessing},
  author    = {Fomin, Fedor V. and Lokshtanov, Daniel and Saurabh, Saket and Zehavi, Meirav},
  year      = {2019},
  publisher = {Cambridge University Press},
  doi       = {10.1017/9781107415157}
}

@inproceedings{caselCombiningCrownStructures2024,
  title     = {Combining Crown Structures for Vulnerability Measures},
  booktitle = {19th International Symposium on Parameterized and Exact Computation, {IPEC} 2024, September 4-6, 2024, Royal Holloway, University of London, Egham, United Kingdom},
  author    = {Katrin Casel and Tobias Friedrich and Aikaterini Niklanovits and Kirill Simonov and Ziena Zeif},
  editor    = {{\'{E}}douard Bonnet and Pawel Rzazewski},
  year      = {2024},
  series    = {LIPIcs},
  volume    = {321},
  pages     = {1:1--1:15},
  publisher = {Schloss Dagstuhl - Leibniz-Zentrum f{\"{u}}r Informatik},
  doi       = {10.4230/LIPIcs.IPEC.2024.1},
  url       = {https://doi.org/10.4230/LIPIcs.IPEC.2024.1},
  urldate   = {2025-04-09},
  bibsource = {dblp computer science bibliography, https://dblp.org}
}

@inproceedings{DBLP:conf/esa/Casel0INZ21,
  title     = {Balanced Crown Decomposition for Connectivity Constraints},
  booktitle = {29th Annual European Symposium on Algorithms, {ESA} 2021, September 6-8, 2021, Lisbon, Portugal (Virtual Conference)},
  author    = {Katrin Casel and Tobias Friedrich and Davis Issac and Aikaterini Niklanovits and Ziena Zeif},
  editor    = {Petra Mutzel and Rasmus Pagh and Grzegorz Herman},
  year      = {2021},
  series    = {LIPIcs},
  volume    = {204},
  pages     = {26:1--26:15},
  publisher = {Schloss Dagstuhl - Leibniz-Zentrum f{\"{u}}r Informatik},
  doi       = {10.4230/LIPIcs.ESA.2021.26},
  url       = {https://doi.org/10.4230/LIPIcs.ESA.2021.26},
  bibsource = {dblp computer science bibliography, https://dblp.org}
}

@article{DBLP:journals/algorithmica/BaguleyFNNPZ25,
  title     = {Fixed Parameter Multi-Objective Evolutionary Algorithms for the W-Separator Problem},
  author    = {Samuel Baguley and Tobias Friedrich and Aneta Neumann and Frank Neumann and Marcus Pappik and Ziena Zeif},
  year      = {2025},
  journal   = {Algorithmica},
  volume    = {87},
  number    = {4},
  pages     = {537--571},
  doi       = {10.1007/s00453-024-01290-9},
  url       = {https://doi.org/10.1007/s00453-024-01290-9},
  bibsource = {dblp computer science bibliography, https://dblp.org}
}

@article{soleimanfallahKernelOrder2kc2011,
  author    = {Arezou Soleimanfallah and Anders Yeo},
  title     = {A kernel of order 2k-c for Vertex Cover},
  journal   = {Discret. Math.},
  volume    = {311},
  number    = {10-11},
  pages     = {892--895},
  year      = {2011},
  url       = {https://doi.org/10.1016/j.disc.2011.02.014},
  doi       = {10.1016/j.disc.2011.02.014},
  bibsource = {dblp computer science bibliography, https://dblp.org}
}

@article{lampisKernelOrder22011,
  author    = {Michael Lampis},
  title     = {A kernel of order 2 k-c log k for vertex cover},
  journal   = {Inf. Process. Lett.},
  volume    = {111},
  number    = {23-24},
  pages     = {1089--1091},
  year      = {2011},
  url       = {https://doi.org/10.1016/j.ipl.2011.09.003},
  doi       = {10.1016/J.IPL.2011.09.003},
  bibsource = {dblp computer science bibliography, https://dblp.org}
}

@article{xiaoLinearKernelsSeparating2017a,
  title     = {Linear kernels for separating a graph into components of bounded size},
  author    = {Mingyu Xiao},
  year      = {2017},
  journal   = {J. Comput. Syst. Sci.},
  volume    = {88},
  pages     = {260--270},
  doi       = {10.1016/j.jcss.2017.04.004},
  url       = {https://doi.org/10.1016/j.jcss.2017.04.004},
  bibsource = {dblp computer science bibliography, https://dblp.org}
}

@phdthesis{kazmierczak2003relationship,
  title  = {On the relationship between connectivity and component order connectivity},
  author = {Kazmierczak, Lawrence William},
  year   = {2003},
  school = {Stevens Institute of Technology},
  url    = {https://dl.acm.org/doi/10.5555/959796}
}

@article{ganianStructuralParameterizationsBoundedDegree2021,
  title    = {On Structural Parameterizations of the Bounded-Degree Vertex Deletion Problem},
  author   = {Ganian, Robert and Klute, Fabian and Ordyniak, Sebastian},
  year     = {2021},
  month    = jan,
  journal  = {Algorithmica},
  volume   = {83},
  number   = {1},
  pages    = {297--336},
  issn     = {1432-0541},
  doi      = {10.1007/s00453-020-00758-8},
  url      = {https://doi.org/10.1007/s00453-020-00758-8},
  urldate  = {2023-07-12},
  langid   = {english}
}

@inproceedings{dellKernelizationPackingProblems2012,
  author    = {Holger Dell and
               D{\'{a}}niel Marx},
  editor    = {Yuval Rabani},
  title     = {Kernelization of packing problems},
  booktitle = {Proceedings of the Twenty-Third Annual {ACM-SIAM} Symposium on Discrete
               Algorithms, {SODA} 2012, Kyoto, Japan, January 17-19, 2012},
  pages     = {68--81},
  publisher = {{SIAM}},
  year      = {2012},
  url       = {https://doi.org/10.1137/1.9781611973099.6},
  doi       = {10.1137/1.9781611973099.6},
  bibsource = {dblp computer science bibliography, https://dblp.org}
}

@article{giannopoulouUniformKernelizationComplexity2017,
  author    = {Archontia C. Giannopoulou and
               Bart M. P. Jansen and
               Daniel Lokshtanov and
               Saket Saurabh},
  title     = {Uniform Kernelization Complexity of Hitting Forbidden Minors},
  journal   = {{ACM} Trans. Algorithms},
  volume    = {13},
  number    = {3},
  pages     = {35:1--35:35},
  year      = {2017},
  url       = {https://doi.org/10.1145/3029051},
  doi       = {10.1145/3029051},
  bibsource = {dblp computer science bibliography, https://dblp.org}
}

@article{cockayneMatchingsTransversalsHypergraphs1979,
  author    = {Ernest J. Cockayne and
               Stephen T. Hedetniemi and
               Peter J. Slater},
  title     = {Matchings and transversals in hypergraphs, domination and independence-in
               trees},
  journal   = {J. Comb. Theory {B}},
  volume    = {26},
  number    = {1},
  pages     = {78--80},
  year      = {1979},
  url       = {https://doi.org/10.1016/0095-8956(79)90044-3},
  doi       = {10.1016/0095-8956(79)90044-3},
  bibsource = {dblp computer science bibliography, https://dblp.org}
}

@article{howieProductsIdempotentOrderpreserving1973,
  title   = {Products of Idempotent Order-Preserving Transformations},
  author  = {Howie, J. M. and Schein, B. M.},
  year    = {1973},
  journal = {Journal of the London Mathematical Society},
  volume  = {s2-7},
  number  = {2},
  pages   = {357--366},
  issn    = {0024-6107},
  doi     = {10.1112/jlms/s2-7.2.357},
  url     = {https://doi.org/10.1112/jlms/s2-7.2.357}
}

@article{howieProductsIdempotentsCertain1971,
  title   = {Products of Idempotents in Certain Semigroups of Transformations},
  author  = {Howie, J. M.},
  year    = {1971},
  journal = {Proceedings of the Edinburgh Mathematical Society},
  volume  = {17},
  number  = {3},
  pages   = {223--236},
  doi     = {10.1017/S0013091500026936}
}

@article{ahmadubellouniversityQuasiidempotentsFiniteSemigroup2023,
  title   = {Quasi-idempotents in finite semigroup of full order-preserving transformations},
  author  = {Imam, A. T. and Ibrahim, S. and Garba, G. U. and Usman, L. and Idris, A.},
  journal = {Algebra and Discrete Mathematics},
  pages   = {62--72},
  volume  = {35},
  number  = {1},
  year    = {2023},
  doi     = {10.12958/adm1846}
}

@article{higginsCombinatorialResultsSemigroups1993,
  title   = {Combinatorial Results for Semigroups of Order-Preserving Mappings},
  author  = {Higgins, Peter M.},
  year    = {1993},
  journal = {Mathematical Proceedings of the Cambridge Philosophical Society},
  volume  = {113},
  number  = {2},
  pages   = {281--296},
  doi     = {10.1017/S0305004100075964}
}

@article{higginsIdempotentDepthSemigroups1994,
  title   = {Idempotent Depth in Semigroups of Order-Preserving Mappings},
  author  = {Higgins, Peter M.},
  year    = {1994},
  journal = {Proceedings of the Royal Society of Edinburgh: Section A Mathematics},
  volume  = {124},
  number  = {5},
  pages   = {1045--1058},
  doi     = {10.1017/S0308210500022502}
}

@article{gomesRanksCertainSemigroups1992,
  title   = {On the Ranks of Certain Semigroups of Order-Preserving Transformations},
  author  = {Gomes, Gracinda M. S. and Howie, John M.},
  year    = {1992},
  journal = {Semigroup Forum},
  volume  = {45},
  number  = {1},
  pages   = {272--282},
  issn    = {1432-2137},
  doi     = {10.1007/BF03025769},
  url     = {https://doi.org/10.1007/BF03025769}
}

@article{kudryavtsevaPartializationCategoriesInverse2008,
  title     = {Partialization of Categories and Inverse Braid-Permutation Monoids},
  author    = {Kudryavtseva, Ganna and Mazorchuk, Volodymyr},
  year      = {2008},
  journal   = {International Journal of Algebra and Computation},
  volume    = {18},
  number    = {6},
  pages     = {989--1017},
  doi       = {10.1142/S0218196708004731},
  url       = {https://doi.org/10.1142/S0218196708004731},
  bibsource = {dblp computer science bibliography, https://dblp.org}
}

@inproceedings{bougeretKernelizationDichotomiesHitting2025,
  author    = {Marin Bougeret and
               Eric Brandwein and
               Ignasi Sau},
  editor    = {Meena Mahajan and
               Florin Manea and
               Annabelle McIver and
               Kim Thang Nguyen},
  title     = {Kernelization Dichotomies for Hitting Minors Under Structural Parameterizations},
  booktitle = {43rd International Symposium on Theoretical Aspects of Computer Science,
               {STACS} 2026, Grenoble, France, March 9-13, 2026},
  series    = {LIPIcs},
  volume    = {364},
  pages     = {17:1--17:19},
  publisher = {Schloss Dagstuhl - Leibniz-Zentrum f{\"{u}}r Informatik},
  year      = {2026},
  url       = {https://doi.org/10.4230/LIPIcs.STACS.2026.17},
  doi       = {10.4230/LIPICS.STACS.2026.17},
  bibsource = {dblp computer science bibliography, https://dblp.org}
}

@inproceedings{guptaLosingTreewidthSeparating2019,
  author    = {Anupam Gupta and
               Euiwoong Lee and
               Jason Li and
               Pasin Manurangsi and
               Michal Wlodarczyk},
  editor    = {Timothy M. Chan},
  title     = {Losing Treewidth by Separating Subsets},
  booktitle = {Proceedings of the Thirtieth Annual {ACM-SIAM} Symposium on Discrete
               Algorithms, {SODA} 2019, San Diego, California, USA, January 6-9,
               2019},
  pages     = {1731--1749},
  publisher = {{SIAM}},
  year      = {2019},
  url       = {https://doi.org/10.1137/1.9781611975482.104},
  doi       = {10.1137/1.9781611975482.104},
  bibsource = {dblp computer science bibliography, https://dblp.org}
}

@inproceedings{fominPlanarFdeletionApproximation2012,
  title     = {Planar F-Deletion: Approximation, Kernelization and Optimal FPT Algorithms},
  booktitle = {53rd Annual IEEE Symposium on Foundations of Computer Science, FOCS 2012, New Brunswick, NJ, USA, October 20-23, 2012},
  author    = {Fomin, Fedor V. and Lokshtanov, Daniel and Misra, Neeldhara and Saurabh, Saket},
  year      = 2012,
  pages     = {470--479},
  publisher = {IEEE Computer Society},
  doi       = {10.1109/FOCS.2012.62},
  url       = {https://doi.org/10.1109/FOCS.2012.62},
  bibsource = {dblp computer science bibliography, https://dblp.org}
}

\end{document}